\documentclass[manuscript]{acmart}
\input{packages.sty}

\newcommand{\FPT}{{\sf FPT}}
\newcommand{\XP}{{\sf XP}}

\let\openbox\relax
\newcommand*{\openbox}{\leavevmode
\hbox to.77778em{%
\hfil\vrule
\vbox to.675em{\hrule width.6em\vfil\hrule}%
\vrule\hfil}}

\newcommand{\problem}[3]{
        \begin{center}
                \begin{boxedminipage}{\textwidth}
                        \textsc{{#1}}\\[2pt]
                        \begin{tabular}{ r p{0.85\textwidth}}
                                \textit{Instance:} & {#2}\\
                                \textit{Question:} & {#3}
                        \end{tabular}
                \end{boxedminipage}
        \end{center}
}

\usepackage{regexpatch}
\makeatletter
\xpatchcmd\thmt@restatable{%
\csname #2\@xa\endcsname\ifx\@nx#1\@nx\else[{#1}]\fi
}{%
\ifthmt@thisistheone\csname #2\@xa\endcsname\ifx\@nx#1\@nx\else[{#1}]\fi
\else
\csname #2\@xa\endcsname[{Restated}]
\fi}{}{}
\makeatother

\newenvironment{claimproof}[1][\proofname]
{\proof[#1]}
{\endproof}
\newtheorem{theorem}{Theorem}[section]
\newtheorem{lemma}[theorem]{Lemma}
\newtheorem{corollary}[theorem]{Corollary}
\newtheorem{claim}[theorem]{Claim}
\newtheorem{proposition}[theorem]{Proposition}
\newtheorem{definition}[theorem]{Definition}

\newtheorem{remark}[theorem]{Remark}

\newcommand{\problemMWC}{\textsc{Multiway Cut}\xspace}

\newcommand{\problemPMWC}{\textsc{Planar Multiway Cut}\xspace}
\newcommand{\problemEMWC}{\textsc{Edge Multiway Cut}\xspace}
\newcommand{\problemNMWC}{\textsc{Node Multiway Cut}\xspace}

\newcommand{\ie}{i.e.\xspace}
\newcommand{\eg}{e.g.\xspace}
\newcommand{\etal}{\emph{et al.}\xspace}
\newcommand{\WLOG}{Without loss of generality, \xspace}

\newcommand{\nph}{\textsc{NP}-hard\xspace}

\newcommand{\ptas}{\textsc{PTAS}\xspace}
\newcommand{\defi}[1]{\emph{#1}\xspace}        
\newcommand{\OO}{{O}}
\newcommand{\Oh}[1]{${O}(#1)$}
\newcommand{\Ostar}[1]{${O}^*(#1)$}
\newcommand{\one}{\alpha\xspace}
\newcommand{\two}{\beta\xspace}
\newcommand{\three}{\gamma\xspace}

\newcommand{\mst}{minimum Steiner tree\xspace }
\newcommand{\msts}{minimum Steiner trees\xspace}

\newcommand{\terminal}[2]{^{#1} t_{#2}\xspace}
\newcommand{\augvertex}[2]{^{#1} t^+_{#2}\xspace}
\newcommand{\augset}[1]{T^+_{#1}\xspace}

\newcommand{\bd}[1]{\ensuremath{\partial #1}}
\newcommand{\nooseone}{\ensuremath{\vec{\alpha}}}
\newcommand{\noosetwo}{\ensuremath{\vec{\beta}}}
\newcommand{\noose}{\ensuremath{\vec{\gamma}}}
\newcommand{\enc}{\ensuremath{\mathbf{enc}}}
\newcommand{\exc}{\ensuremath{\mathbf{exc}}}
\newcommand{\midset}{\ensuremath{\mathbf{mid}}}

\def\colin{Colin de Verdi\`{e}re\xspace}

\DeclareMathDelimiter{(}{\mathopen} {operators}{"28}{largesymbols}{"00}
\DeclareMathDelimiter{)}{\mathclose}{operators}{"29}{largesymbols}{"01}

\begin{document}

\title{Planar Multiway Cut with Terminals on Few Faces}

\author{Sukanya Pandey}
\email{pandey@algo.rwth-aachen.de}
\orcid{0000-0001-5728-1120}
\affiliation{%
  \institution{Dept.\ Computer Science, RWTH Aachen}
  \city{Aachen}
  \country{Germany}
}

\author{Erik Jan van Leeuwen}
\email{e.j.vanleeuwen@uu.nl}
\orcid{0000-0001-5240-7257}
\affiliation{%
  \institution{Dept.\ Information and Computing Sciences, Utrecht University}
  \city{Utrecht}
  \country{The Netherlands}
}

\newcommand{\rev}[1]{{#1}}

\begin{abstract}
    We consider the {\problemEMWC} problem on planar graphs. It is known that this can be solved in $n^{\OO(\sqrt{t})}$ time [Klein, Marx, ICALP 2012] and not in $n^{\OO(\sqrt{t})}$ time under the Exponential Time Hypothesis [Marx, ICALP 2012], where $t$ is the number of terminals. \rev{A stronger parameter} is the number $k$ of faces of the planar graph that jointly cover all terminals. For the related {\sc Steiner Tree} problem, an $n^{\OO(\sqrt{k})}$ time algorithm was recently shown [Kisfaludi-Bak \etal, SODA 2019]. By a completely different approach, we prove in this paper that {\problemEMWC} can be solved in $n^{\OO(\sqrt{k})}$ time as well.
    Our approach employs several major concepts on planar graphs, including homotopy and sphere-cut decomposition. We also mix a global treewidth dynamic program with a Dreyfus-Wagner style dynamic program to locally deal with large numbers of terminals.   
\end{abstract}



\maketitle

\section{Introduction}
A graph with weighted edges, and a subset of its vertices called \emph{terminals}, is given. How would you pairwise disconnect the terminals by removing a subset of minimum possible weight of the edges of the graph? Widely known as the {\sc (Edge) Multiway Cut} problem, this question is a natural generalization of the popular minimum $(s,t)$-cut problem. A variant of the problem was first introduced by T.C.\ Hu in 1969 \cite{hu1969}. Formally, the {\sc Edge Multiway Cut} problem is follows:
\smallskip

\problem{\problemEMWC}{A graph $G$, a terminal set $T \subseteq V(G)$, and a weight function $\omega: E(G) \rightarrow \mathbb{Q}^{+}$.}{What is the minimum possible weight of the cut that pairwise separates the vertices in $T$?}

The study of the complexity of \problemEMWC was pioneered by Dahlhaus \etal in their seminal paper \cite{DJPSY94}. The authors showed that \problemEMWC is {\nph} on arbitrary graphs for any fixed $t \geq 3$, where $t$ is the number of terminals. They also showed that the problem is \rev{APX-hard} for all fixed $t \geq 3$ even when all weights are unit, while giving an \Oh{tnm \log(n^2/m)}-time approximation algorithm, which given an arbitrary graph and an arbitrary $t$ finds a solution within a $2-2/t$ ratio of the optimum. Following these hardness results, began a twofold quest: to find approximation algorithms with a better approximation ratio~ \cite{CKR00, Karger_approx} and to find exact algorithms that were more efficient than any brute-force approach \cite{Xiao, CCF14, CLL09, GUILLEMOT_NMWC, CPPW13}. 
Starting with the seminal work of Marx~\cite{tcs/Marx06}, significant research effort was also put in understanding the parameterized complexity of {\problemEMWC} for other parameters, such as the size of the solution~\cite{Xiao, CCF14, CLL09, GUILLEMOT_NMWC, CPPW13}.

\paragraph{Restriction to planar graphs.} Given the hardness of the problem on arbitrary graphs, even for small constant values of $t$, it was but natural to look for specific graph classes where the problem might be more tractable with respect to the number of terminals $t$. In fact, Dahlhaus \etal had considered the restriction to planar graphs, which we call {\problemPMWC}. When \rev{$t$ is part of the input}, the problem still is {\nph}, even if all the edges of the planar graph are of unit weight. However, they showed that \problemPMWC could be solved in polynomial time for any fixed $t$. For $t=3$, their algorithm runs in time \Oh{n^3 \log n} and for $t \geq 4$ in time \Oh{(4t)^{t} \cdot n^{2t-1} \log n}. In the parameterized complexity parlance, their result implied that \problemPMWC belongs to the class \textsc{XP}. The obvious next step was to find out if it was also \textsc{FPT}. In 2012, however, Marx~\cite{Ma12} showed that assuming ETH holds, the problem does not admit any algorithm running in time $f(t) \cdot n^{\OO(\sqrt{t})}$ and showed that it is W[1]-hard.

In a companion paper, Klein and Marx \cite{KM12} gave a subexponential algorithm for \problemPMWC with a run-time of $2^{\OO(t)} \cdot n^{\OO(\sqrt{t})}$. Later, \colin~\cite{CDV} showed that this result extends to surfaces of any fixed genus. He showed this even holds for the more general {\sc Edge Multicut} problem, where given a set of terminal pairs, we ask for the smallest set of edges, which when removed, disconnects all the given terminal pairs. This yielded an algorithm which solves the problem in time $(g+t)^{\OO(g+t)} n^{\OO(\sqrt{g^2+gt+t})}$, where $g$ is the genus of the surface. This bound is almost tight assuming ETH holds, as was recently shown by Cohen-Added~\etal\cite{Cohen-AddadVMM21}.

Recently, Bandhopadhyay \etal~\cite{BLLSX22} showed that for both \problemEMWC and \problemNMWC \rev{(in the restricted variant, where terminals may not be deleted)} the square-root phenomenon extends to the class of $H$-minor-free graphs. In particular, on this class, they show that \problemEMWC can be solved in time $2^{\OO(\sqrt{s}\log^4 s)}n^{\OO(1)}$ and \textsc{Restricted}~\problemNMWC in time $2^{\OO(t\sqrt{s}\log s)}n^{\OO(1)}$, \rev{where $s$ is the size (weight) of the multiway cut.}

\paragraph{Parameterization by terminal face cover number.}
Due to the intractability of \problemPMWC~when the number of terminals is part of the input, it is worthwhile to look for other parameters that might generalize $t$. In particular, we could look at imposing restrictions on the location of terminals in the input plane-embedded graph. An extensively explored such restriction is when all the terminals are present on a small number of faces of the planar graph. This restriction on the input graph was studied by Ford and Fulkerson in their classic paper~\cite{FF56}. The minimum number of faces of the input planar graph that cover all the terminals was termed the \emph{terminal face cover number} \rev{$\three(G,T)$} by Krauthgamer \etal~\cite{Krauthgamer1}. The terminal face cover number is of broad interest as a parameter and has been studied with respect to several cut and flow problems~\cite{Krauthgamer1, krauthgamer2020refined, Chen-Wu, Multi-commodity}, shortest path problems~\cite{Frederickson91, ChenXu_shortestpaths}, finding non-crossing walks~\cite{EricksonN11}, and the minimum Steiner tree problem~\cite{KNL20}. In particular, {\sc Planar Steiner Tree} has an algorithm running in time $2^{\OO(\three(G,T) \log \three(G,T))} n^{\OO(\sqrt{\three(G,T)})}$.

The case when $\three(G,T) = 1$ is known as an Okamura-Seymour graph. It was shown by Chen and Wu~\cite{Chen-Wu} that when $\three(G,T) = 1$ and $G$ is biconnected, the minimum multiway cut in $G$ forms a minimum Steiner tree in its \rev{augmented} planar dual\footnote{An augmented planar dual differs from the standard dual graph in that the outer face is represented by several vertices, one per interval of the boundary vertices between two consecutive terminals (see also Section~\ref{sec:connected-duals}.)}. Consequently, one can find a minimum Steiner tree in the dual graph using the algorithm by Erickson~\etal\cite{Erickson} (see also Bern~\cite{Bern}) for the case when all the terminals lie on the outer face boundary of the graph.  They also gave an approximation algorithm for the case when $\three(G,T)>1$, which runs in time $\max\{\OO(n^2 \log n \log \three(G,T)),\OO({\three(G,T)}^2 n^{1.5} {\log}^2n + tn)\}$ and finds a solution within a ratio of $2-2/t$ of the optimum multiway cut. However, for the case when $\three(G,T) > 1$, no \rev{polynomial-time exact} algorithm was known.

\subsection{Our contribution}
In this paper, we resolve the complexity of \problemPMWC parameterized by the terminal face cover number, hereafter referred to as $k$. Given an edge-weighted planar graph $G=(V, E)$ with a fixed embedding, a set of terminals $T \subseteq V$, and a collection of faces $\mathcal{F}$ that cover all the terminals in $T$, the goal is to find a minimum weight cut of $G$ that pairwise separates the terminals in $T$ from each other. \rev{We may assume that such a set $\mathcal{F}$ of size $k$ is given as part of the input; if not, it can be computed in time $2^{\OO(k)} n^{\OO(1)}$ by the algorithm of Bienstock and Monma~\cite{BienstockM88} (this will not affect the final running time of our algorithm).} Our main contribution is:

\begin{theorem}[restate=thmmain] \label{thm:main}
\problemPMWC can be solved in time $2^{\OO(k^2 \log k)} n^{\OO(\sqrt{k})}$, where $k$ is the terminal face cover number of the instance. 
\end{theorem}

Since $k \leq t$, an \rev{almost matching} lower bound immediately follows from Marx's result~\cite{Ma12}.

\begin{theorem}[implied by Marx~\cite{Ma12}]
\rev{Unless ETH fails, there can be no algorithm that solves \problemPMWC and runs in time $n^{\OO(\sqrt{k})}$, where $k$ is the terminal face cover number of the instance.} 
\end{theorem}

We note that while the running time of our algorithm is reminiscent of the algorithm by Kisfaludi-Bak \etal~\cite{KNL20} for {\sc Planar Steiner Tree} parameterized by the terminal face cover number, our algorithm is and needs to be substantially more involved. \rev{The intuition of their algorithm is to show that the union of a minimum Steiner tree for a set of terminals that can be covered by $k$ faces and the graph induced on the $k$ terminal faces, has bounded treewidth.} Then they can `trace' this tree by a recursive algorithm that guesses the vertices of the solution in a separator implied by the tree decomposition. While one can show that the dual of a minimum multiway cut is a subgraph of the dual of bounded treewidth, tracing a solution through such a tree decomposition is not straightforward, and we need significant assistance from tools from topology, particularly homotopy, which were not needed for {\sc Planar Steiner Tree}.

Intuitively, we describe the high-level topology that must be respected by the dual of some optimum solution. Part of this topology is a planar graph, of which we can find a sphere-cut decomposition of width $O(\sqrt{k})$. We then apply a dynamic programming routine on this decomposition to find the optimum solution. While this dynamic programming routine is sufficient to find the high-level structure, the number of terminals is too large for this to effectively deal with the local structure of the solution. We then merge the popular algorithm by Dreyfus-Wagner \cite{Dreyfus-Wagner} for finding \msts, which runs in polynomial time in cases relevant for us~\cite{Erickson,Bern}, with the algorithm of Frank-Schrijver \cite{Homotopy} to find shortest paths homotopic to a given prescription. We argue that this prescription can be efficiently guessed. This leads to a neat dynamic program used to find the local structure, and then finally, an optimum solution.

In this sense, our algorithm has more in common with the approaches of Klein and Marx~\cite{KM12} and \colin~\cite{CDV}, which extensively rely on topological arguments, and particularly homotopy. However, significant effort is needed to generalize from the parameter number of terminals employed in those works to the parameter number of faces covering terminals. This not only affects the structural results that employ homotopy, but also makes the final algorithms more involved.

We provide a more detailed overview of our algorithm in Section~\ref{sec:overview}. The full paper then follows.

\subsection{Related Work}
We discuss related work in some more detail, particularly on the parameterized complexity of \problemMWC. Until 2006, not much was known with regard to the parameterized complexity of \problemMWC. Then, \rev{Marx's highly influential result~\cite{tcs/Marx06} showed that} \problemNMWC, where one must remove a subset of vertices instead of edges of the input graph, is \textsc{FPT} parameterized by the size of the solution (hereby denoted by $s$). His algorithm \rev{runs} in time \Ostar{4^{s^3}}, \rev{where \Ostar{\cdot} absorbs any polynomial factors}. Since then, the \rev{running time} has been considerably improved~\cite{Xiao, CCF14, CLL09, GUILLEMOT_NMWC, CPPW13}, with the current best being \Ostar{1.84^{s}} for \problemEMWC~\cite{CCF14}, and \Ostar{2^{s}} for \problemNMWC~\cite{CPPW13}. It is a notoriously hard problem to determine if \problemMWC has a polynomial kernel. The best known kernel on general graphs has size $2^{\OO(s)}$, which follows from the \FPT~algorithm of Marx~\cite{tcs/Marx06}. Kratsch and Wahlst\"{o}m~\cite{KW12} presented a kernel that has $O(s^{t+1})$ vertices for \problemNMWC. Their result implies a polynomial kernel also for \problemEMWC. However, when $|T|$ is unbounded, the question whether a polynomial kernel exists for \problemMWC is far from resolved. Recently, Wahlst\"{o}m~\cite{Wa22} made progress towards answering this question by demonstrating a kernel of size \rev{$s^{\OO(\log^{3}s)}$}.  
On directed graphs, Chitnis \etal~\cite{CHM13} showed that \problemNMWC~and \problemEMWC~are equivalent and admit an \textsc{FPT} algorithm running in time \Ostar{2^{2^{\OO(s)}}}. 

For planar graphs, we also mention some other exact algorithms. In particular, Hartvigsen~\cite{Hartvigsen98} gave an $O(t4^tn^{2t-4} \log n)$ time algorithm, which improved on the original exact algorithm by Dahlhaus~\etal\cite{DJPSY94}. The simple algorithm by Yeh~\cite{Yeh01}, unfortunately, seems incorrect~\cite{CheungH10}. The mentioned algorithm of Klein and Marx~\cite{KM12} is faster than all of them. \rev{Bentz} also considered the generalization to {\sc Edge Multicut}, first with terminals only on the outer face~\cite{Bentz19} and for the general case~\cite{Bentz12}. Unfortunately, the latter general result seems to have several flaws~\cite{CDV}. The algorithm by \colin\cite{CDV} for this problem that we already mentioned is also faster and more general. Finally, for planar graphs and parameter $s$, Pilipczuk \etal\cite{pili} showed that even a polynomial kernel in $s$ exists for \problemEMWC, leading to an algorithm with running time \Ostar{2^{\OO(\sqrt{s}\log s)}}. This kernel was recently extended to \problemNMWC, \rev{again parameterized by solution size}, by Jansen \etal~\cite{EJ}.

There have been surprisingly few studies investigating the complexity of \problemMWC~on other hereditary \rev{graph classes}. Until recently, even the classical complexity of the problem was unknown on well-known hereditary \rev{graph} classes like chordal or interval graphs. Misra \etal~\cite{MPRSS20} showed that on chordal graphs, \problemNMWC admits a polynomial kernel parameterized by the solution size. Later, Bergougnoux \etal~\cite{BergougnouxPT22} showed that \problemNMWC is \FPT~parameterized by the rankwidth, and \XP~parameterized by the mimwidth of the input graph. Their result implies a polynomial time algorithm for the problem on interval graphs, permutation graphs, bi-interval graphs, circular arc and circular permutation graphs, convex graphs, and $k$-polygon and dilworth-$k$ graphs for fixed $k$. Recently, Galby \etal~\cite{GalbyMSST22} improved both the aforementioned results on chordal graphs by showing that \problemNMWC is polynomial-time solvable \rev{on this class}.

We also discuss further works on approximation algorithms. We already mentioned several constant-factor approximations for general graphs~\cite{BHKM11, CKR00, Karger_approx} and planar graphs~\cite{Chen-Wu}. Bateni~\etal\cite{BHKM11} presented a \ptas for \problemPMWC which combined the technique of brick-decomposition from \cite{brick-decomposition}, the clustering technique from \cite{BHM11}, and a technique to find short cycles enclosing prescribed amounts of weight from \cite{shortcycle}. The more general \textsc{Edge Multicut} problem is known to be APX-hard, even in trees. 

However, Cohen-Addad~\etal\cite{Multicut-approx} recently gave a $(1+\varepsilon)$-approximation scheme for \textsc{Edge Multicut} on graphs embedded on surfaces, with a fixed number of terminals. Their approximation scheme runs in time $(g+t)^{\OO((g+t)^3)}\cdot (1/\varepsilon)^{\OO(g+t)}\cdot n \log n$, \rev{where $g$ is the genus of the surface}. It also contains uses of the Dreyfus-Wagner and Frank-Schrijver algorithm in a combination not dissimilar from our own, but in different circumstances.

Finally, our work leans on a proper specification of the homotopy of the dual of the solution. This approach was first applied explicitly to {\sc Edge Multicut} (and to \problemEMWC) by \colin\cite{CDV} and is also evident in the work of Klein and Marx~\cite{KM12}. The understanding of homotopy and its uses in cut and flow problems on graphs on surfaces was built through a sequence of works, see \eg\cite{Chambers,ChambersEN09,EricksonFN12,EricksonN11a,EricksonN11}.

\section{Overview} \label{sec:overview}
At a high level, our approach is reminiscent of the one by Klein and Marx \cite{KM12} as well as \colin \cite{CDV} used for the weaker parameter $t$, the number of terminals. Therefore, we start by giving a short overview of their work, before discussing our own algorithm for parameter $k$.

\subsection{Previous Approaches}
The starting observation is that cuts in planar graphs correspond to cycles in the dual of the graph~\cite{Reif83}. Hence, it makes sense that almost the entire algorithms and structural descriptions of Klein and Marx and \colin, as well as ours, work with the planar dual. If the number $t$ of terminals is bounded, this quickly leads to the intuition that the planar dual of a multiway cut should be a planar graph with $t$ faces, one for each terminal~\cite{DJPSY94}. By dissolving vertices of degree~$2$ of this planar graph and applying Euler's formula, one obtains a planar graph $S$ with $O(t)$ faces, vertices, and edges. One can then guess what $S$ looks like by exhaustive enumeration and guess the vertices of the dual corresponding to the vertices of $S$. Then it remains to expand the edges of the graph back into shortest paths.

However, these paths are not just any shortest paths. Instead, they must contort themselves between the terminals in a particular way, such that they perform their job in separating the terminals. This is where topological arguments come in. The layout of these paths is described using the sequence of crossings with a Steiner tree on the terminals. The main thrust of the work of Klein and Marx and \colin~is to argue that in some optimal solution these crossing sequences are short, in the sense that their length depends only on $t$. One can then guess the crossing sequences of each of the paths in the optimum and find a shortest path following a particular crossing sequence in polynomial time using Frank and Schrijver's algorithm~\cite[Section 5]{Homotopy}.

The above leads to an $f(t) \cdot n^{\OO(t)}$ time algorithm. To improve to an $f(t) n^{\OO(\sqrt{t})}$ time algorithm, it suffices to observe that since $S$ is planar, it has treewidth $O(\sqrt{t})$, as follows from the planar separator theorem~\cite{LiptonTarjan}. One can then replace the guessing of the vertices of the dual corresponding to vertices of $S$ by a dynamic program on the tree decomposition.

\subsection{\texorpdfstring{Warm-up: An $n^{\OO(k)}$-time Algorithm}{Warm-up: An n\^{O(k)}-time Algorithm}}
We design our approach along the same lines, in that we show that the topology of the dual graph of the minimum multiway cut is constrained and then enumerate all the possible topologies. For a certain topology, we find the shortest solution respecting the topology. However, since we must deal with a huge number of terminals, possibly \Oh{n} many of them, this is not straightforward. Indeed, we need stronger structural properties of the solution to limit the number of diverse topologies.

\paragraph{Structure: Global and Local.} 
From a global perspective, however, the structure of an optimal solution looks very similar. If we think of a terminal face as a single terminal, then as in the previous works, we see that (some part of) the dual of the solution must separate the different terminal faces. We call this part the \emph{skeleton} of the dual of the solution\footnote{Our notion of skeleton should not be confused with the notion of skeleton as defined by Cohen-Addad~\etal\cite{Multicut-approx}. In particular, our notion of a (shrunken) skeleton is very reminiscent of the multicut dual of \colin\cite{CDV}, provides a high-level view of the solution, and has no weight restrictions. The skeleton of~\cite{Multicut-approx} instead can be seen as `orthogonal' to the solution; by controlling its weight, the skeleton can be portalized (as in Arora~\cite{Arora98}) and the solution can be approximated by combining Steiner trees of a particular homotopy between portals. This is different from our use and definition of a skeleton.}. The paths between its branching points are called \emph{bones}. \rev{See Figure~\ref{fig:broken bones}}. By dissolving dual vertices of degree~$2$ in the skeleton, we obtain the \emph{shrunken skeleton}; its edges are called \emph{shrunken bones}. The shrunken skeleton has $k$ faces and $O(k)$ vertices and edges by Euler's formula. It follows that the shrunken skeleton of the optimum can be guessed by exhaustive enumeration and has treewidth $O(\sqrt{k})$, a great starting point. This is discussed in more detail in Section~\ref{sec:Skeleton and Nerves}.

\begin{figure}[t]
    \centering
    \includegraphics[width=0.4\textwidth]{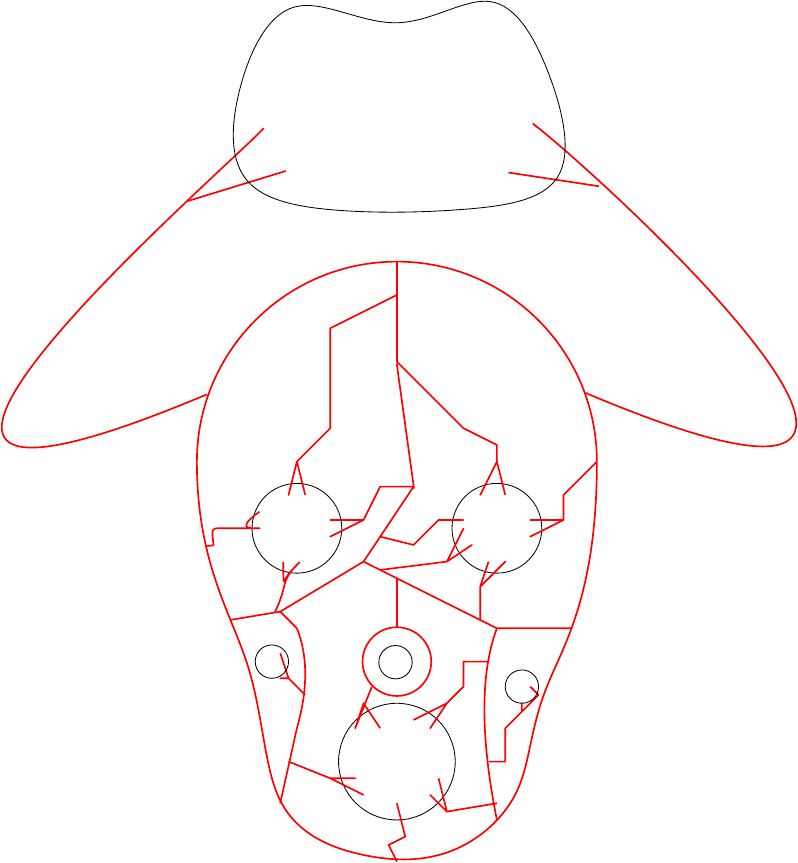}
    \caption{This figure illustrates the bones and nerves of the skeleton, all drawn in red. The terminal faces are drawn as closed black curves.}\label{fig:broken bones}
\end{figure}

Now consider the local parts of the solution, namely the part of the dual of the solution inside each of the $k$ faces of the skeleton. Chen and Wu~\cite{Chen-Wu} proved that when there is a single terminal face that is a simple cycle, then the dual of a minimum multiway cut is a minimum Steiner tree in the dual graph. To be more precise, this holds in the augmented dual graph. This graph splits the dual vertex $s$ of the face corresponding to the simple cycle into $r$ new vertices \rev{in the dual}, where $r$ is the number of terminals of the face. \rev{For each maximal subpath of the cycle that contains no terminals as inner vertices, one of the new vertices is made incident to all dual edges corresponding to the edges of this subpath.} Then, a minimum multiway cut is a minimum Steiner tree in the augmented dual with the new vertices as terminals. We generalize this argument to prove that inside each face of the skeleton, the solution is a forest of minimum Steiner trees on the augmented terminals. We call these trees \emph{nerves}. All nerves attach to the boundary of the face of the skeleton. Crucially, the augmented terminals belonging to each of these nerves form an \emph{interval} of the set of augmented terminals. Then, by applying the Dreyfus-Wagner algorithm~\cite{Dreyfus-Wagner}, any nerve can be found in polynomial time~\cite{Erickson,Bern}, if we know the interval \rev{and the attachment point on the boundary of the face}. See Section~\ref{sec:Skeleton and Nerves}.

\paragraph{Algorithm: Global and Local.}
As a warm-up, we now discuss how to find an algorithm with running time $f(k) n^{\OO(k)}$. First, guess the shrunken skeleton of the optimum solution by exhaustive enumeration. Then, for each shrunken bone of this shrunken skeleton, we note that it needs to expand to separate two terminal faces, say $F_\one, F_\two \in \mathcal{F}$, and some terminals on each of them. For both the terminal faces, we guess the intervals of augmented terminals $I_\one$ and $I_\two$ covered by nerves that attach to the corresponding bone. Since intervals are between two augmented terminals, there are $n^{\OO(k)}$ intervals, and we can guess the optimal intervals by exhaustive enumeration in the same time.

\rev{A major complication now is that we may need to remember a large number of crossing sequences for the subpaths between the numerous nerves.} While each of those paths individually again has a short crossing sequence, as can be argued by adapting the approach of Klein and Marx~\cite{KM12} and \colin~\cite{CDV}, there are too many of these paths to efficiently guess all these crossing sequences. Instead, we argue that we can group these crossing sequences while keeping them small, and that we do not need to guess the crossing sequences for all groups.

We start by discussing the grouping. \rev{We group consecutive nerves on the bone that span intervals of terminals on the same face.} We then prove that the union of the paths between the nerves in such a group has a bounded crossing sequence (see Section~\ref{sec:Bones and Homotopy}). Next, we argue that we only need few groups. We say that two groups of nerves \emph{alternate} if they are towards different terminal faces. If there are two alternating groups of nerves at the \rev{start} of the bone, and two alternating groups of nerves at the \rev{end} of the bone, then we can observe that these nerves effectively cordon off \rev{a region} of the plane. Inside this region, we can show that all terminals effectively lie on a cycle \rev{in the primal graph}. This enables us to use the ideas of Chen and Wu~\cite{Chen-Wu} again to describe an optimal solution.

\rev{To conclude our algorithm, we guess for each bone how many groups there are in the solution: one, two, three, or four or more. If there are at most three groups, then we guess their starting and ending nerves (defined by their intervals and attachment points) and their crossing sequences by exhaustive enumeration. We can then employ a dynamic program to find all nerves in between, while ensuring that the path between the attachment points of the nerves follows the guessed crossing sequence. 
\rev{In particular, we emply Frank and Schrijver's algorithm~\cite[Section 5]{Homotopy} that can find shortest paths that follow a particular crossing sequence in polynomial time. See Section~\ref{sec:splinting} for details.}
This will take care of all (augmented) terminals between the intervals.} 

\rev{If there are four or more groups, then from the preceding, there is a special region between the two alternating groups at each end of the bone. We can compute a solution for the terminals in that region in polynomial time, using ideas of Chen and Wu~\cite{Chen-Wu}, and determine the groups again as before, after guessing their starting and ending nerves (defined by their intervals and attachment points) and their crossing sequences.}

The total running time of this algorithm is indeed $f(k) \cdot n^{\OO(k)}$. The latter term originates from the guessing of the terminals corresponding to the branching points of the skeleton, the guessing of the intervals, and \rev{the guessing of the intervals and attachment points of the nerves at the bookends of the up to four groups we need to consider}. For each shrunken bone, we call this information a \emph{broken bone} (we effectively guess its parts) and the solution that fixes it a \emph{splint}.

\subsection{Towards the Final Algorithm}
We now develop the ideas for the final algorithm. To avoid guessing the broken bones globally, we aim to apply the fact that the shrunken skeleton has bounded treewidth. Using a tree decomposition directly is cumbersome and not very intuitive. Instead, we use sphere-cut decomposition.

A sphere-cut decomposition is a branch decomposition, which can be thought of as a set of separators of the graph organized in a tree-like fashion. The tree structure enables dynamic programming in the usual manner. The crux is that the vertices of each separator induce a noose in the planar graph. A \emph{noose} in this sense is a closed curve in the plane that intersects every face of the planar graph at most once and intersects the drawing only in the vertices of the separator. The tree-structure of the decomposition is organized in such a way that each separator has two child-separators and the symmetric difference of the corresponding two nooses is the noose of the parent-separator. A formal definition is in Section~\ref{sec:sphere-cut}.

It is known that there is a sphere-cut decomposition of a planar multigraph on $k$ vertices where all nooses (and separators) have $O(\sqrt{k})$ vertices~\cite{DornPBF10,MarxP15,PilipczukLW20}. This bound will lead to the $n^{\OO(\sqrt{k})}$ running time. This decomposition requires that the planar graph is connected, has no bridges, nor has self-loops. Unfortunately, our shrunken skeleton does not satisfy any of those demands out of the box. The issue of self-loops is quickly handled by not dissolving all vertices of degree~$2$ when obtaining a shrunken skeleton from a skeleton, but only those that do not lead to a self-loop.

Next, we ensure connectivity of the shrunken skeleton. To this end, we consider a connected component of the shrunken skeleton that is `innermost' in the embedding of the shrunken skeleton. We can guess which of the terminal faces are enclosed by this connected component. \rev{However, there may also be a terminal face for which all but one of its terminals are enclosed by the component}. We cannot guess this terminal (as there are $n$ choices), even though it is necessary to know the terminal to correctly guess the structure of the optimum solution. To circumvent this issue, we argue that we can pick a single terminal of this face as a representative of this `exposed' terminal. Thus we completely avoid knowing which terminal is exposed. This enables a $2^{\OO(k)} n^{\OO(1)}$ time subroutine to guess the components of the dual of an optimum solution and to reduce to the case where such duals are connected. We then argue that this implies that the shrunken skeleton is connected as well. See Section~\ref{sec:connected-duals} for details.

Then, we consider bridges of the shrunken skeleton. It seems hard to avoid them completely. Instead, we design another dynamic program (see Section~\ref{sec:bridge-block-reduction}). We use a bridge block tree of the skeleton to guide this dynamic program. Recall that a bridge block tree is a tree representation of the bridge blocks (bridgeless components) of a graph and the cut vertices between those bridge blocks, generalizing the more familiar block cut tree. To ensure this bridge block tree is suitable for the dynamic program, we need a modified version of a bridge block tree that organizes itself according to the embedding of the shrunken skeleton. To this end, we develop an embedding-aware bridge block tree (see Section~\ref{sec:eabb}), which may be of independent interest.

\rev{We can now indeed obtain the sphere-cut decomposition of the shrunken skeleton.} We now apply a dynamic program where we just maintain the broken bones for the shrunken bones of the faces intersected by a noose of the decomposition. In fact, this still is too much information, and instead we maintain only a relevant part of those broken bones, namely the first nerve that we encounter in either direction on the shrunken bones of each face that is intersected by the noose. We argue that this yields sufficient information to know all broken bones and obtain an optimum solution. See Section~\ref{sec:bridge-block} for details.

In conclusion, our final algorithm is as follows. First, we perform a subroutine to ensure that the dual of any minimum multiway cut is connected. Then we guess the structure of the dual of such a solution, namely its shrunken skeleton, how the nerves of each bone are grouped, and what the crossing sequences are of each path between nerves of the same group and between the groups. We call this a \emph{topology}.\footnote{The word topology has many well-known meanings, including a branch of mathematics. We use the term here in a cartographic sense, as an abstract map of the solution. This is in line with previous uses in the literature, \eg\cite{CDV}.} We guess the optimal topology by exhaustive enumeration. Consider its shrunken skeleton and define a dynamic program on its bridge blocks to combine partial solutions for each bridge block. For each bridge block, we show that it has a sphere-cut decomposition with nooses of bounded size, which enables us to find a partial solution using a dynamic program. The guessing of the topology and the dynamic programs combined lead to an algorithm running in time $2^{\OO(k^2 \log k)} n^{\OO(\sqrt{k})}$.
 
\section{Preliminaries}\label{Prelims}

\paragraph{Graphs.}
A \emph{graph} is a pair $G=(V, E)$, such that \rev{$E \subseteq {V \choose 2}$}. The elements of $V$ are called the vertices of the graph and the elements of $E$ its edges. The number of vertices of a graph is its order denoted by $|G|$. A \emph{subgraph} $G'= (V', E')$ of $G$, written as $G'\subseteq G$, is a graph such that $V'\subseteq V$ and $E'\subseteq E$. If $G'\subseteq G$ and $G'\neq G$, then $G'$ is a \emph{proper subgraph} of $G$. $G'$ is an \emph{induced subgraph} of $G$, if for all $x,y \in V'$, if $xy \in E$, then $xy \in E'$.

A \emph{path} is a non-empty graph $P=(V,E)$ of the form $V=\{x_0, x_1, \ldots, x_k\}$ and $E=\{x_0x_1, \ldots,$ $x_{k-1}x_k\}$, where all $x_i$ are distinct. The number of edges in a path is its \emph{length}. We use the notation $P[x_i,x_j]$ to denote the subpath of $P$ between vertices $x_i$ and $x_j$. For a tree $N$, we use $N[x,y]$ to denote the unique path in $N$ between the vertices $x$ and $y$.

The graph is \emph{connected} if there is a path between any two vertices in \rev{the} graph, and \emph{disconnected} otherwise.

The \emph{contraction} of an edge $e=(u,v)$ of $G$ is the operation of identifying $u$ and $v$, while removing any loops or parallel edges that arise. We use the notation $G/e$. This extends to sets $F \subseteq E(G)$ of edges, for which we can use the notation $G/F$.

Given a subset $T \subseteq V(G)$ (called terminals), a \emph{Steiner tree} on $T$ is a minimal connected subgraph $H$ of $G$ such that there is a path in $H$ between any two terminals in $T$.

\paragraph{Connectivity.}
A \emph{cut vertex} of a connected graph is a vertex whose removal yields a disconnected graph. A \emph{biconnected graph} is a \rev{connected} graph without cut vertices. A \emph{biconnected component} or \emph{block} of a graph is a maximal subgraph that is biconnected. 

A \emph{bridge} of a connected graph is an edge whose removal yields a disconnected graph. A \emph{bridgeless graph} is a graph that has no bridges. A \emph{bridgeless component} or \emph{bridge block} of a graph is \rev{either a single edge that is a bridge or} a maximal subgraph that is bridgeless. \rev{We call the former a \emph{trivial} bridge block and the latter (a bridge block that consists of more than one edge) a \emph{nontrivial} bridge block}. Bridge blocks are incident to each other at cut vertices of the graph. \rev{For simplicity, we call two bridge blocks \emph{neighboring} if they share a cut vertex. We can observe that any nontrivial bridge block neighbors only trivial bridge blocks.}

Given two disjoint vertex subsets $X,Y \subseteq V(G)$, an \emph{$(X,Y)$-cut} is a set of edges whose removal leaves no path between any vertex of $X$ and any vertex of $Y$.

Given a subset $T \subseteq V(G)$, a \emph{(edge) multiway cut} \rev{of $(G,T)$ (or $T$)} is a set $C \subseteq E(G)$ whose removal leaves no path between any pair of distinct vertices in $T$. \rev{A multiway cut of $(G,T)$ is \emph{inclusion-wise minimal} or simply \emph{minimal} if no proper subset of $C$ is also a multiway cut of $(G,T)$. A minimum multiway cut of $(G,T)$ is a smallest set of edges that is a multiway cut of $(G,T)$. If, additionally, a weight function $\omega$ on the edges is given, then we may speak of a (minimal/minimum) multiway cut of $(G,T,\omega)$; in particular, a minimum multiway cut of $(G,T,\omega)$ is a minimum-weight set of edges that is a multiway cut of $(G,T,\omega)$.}

\paragraph{Topology and Planar Graphs.} \label{sec:topology}
A Jordan \emph{arc} in the plane is an injective continuous map of $[0,1]$ to $\mathbb{R}^2$. A Jordan \emph{curve} in the plane is an injective continuous map of $\mathbb{S}^1$ to $\mathbb{R}^2$. If one of the points on this curve is special, we may also call this a closed arc on this point. 
Consider a \rev{finite} set $Z$ of Jordan (possibly closed) arcs in the plane. Let $P(Z)$ denote the union of the set of points of each arc of $Z$. Observe that $\mathbb{R}^2 \setminus P(Z)$ is an open set. A \emph{region} of $Z$ is a maximal subset $X$ of $\mathbb{R}^2 \setminus P(Z)$ that is \emph{arc-connected}; that is, there is a Jordan arc in $X$ (meaning all its points belong to $X$) between any pair of points in $X$.
If $\mathbb{R}^2 \setminus P(Z)$ has more than one region, then $Z$ is \emph{separating}.
Let $X$ be a region of a separating set $Z$. The \emph{boundary} $\bd{X}$ of a region $X$ is the set of all points $p \in \mathbb{R}^2$ such that every open disk around $p$ contains both a point of $X$ and of $\mathbb{R}^2 \setminus X$. The \emph{complement} of a region $X$ is the union of $Y \cup \bd{Y}$ for each region of $\mathbb{R}^2 \setminus (X \cup \bd{X})$. Note that the complement of a region is not necessarily a region itself, but possibly a union of regions (and their boundaries), particularly if $X$ has holes.

We say that a region $X$ \emph{encloses} a set $Y$ if $Y \subseteq X \cup \bd{X}$ and \emph{strictly encloses} $Y$ if $Y \subseteq X$.

For the definition of planar graphs, we follow Diestel~\cite{Di05}. A graph $G=(V,E)$ is \emph{plane} if $V$ corresponds to a set of points (vertex points) in the plane and $E$ corresponds to a set of arcs in the plane (edge arcs) between the points corresponding to its endpoints, such that the interior of each edge arc contains no vertex point and no point of any other edge arcs. We call the vertex points and edge arcs an \emph{embedding} of $G$. If $G$ admits an embedding, we call the graph a \emph{planar graph}. A \emph{face} of a plane graph is any region of the set of edge arcs. Exactly one face is \emph{unbounded}, also called the \emph{outer face}, whereas all other faces are \emph{bounded}.

\rev{We note that the boundary of each face is a closed walk. We call the length of this walk the \emph{length} of the face. Finally, we say that one face neighbors another if their boundaries share an edge.}

Two plane graphs are \emph{equivalent} if they are isomorphic and the circular order of the edges around each vertex is the same in both embeddings. In particular, this means that boundaries of the faces of the embeddings have the same edge sets and there is a bijection \rev{between the sets of faces such that each face is incident to the same set of edges as its image under the bijection}.

Let $G=(V,E,F)$ and $G^*=(V^*, E^*, F^*)$ be two plane graphs, where $V, E$ and $F$ ($V^*, E^*$ and $F^*$) denote the set of vertices, edges, and faces of $G$ ($G^*$). $G^*$ is called a \emph{plane dual} of $G$, if there exist bijections\\
\begin{math}
\begin{array}{c c c}

   f^*: V\rightarrow F^*    &    e^*: E\rightarrow E^*      &     v^*: F\rightarrow V^*\\
    v\rightarrow f^*(v)     &    e\rightarrow e^*(e)        &     f\rightarrow v^*(f)\\

\end{array}
\end{math}
satisfying the following conditions:
\begin{enumerate}[label= (\alph*.)]
    \item $v^*(f) \in f$, \rev{for all} $f\in F$
    \item $e$ and \rev{$e^*(e)$} intersect in exactly one point, which lies in the interior of both $e$ and \rev{$e^*(e)$}, \rev{for all $e\in E$}.
    \item $v \in f^*(v)$, \rev{for all} $v\in V$
\end{enumerate}

We note the following basic properties, which follow from Diestel~\cite{Di05}. We will often use them without explicitly referring to this proposition.

\begin{proposition}
If $G$ is connected, then the edges of $G$ bounding each face form a closed walk (also known as a face walk).
If $G$ is bridgeless, then for any edge, the faces on both sides are distinct. 
\end{proposition}

Let $O$ be a set of points in the plane, called \emph{obstacles}. Consider two Jordan arcs $a,b$ between the same pair of points, such that neither arc contains a point of $O$. Then these arcs are \emph{homotopic} if and only if  there is a continuous deformation between $a$ and $b$ that does not cross a point of $O$. In the plane, this means that \rev{the union of the bounded regions induced by $a$ and $b$} does not contain any point of $O$.

\subsection{Sphere-cut Decompositions}\label{sec:sphere-cut}
\rev{\emph{Note to the reader: a reader only interested in an $n^{\OO(k)}$-time algorithm may skip this subsection.}}

A main component of our algorithm is a dynamic program over a planar graph of bounded treewidth. However, using a normal tree decomposition is rather cumbersome in this case, and it turns out to be much easier to instead use a sphere-cut decomposition, a branch decomposition especially suited for planar graphs. We define all necessary notions and state the relevant theorem below.

A \emph{branch decomposition} of a graph $G=(V,E)$ is a pair $(R,\eta)$ of a ternary tree $R$ and a bijection $\eta$ between the leaves of $R$ and the edges in $E$. For an edge $e$ of $R$, we define the \emph{middle set} $\midset(e)$ to be the set of vertices in $V$ for which an incident edge is mapped by $\eta$ to a leaf in the one component of $R-e$ and an incident edge is mapped by $\eta$ to a leaf in the other component. The \emph{width} of the branch decomposition is defined as the maximum size of the middle set of any edge of $R$. The \emph{branchwidth} of $G$ is then the minimum width of any branch decomposition of $G$.

Let $G$ be a (planar) graph with a fixed embedding on the sphere. Then a \emph{noose} $\noose$ (with respect to $G$) is a closed, directed curve in the sphere that meets the embedding of $G$ only in its vertices and that traverses each face at most once. The \emph{length} of the noose is equal to the number of vertices of $G$ traversed by it. If we enumerate the vertices of the noose, we implicitly assume that this enumeration follows the order of appearance on the noose, that is, following its direction. Note that a noose cuts the sphere into two regions, each homeomorphic to an open disk. The region bounded by and to the right when following the noose with its direction is denoted by $\enc(\noose)$ and the other by $\exc(\noose)$.

A sphere-cut decomposition of a graph $G$ with a fixed embedding on the sphere is a triple $(R,\eta,\delta)$ consisting of a branch decomposition $(R,\eta)$ and a mapping $\delta$ from the set of ordered pairs of adjacent vertices $x,y$ of $R$ to nooses (with respect to $G$) on the sphere, \rev{such that for all pairs of adjacent vertices $x,y$ of $R$}:
\begin{itemize}
\item $\delta(x,y)$ is the same noose as $\delta(y,x)$ but with the direction reversed. Note that then it holds $\enc(\delta(x,y)) = \exc(\delta(y,x))$;
\item $\delta(x,y)$ meets the embedding of $G$ exactly in the vertices of the middle set $\midset(x,y)$; moreover, $\enc(x,y)$ contains all the embeddings of all edges of the one component of $R-xy$ and $\exc(x,y)$ contains the embeddings of all other edges.
\end{itemize}
As noted by Dorn \etal\cite{DornPBF10} and Pilipczuk \etal\cite{PilipczukLW20}, we may assume that a sphere-cut decomposition is \emph{faithful}. That is, for every internal vertex $x$ of $R$ with adjacent vertices $y_1,y_2,y_3$, we may assume that $\enc(x,y_1)$ is equal to the disjoint union of $\enc(y_2,x)$, $\enc(y_3,x)$, and $(\delta(y_2,x) \cap \delta(y_3,x)) \setminus \delta(x,y_1)$. We also note that $\delta(y_2,x) \cap \delta(y_3,x) \cap \delta(x,y_1)$ consists of two points, each of which may (or may not) coincide with a vertex of $G$.

As described by Dorn \etal\cite{DornPBF10}, we can ``root'' any sphere-cut decomposition $(R,\eta,\delta)$ as follows. \rev{We first subdivide} an arbitrary edge $e$ of $R$. Let $u$ be the newly created vertex and $e',e''$ be the newly created edges. \rev{Note that, by definition,} $\midset(e')=\midset(e'')=\midset(e)$. Add a new vertex $r$ and connect it to $u$. \rev{By abuse of notation}, we set $\midset(ru) = \emptyset$ \rev{(note that $r$ is not mapped to an edge of $E$)}. We then direct the tree $R$ towards the root $r$. In the remainder, we assume our sphere-cut decompositions are rooted in this way.

The following result was observed by Dorn \etal\cite{DornPBF10} and follows from~\cite{GuT12,SeymourT94} (see also Marx and Pilipczuk~\cite{MarxP15} and Pilipczuk \etal~\cite{PilipczukLW20}).

\begin{theorem}\label{thm:sc}
Every $n$-vertex connected, bridgeless multigraph without self-loops but with a fixed embedding on the sphere has a faithful sphere-cut decomposition of width $\sqrt{4.5n}$. Moreover, such a sphere-cut decomposition can be found in $O(n^3)$ time.
\end{theorem}

\subsection{Embedding-Aware Bridge Block Trees}\label{sec:eabb}
\newcommand{\precp}{\prec_{P}}

\rev{\emph{Note to the reader: a reader only interested in an $n^{\OO(k)}$-time algorithm may skip this subsection.}}

An important aspect of our algorithm will be to deal with bridge blocks of a planar graph, as a sphere-cut decomposition can not. We define the following notion, which lends itself in a nice way to the dynamic programming algorithm we develop towards the end of the paper.

We can define a \emph{bridge block tree} of a graph $H$ as follows. Consider the graph $F$ that has a node for every bridge block (\ie a \rev{nontrivial} bridgeless component or a bridge), called the BB-nodes of $F$, and for every endpoint of a bridge, called the C-nodes of $F$. There is an edge between a BB-node corresponding to a bridge block $B$ and a C-node that corresponds to a cut vertex $v$ if $v \in V(B)$. It is immediate that \rev{this graph is a forest that is connected (a tree) if $H$ is connected and that is disconnected otherwise}. \rev{Recall that any nontrivial bridge block neighbors only trivial bridge blocks. Thus, we can observe that each C-node of a bridge block tree is adjacent to at most one BB-node that corresponds to a nontrivial bridge block.}

\begin{theorem}[{Tarjan~\cite{Tarjan74}}]\label{thm:bridge-block:tarjan}
A bridge block tree of a graph can be computed in linear time.
\end{theorem}

If $H$ is plane and connected, then we use an extension of this definition. An \emph{embedding-aware bridge block tree} (or \emph{EABB} tree) $L(H)$ of $H$ is formed from the bridge block tree $F$ as follows. Root $F$ at a BB-node $\ell(F)$ that corresponds to a bridge block that has an edge bordering the outer face of $H$. We now perform two operations on the BB-nodes.

\begin{figure}[t]
    \centering
    \includegraphics[width=\textwidth, keepaspectratio]{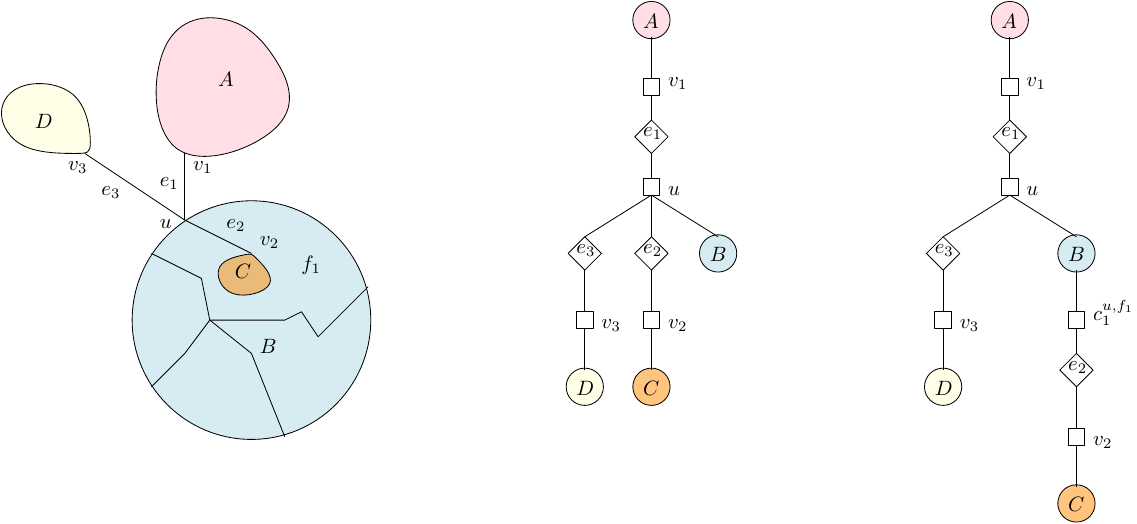}
    \caption{The leftmost image is that of a graph with bridge blocks $A, B, C$ and $D$. The bridges are denoted by $e_1, e_2$ and $e_3$, whereas the cut vertices are $u$, $v_1, v_2$, and $v_3$. The figure in the middle shows a bridge-block tree corresponding to the graph. The rightmost figure shows the resulting embedding-aware bridge block tree. \rev{In particular, the first of the two mentioned operations needed to be applied.}}
    \label{fig:eabb tree}
\end{figure}

First, for any nontrivial bridge block $B$ whose corresponding BB-node $b$ has a C-node parent $c$ in $F$ corresponding to a cut vertex $v$, note that all other BB-node children $b'$ of $c$ correspond to trivial bridge blocks. For each bounded face $f$ of $B$, create a new C-node child $c_1^{v,f}$ of $b$ corresponding to $v$ and $f$, and for any child $b' \not= b$ of $c$ corresponding to a bridge edge that is contained in $f$, make the subtree of $F$ rooted at $b'$ a child of $c_1^{v,f}$. We only add $c_1^{v,f}$ if there are any such children $b'$.

Second, for any nontrivial bridge block $B$ whose corresponding BB-node $b$ has a C-node child $c$ in $F$ corresponding to a cut vertex $v$, note that all BB-node children $b'$ of $c$ correspond to trivial bridge blocks. For each bounded face $f$ of $B$, create a new C-node child $c_2^{v,f}$ of $b$ corresponding to $v$ and $f$, and for any child $b'$ of $c$ corresponding to a bridge edge that is contained in $f$, make the subtree of $F$ rooted at $b'$ a child of $c_2^{v,f}$. We only add $c_2^{v,f}$ if there are any such children $b'$.

Perform these two operations on all nontrivial bridge blocks. Observe that the operations essentially apply to the children of C-nodes corresponding to cut vertices contained in a nontrivial bridge block. As nontrivial bridge blocks are not neighboring, the sets of cut vertices contained in nontrivial bridge blocks are pairwise disjoint. Hence, these operations do not interfere with each other and can be performed independently.

Then, finally, order the children of a C-node $p$ in the tree according to the order in which their edges appear around the corresponding cut vertex. The resulting tree is $L(H)$. \rev{See Figure~\ref{fig:eabb tree} for an example.}

\begin{lemma}\label{lem:bridge-block:build}
An embedding-aware bridge block tree of a plane, connected graph $H$ can be computed in polynomial time.
\end{lemma}
\begin{proof}
First, we invoke Theorem~\ref{thm:bridge-block:tarjan} to obtain a bridge block tree of $H$. Using a Doubly-Connected Edge List (DCEL) of the embedding, we can then find the information necessary to make it embedding-aware in polynomial time.
\end{proof}

We now make several observations about $L(H)$. Let $B$ and $B'$ be distinct bridge blocks of a plane graph $H$. Since $H$ is plane (and ignoring possible intersections of the embedding on the cut vertex), $B$ is enclosed by a bounded face of $B'$, or vice versa, or $B$ and $B'$ are both in each other's outer face. We now define a strict partial order $\precp$ on the bridge blocks of $H$, where $B \precp B'$ if $B$ is embedded in a bounded face of $B'$.

\begin{lemma}\label{lem:bridge-block:precp}
Let $B$ and $B'$ be two bridge blocks of a connected plane graph $H$. If $B \precp B'$, then the node $b$ corresponding to $B$ is a descendant of the node $b'$ corresponding to $B'$ in $L(H)$.
\end{lemma}
\begin{proof}
Observe that $B'$ needs to be a nontrivial bridge block. Let $F$ be the bridge block tree of $H$, rooted at $\ell(F)$. Let $p$ denote the C-node parent of $b'$ in $F$, or $p = \ell(F)$ if $b'=\ell(F)$. We claim that $b$ is a descendant of $p$ in $F$. Indeed, suppose that $b$ is not a descendant of $p$.
Consider any path in $G$ from a vertex in $B$ to a vertex bordering the outer face. Since $B \precp B'$, any such path cannot avoid a vertex of $B'$. Hence, $b$ is a descendant of $p$ in $F$. Moreover, any such path must enter $B'$ at the same cut vertex $v$, which can correspond to $p \not= \ell(F)$ or a C-node that is a child of $b'$. We only consider the case when this cut vertex corresponds to the C-node $p \not= \ell(F)$; the other case is similar. Let $P$ be a path from a vertex in $B$ through the cut vertex $v$ corresponding to $p \not=\ell(F)$ to a vertex bordering the outer face. Then the edge of $P$ preceding $v$ must be a bridge in $H$, and thus a trivial bridge block $B''$. Let $b''$ be the corresponding node of $F$. The first operation ensures that $b''$ becomes a descendant of $b'$, and thus $b$ becomes a descendant of $b'$ in $L(H)$, as claimed.
\end{proof}

\begin{lemma}\label{lem:bridge-block:cnode}
Let $H$ be a connected plane graph. Let $c$ and $c'$ be two C-nodes in $L(H)$ corresponding to the same cut vertex $v$ of $H$. Then on the rev{unique} path between $c$ and $c'$ in the \rev{underlying undirected tree of} $L(H)$, all other C-nodes correspond to $v$.
\end{lemma}
\begin{proof}
This is immediate from the construction of $L(H)$. Indeed, the C-node corresponding to a cut vertex $v$ is only replicated in $c_1^{v,\cdot}$ or $c_2^{v,\cdot}$ for a particular nontrivial bridge block $B$. As nontrivial bridge blocks are not neighboring, the sets of cut vertices contained in nontrivial bridge blocks are pairwise disjoint. Hence, there can be at most one such nontrivial bridge block $B$ that is responsible for replicating the C-node corresponding to $v$. From the construction, the property set forth in the lemma holds.
\end{proof}

\section{Basic Properties and Connectivity of the Planar Multiway Cut Dual} \label{sec:connected}

Let $G=(V, E)$ be a connected simple undirected planar graph on $n$ vertices and $m$ edges with a fixed embedding. Let $T \subseteq V$ be a set of terminals. The set of faces covering all the terminals in $T$ is denoted by $\mathcal{F} = \{F_\one: 1 \leq \one \leq k\}$. \rev{Recall that we assume such a set to be part of the input; otherwsie, we may compute it using the algorithm of Bienstock and Monma~\cite{BienstockM88}}. We call these the \emph{terminal faces} of $G$. The edges of $E$ are weighted. By removing edges of weight $0$ or less and then scaling the weights of the remaining edges, we can obtain an equivalent instance with weights specified by the function $\omega : E \rightarrow [1,\ldots,W]$ for some integer $W$.
Note that during this transformation, possibly, the set of terminal faces changes, but there will still be at most $k$ of them. Moreover, the graph might become disconnected, but we can solve the instance associated with each connected component independently. Hence, by abuse of notation, we may assume that our instance is still defined by $G$, $T$, $\mathcal{F}$, and $\omega$ as defined previously.

In the remainder, we use edge weight $\infty$ to indicate undeletable edges. Instead of $\infty$, one could use $mW+1$, but using $\infty$ simplifies later notation. Note that the initial instance has no undeletable edges, and thus has a finite-weight solution. In future transformations and reductions, we shall always maintain the property that a finite-weight solution exists. By abuse of notation, we still use $\omega : E \rightarrow [1,\ldots,W] \cup \{\infty\}$ to denote the weights.

Arbitrarily assign each terminal $t \in T$ to a face of $\mathcal{F}$ that has $t$ on its boundary.
For each face $F_\one \in \mathcal{F}$, let $T_\one \subseteq T$ be the set of terminals on the boundary of $F_\one$ that are assigned to $F_\one$. Let $p_\one = |T_\one|$. We may assume that $p_\one > 0$ for each terminal face, or we could reduce the set of terminal faces. Observe that the sets $T_\one$ form a partition of $T$. Note that $\mathcal{F}$ can be partitioned into $\mathcal{F}_1 = \{F_\one \mid p_\one = 1\}$ (called the \emph{singular faces}) and $\mathcal{F}_2 = \mathcal{F} \setminus \mathcal{F}_1$ (called the \emph{plural faces}); possibly, one of these sets is empty. 

For a terminal face $F_\one$, we order the terminals in $T_\one$ as follows. Note that $G$ is connected and thus the boundary of $F_\one$ forms a closed walk. Pick an arbitrary starting vertex on this closed walk. Now traverse the walk in clockwise direction and add a terminal to the ordering at the first moment it is encountered. Index the terminals in $T_\one$ as $\terminal{\one}{1},\ldots,\terminal{\one}{p_\one}$ according to this ordering.

\rev{We now reduce the instance to a more structured instance, extending ideas of Chen and Wu~\cite[Lemma 8]{Chen-Wu}.}

\begin{definition} \label{def:transform}
\rev{An instance $(G,T,\omega,\mathcal{F})$ is \emph{transformed} if
\begin{itemize}
\item $G$ is bridgeless;
\item the faces of $\mathcal{F}$ are vertex disjoint;
\item all vertices of each face in $\mathcal{F}$ are terminals;
\item each face of $G$ that is not in $\mathcal{F}$ has length at most $3$;
\item each face of $G$ that has length~$3$ or is in $\mathcal{F}$ neighbors only faces of length~$2$.
\end{itemize}}
\end{definition}
\rev{The final condition here is only needed to deal with a particular edge case that appears much later in the paper (Remark~\ref{rem:problem}).}

\begin{lemma} \label{lem:transformnew}
\rev{An instance $(G,T,\omega,\mathcal{F})$ can be reduced in polynomial time to an equivalent transformed instance.}
\end{lemma}
\begin{proof}
\rev{As a first step, we make sure that the faces of $\mathcal{F}$ are vertex disjoint and all vertices of each face in $\mathcal{F}$ are terminals.}
For each terminal face $F_\one \in \mathcal{F}$, add an edge of weight $1$ from $\terminal{\one}{i}$ to $\terminal{\one}{i+1}$ (indices modulo $p_\one$) for each $1 \leq i \leq p_\one$. If an edge $e$ from $\terminal{\one}{i}$ to $\terminal{\one}{i+1}$ already existed, first subdivide $e$ to obtain edges $e_1$ and $e_2$, and set the weight of $e_1$ and $e_2$ to $\omega(e)$; then add the new edge. Call the resulting graph $G'$ and the resulting weight function $\omega'$. Because $G$ is connected and thus the boundary of $F_\one$ forms a closed walk, the new edges can be embedded inside the corresponding terminal faces and thus $G'$ is planar. The embedding of $G$ can be extended to $G'$ in a natural way. In particular, a new face $F'_\one$ is created for each face $F_\one$ whose boundary consists of the vertices of $T_\one$ and the newly created edges. \rev{Recall that the sets $T_\one$ are disjoint by construction. Thus, the faces $F'_\one$ are pairwise vertex disjoint.} Hence, $G'$ and the set $\rev{\mathcal{F}' =} \{F'_\one \mid F_\one \in \mathcal{F}\}$ has \rev{the required two properties.}

To show equivalence, we observe that $(G,T,\omega)$ has a multiway cut of weight $K$ if and only if $(G',T,\omega')$ has a multiway cut of weight $K+\sum_{F_\one : p_\one > 1} p_\one$. Indeed, all newly added edges for faces with $p_\one > 1$ must be in any multiway cut of $(G',T,\omega')$. After removing those edges, the remaining graph is $G$; note that the subdivisions that were potentially performed do not affect anything.

\rev{By abuse of notation, we still denote the resulting instance by $(G,T,\omega, \mathcal{F})$. Let $n = |V(G)|$ as before. Let $\mathcal{G}$ denote the set of all faces of $G$. We now wish to triangulate $G$. To ensure that the new edges created during the triangulation do not disturb the optimal solution, we would ideally give them weight $0$. However, in our setting it is important that we work with positive weights. We also work with integral weights. Hence, we first make the original edges `heavy' instead. We proceed as follows.}

\rev{As a second step, add a copy $e'$ of each edge $e=(u,v)$ and give both $e$ and $e'$ weight $3n\cdot \omega(e)$. Embed $e'$ naturally (following the same `route' of $e$). Let $G'$ be the resulting plane graph. Observe that there is an obvious injection from the faces of $G$ to the faces of $G'$. Therefore, with some abuse, we may say that the faces of $\mathcal{G}$ also appear in $G'$.}

\rev{As a third step, triangulate each face in $\mathcal{G} \setminus \mathcal{F}$ in $G'$. Replace each new edge $f$ created during the triangulation by two new edges. Embed these naturally (following the same `route' of $f$). Let $G''$ be the resulting plane graph. Let $X$ denote the set of all newly created edges in this third step. Give each edge of $X$ weight~$1$. Let $\omega''$ be the resulting weight function of the edges of $G''$. Observe that $\omega''(X) < 6n$, because a simple planar graph can have at most $3n-6$ edges by Euler's formula.}

\rev{We first show that $(G,T,\omega)$ has a multiway cut of weight (under $\omega$) at most $K$ if and only if $(G'',T,\omega'')$ has a multiway cut of weight (under $\omega''$) strictly less than $6n(K+1)$. Suppose that $(G,T,\omega)$ has a multiway cut $C$ of weight (under $\omega$) at most $K$. Construct a set $C'' \subseteq C$ by for each edge $e \in C$ adding both $e$ and $e'$ to $C''$. Clearly, $C'' \cup X$ is a multiway cut of $(G'',T,\omega'')$ and has weight (under $\omega''$) at most $6nK + \omega''(X)$, which is strictly less than $6n(K+1)$.} 

\rev{For the converse, suppose that $(G'',T,\omega'')$ has a multiway cut $C''$ of weight (under $\omega''$) strictly less than $6n(K+1)$. We may assume that $C''$ is minimal. Then, for each edge $e \in E(G)$, both $e$ and $e'$ are in $C''$ or neither of them is. Let $C$ be the set of edges $e \in E(G)$ for which both $e$ and $e'$ are in $C''$. By the construction of the weights, each edge of $e \in C$ accounts for $6n \cdot \omega(e) \geq 6n$ weight of $C''$. As $\omega''(X) < 6n$ and $\omega''(C'') < 6n(K+1)$, it follows that $C$ has weight (under $\omega$) at most $K$. Moreover, any path in $G$ corresponds to a path in $G''$ that avoids the edges of $X$. Hence, $C$ is a multiway cut of $(G,T,\omega)$.}

\rev{Finally, we verify that $(G'',T,\omega'',\mathcal{F})$ satisfies all properties of Definition~\ref{def:transform}. $G''$ is bridgeless, because there are at least two paths between pair of adjacent vertices. The first step already ensured that the faces of $\mathcal{F}$ are vertex disjoint and all vertices of each face in $\mathcal{F}$ are terminals; this did not change during the construction of $G''$. Each face of $G''$ that is not in $\mathcal{F}$ has length at most~$3$, by the triangulation of the third step and the embedding of the newly created edges in the second and third. Finally, each face of $G$ that has length~$3$ or is in $\mathcal{F}$ neighbors only faces of length~$2$, by the duplication of edges in the second step and the added parallel triangulation edges in the third step.}
\end{proof}
We note that in a \rev{transformed instance, $G$ is not necessarily simple, but may} have parallel edges or self-loops.
In the remainder, we assume that the instance is transformed.

\subsection{Dual, Cuts, and Connectivity Properties}
Let $G^*$ be the dual of $G$. By definition, $G^*$ has an embedding in the plane such that each vertex of $G^*$ is embedded in the corresponding face and each dual edge crosses the corresponding primal edge exactly once and no other edges. For practical purposes, any time we consider a set $C^*$ of dual edges, we also denote by $C^*$ the subgraph of the dual induced by the edges in $C^*$. Then $C^*$ is again a planar graph with an embedding where each edge of $C^*$ is embedded as it is in the embedding of $G^*$. We denote by $C$ the set of edges in $G$ corresponding to the dual edges in $C^*$.

The following was observed by Dahlhaus \etal~\cite{DJPSY94}, based on the original observation of Reif~\cite{Reif83}.

\begin{proposition}\label{prp:dualfaces}
Let $C$ be a (minimum) multiway cut of $(G,T,\omega)$ and let $C^*$ be the set of corresponding dual edges. Then each face of $C^*$ encloses at most (exactly) one terminal.
\end{proposition}
We will often use this fact without explicitly referring to it. We also require the following structural properties of $C^*$. 

\begin{lemma}\label{lem:dualbridgeless}
Let $C$ be any inclusion-wise minimal multiway cut of $(G,T,\omega)$ and let $C^*$ be the set of corresponding dual edges. Then $C^*$ is bridgeless.
\end{lemma}
\begin{proof}
Let $e^*$ be a bridge in $C^*$. Then there is a face of $C^*$ for which $e^*$ appears on the boundary twice. Removing $e^*$ from $C^*$ does not change the set of vertices of $G$ enclosed by the face. Hence, $C-\{e\}$ is a \rev{multiway cut of $(G,T,\omega)$}. This contradicts that $C$ is inclusion-wise minimal.
\end{proof}

\rev{Recall that we assume throughout that the instance is transformed.}

\begin{lemma}\label{lem:alldualincident}
Let $C$ be any inclusion-wise minimal multiway cut of $(G,T,\omega)$ and let $C^*$ be the set of corresponding dual edges. Then for any dual vertex $v_\one$ corresponding to a plural terminal face $F_\one \in \mathcal{F}$, all dual edges incident to $v_\one$ are in $C^*$.
\end{lemma}
\begin{proof}
Since the instance is transformed, any edge of $F_\one$ is between a pair of distinct terminals, and thus must be in $C$.
\end{proof}

\begin{lemma}\label{lem:dualcutvertex}
Let $C$ be any multiway cut of $(G,T,\omega)$ and $C^*$ the set of corresponding dual edges. Then no terminal face of $\mathcal{F}$ corresponds to a cut-vertex of $C^*$. 
\end{lemma}
\begin{proof}
Suppose that $v_\one$ is a cut vertex of $C^*$ that corresponds to the terminal face $F_\one \in \mathcal{F}$. Let $\mathcal{B^*}$ be the set of maximal biconnected components of $C^*$ that intersect exactly in $v_\one$. Since $v_\one$ is a cut vertex, $|\mathcal{B^*}| \geq 2$.
For any maximal biconnected component $B^* \in \mathcal{B^*}$, there is a simple cycle $X_{B^*}$ of $G^*$ that determines the outer face of ${B^*}$, because ${B^*}$ is biconnected. We call this the \emph{bounding cycle} of ${B^*}$. Note that all of ${B^*}$ is enclosed by $X_{B^*}$. The planarity of $G^*$ ensures that no two bounding cycles of biconnected components in $\mathcal{{B^*}}$ can cross, and in fact, they intersect exactly in $v_\one$.

Choose a biconnected component ${B^*} \in \mathcal{{B^*}}$ for which its bounding cycle encloses the smallest region in the plane among all biconnected components in $\mathcal{{B^*}}$. Since the bounding cycles of the biconnected components of $\mathcal{{B^*}} \setminus\{{B^*}\}$ do not cross the bounding cycle of ${B^*}$ nor can they be enclosed by it (by definition of ${B^*}$), there is a biconnected component $\tilde{{B^*}}$ of $\mathcal{{B^*}} \setminus\{{B^*}\}$ in the outer face of ${B^*}$.

Consider the two edges of $X_{B^*}$ that are incident to $v_\one$. Let $e$ and $e'$ be the edges of $G$ dual to these edges. Now, let $t$ and $t'$ be the endpoints of $e$ and $e'$ that are not enclosed by $X_{B^*}$. Observe that $t$ and $t'$ are distinct terminals of $T_\one$ by the existence of $\tilde{{B^*}}$; indeed, $\tilde{{B^*}}$ is in the outer face of ${B^*}$ and encloses at least one terminal of $T_\one$ as the instance is transformed.

We claim that $t$ and $t'$ are in the same face of $C^*$, contradicting that $C$ is a multiway cut. Since no two bounding cycles of blocks of $\mathcal{B}^*$ can cross each other, if there were a face of $C^*$ enclosing $t$ but not $t'$, its boundary would contain at least one edge of $X_{B^*}$, namely the one dual to $e$. This, however, contradicts that $X_{B^*}$ is a bounding cycle.
\end{proof}

\rev{Using that the instance is transformed, we can observe the following fact.}

\begin{corollary}
\rev{Let $C$ be any multiway cut of $(G,T,\omega)$ and $C^*$ the set of corresponding dual edges. Then $C^*$ has no cut vertices.}
\end{corollary}
\begin{proof}
\rev{Suppose that $C^*$ has a cut vertex $v$. As $C^*$ is bridgeless by Lemma~\ref{lem:dualbridgeless}, $v$ must have degree at least~$4$ in $C^*$. As the instance is transformed, only dual vertices corresponding to faces in $\mathcal{F}$ can have degree more than~$3$. But this is not possible by Lemma~\ref{lem:dualcutvertex}. Hence, $v$ is not a cut vertex of $C^*$.}
\end{proof}

\subsection{Reduction to Connected Duals}\label{sec:connected-duals}
We \rev{show that, using $2^{\OO(k)}$ time}, we can reduce to the case where we may assume that the graph induced by the set of edges dual to the edges of any minimum multiway cut is connected. Intuitively, we want to focus on connected components of $C^*$ that are `innermost' in the embedding: none of its bounded faces encloses another connected component of $C^*$. We aim to solve a simplified instance for such a connected component separately (by an algorithm we describe in subsequent sections) and then solve the remaining instance recursively. To make the formalities work, we need an expansive definition of `innermost component', \rev{which we describe in a moment}.

\rev{Our approach extends the work of Klein and Marx~\cite[Lemma 3.1]{KM12}. The essence of both our approach and theirs is to consider the biconnected components of the dual of a solution. However, we consider them jointly, to build an idea of the `innermost' components of the dual of a solution. Klein and Marx instead treat them individually and branch on the subset of terminals covered by one. Then in their sub-instances, for a subset $X$ of terminals, they not only need to pairwise separate the terminals of $X$, but also need to separate $X$ from $T \setminus X$, \rev{leading to complex sub-instances}. We use the planarity of the solution to effectively argue that the latter condition is not necessary. Hence, our sub-instances are `normal' instances of {\sc Planar Multiway Cut}, without further constraints. This requires more involved topological tools and an involved defintion of ``innermost''. Regardless, we need substantially more effort to deal with plural faces.}

\subsubsection{Internal Sets}
We formalize the intuition of an `innermost component' of a plane graph with the following notion. \rev{This notion is more general than we need in this section, but it will be useful later to define it in this manner.}

\begin{definition}
Let $H$ be any plane graph. Let $J$ be the union of a subset of the biconnected components of $H$ \rev{and let $\zeta(H,J)$ denote the subgraph of $H$ formed from $H$ by removing all vertices and edges that are in $J$ but not in $H$. Then $J$ is an \defi{internal set} if there is a single face $f$ of $\zeta(H,J)$ that encloses $J$ and there is a single face $f'$ of $J$ that encloses $\zeta(H,J)$. We call $f$ the \defi{enclosing region} and $f'$ the \defi{exclosing region}.} The intersection of the enclosing region and exclosing region is called the \defi{middle region}.
\end{definition}

\begin{figure}[t]
    \centering
    \includegraphics[width=0.75\textwidth]{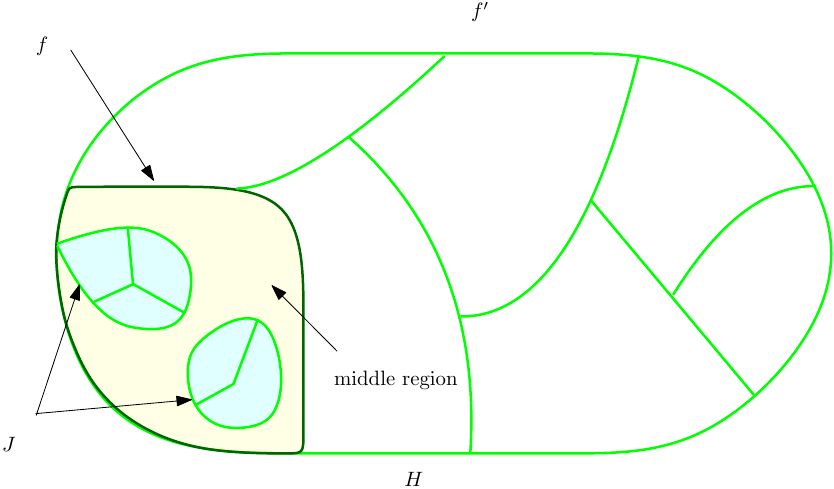}
    \caption{\rev{An example of an internal set. The plane graph $H$ is drawn in green and dark green. Two biconnected components, with faces shaded in blue, are indicated as $J$. Note that $J$ shares a cut vertex with $\zeta(H,J)$. The region of the plane enclosed by $f$ (the dark green cycle) is the enclosing region. The exclosing region $f'$ is the plane minus the two regions shaded in blue. The middle region is the region shaded in yellow.}}\label{fig:disconnectedsol}
\end{figure}

Refer to Figure~\ref{fig:disconnectedsol} for an illustration of the above definition.

Observe that neither the enclosing region nor the exclosing region of \rev{an internal set} is necessarily homeomorphic to a disk, nor are they necessarily equal. Also, the middle region is a face of $H$. In particular, \rev{either the enclosing or the exclosing region} is the outer face of $J$ or \rev{$\zeta(H,J)$}. Crucially, if \rev{$H$} is not connected, then any connected component of $H$ for which none of its bounded faces encloses another connected component of $H$ is an internal set. This corresponds to our earlier intuition that `innermost' connected components are internal sets (but the definition of internal sets is much more general).

\begin{remark}
By abuse of terminology, $H$ itself can be viewed as an internal set, with the complement of the outer face of $H$ as its enclosing region and the outer face as its exclosing and middle region. In particular, when $H = C^*$, where $C^*$ is the set of dual edges corresponding to \rev{some} multiway cut of $(G,T,\omega)$, the highly relevant Lemma~\ref{lem:notwoboundedfaces} below holds when using $C^*$ as the internal set.
\end{remark}

We now prove a useful property of internal sets.

\begin{proposition}\label{prp:internalcomp}
Let $J$ be an internal set of a plane graph $H$. Then \rev{$\zeta(H,J)$} is also an internal set.
\end{proposition}
\begin{proof}
Since $J$ is the union of a set of biconnected components, so is its complement \rev{$\zeta(H,J)$}. The enclosing region of \rev{$\zeta(H,J)$} is the exclosing region of $J$ and the exclosing region of \rev{$\zeta(H,J)$} is the enclosing region of $J$. 
\end{proof}

\rev{Finally, we recall that, for a plane graph $H$ and a set $J \subseteq E(H)$, we may abuse notation and use $J$ to denote the subgraph of $H$ induced by the edges of $J$ as well. In that case, for $J, J' \subseteq E(H)$ that both form a union of biconnected components of $H$, the graph $J-J'$ (the subgraph of $H$ induced by the edges in the set $J-J'$) is equivalent to the subgraph $\zeta(J,J')$.}

\subsubsection{Internal Sets and Multiway Cuts}
We now prove a useful property of internal sets of a minimal multiway cut.

\begin{definition}
Let $C$ be any multiway cut of $(G,T,\omega)$ and $C^*$ the set of corresponding dual edges. Let $D^*$ be an internal set of $C^*$.
We say that a point in the plane is \defi{covered} by an internal set $D^*$ if it is enclosed by the complement of the exclosing region of $D^*$. 
\end{definition}

In our intuition of internal sets being `innermost' connected components, a covered point is enclosed by a bounded face of the component.

\begin{lemma}\label{lem:notwoboundedfaces}
Let $C$ be any \rev{minimum} multiway cut of $(G,T,\omega)$ and $C^*$ the set of corresponding dual edges. Let $D^*$ be an internal set of $C^*$ and $D$ the set of corresponding edges in $G$. Let $F_\one \in \mathcal{F}$ be any terminal face for which a terminal of $T_\one$ is covered by $D^*$. Then $C-D$ does not contain any edge of $F_\one$ and at least $\max\{1,p_\one-1\}$ terminals of $T_\one$ are covered by $D^*$. Moreover, there is at most one terminal face $F_\one$ for which this number is $p_\one-1$.
\end{lemma}
\begin{proof}
When $p_\one=1$, trivially all terminals of $T_\one$ are covered by $D^*$. Moreover, $C$ does not contain the edge of $F_\one$, because the corresponding dual edge is a bridge in $G^*$ and $C^*$ is bridgeless by Lemma~\ref{lem:dualbridgeless}.

So assume that $p_\one > 1$. Let $v_\one$ denote the dual vertex corresponding to $F_\one$. Following Lemma~\ref{lem:alldualincident}, $C^*$ contains all (at least two) dual edges incident to $v_\one$. As a terminal of $T_\one$ is covered by $D^*$, there is an edge of $D^*$ incident to $v_\one$. By Lemma~\ref{lem:dualcutvertex}, $v_\one$ is not a cut vertex of $C^*$ and thus the definition of an internal set implies that $D^*$ contains $v_\one$ and all its incident dual edges. Hence, $C^*-D^*$ does not contain a dual edge incident to $v_\one$ and thus, $C-D$ does not contain an edge of $F_\one$.

Now let $t$ and $t'$ be two distinct terminals of $T_\one$ not covered by $D^*$. Then both $t$ and $t'$ are in the exclosing region $D^*$. As no dual edges incident to $v_\one$ belong to $C^*-D^*$ and $D^*$ is a union of biconnected components, $t$ and $t'$ are in the enclosing region of $D^*$. Hence, they are in the same face of $C^*$, which contradicts that $C$ is a multiway cut. Hence, at least $p_\one-1$ terminals of $T_\one$ are covered by $D^*$.

A similar contradiction holds for any two distinct terminals on distinct terminal faces which are not covered by $D^*$. Hence, there is at most one terminal face of which a terminal is covered by $D^*$ and a terminal is not fully covered by $D^*$.
\end{proof}

The intuition of our approach is now as follows. If a terminal of $T_\one$ for a terminal face $F_\one \in \mathcal{F}$ is covered by $D^*$, then it follows from Lemma~\ref{lem:notwoboundedfaces} that almost all terminals of $T_\one$ are covered by $D^*$. We then say that this component \defi{covers} this terminal face. Then we could guess (by exhaustive enumeration) in $2^k$ time the subset $\mathcal{F}_{D^*}$ of terminal faces covered by $D^*$ and obtain our simplified instance.

However, there is an important exception to this intuition, namely the unique terminal $t$ in the middle region of $D^*$. This terminal $t$ might not be on a face of $\mathcal{F}_{D^*}$ \rev{or is on a face of $\mathcal{F}_{D^*}$ but not covered by $D^*$} (the possible existence of such a terminal is hinted at by Lemma~\ref{lem:notwoboundedfaces}). Knowing this terminal is important to the algorithm, but guessing (by exhaustive enumeration) this terminal could require $O(n)$ guesses. Since $C^*$ has at most $k$ terminal faces, this would lead to an undesirable $O(n^k)$ running time overall. We show that knowing the precise terminal $t$ is actually unnecessary. By some small modifications to the weights, we argue that we can pick an arbitrary terminal of $T_\one$ that functions as a representative for $t$, avoiding the $O(n)$ guesses altogether.

\begin{lemma}\label{lem:disentangleA}
Let $C$ be any inclusion-wise minimal multiway cut of $(G,T,\omega)$ and $C^*$ the set of corresponding dual edges. Let $D^*$ be an internal set. Let $T'$ be the set of terminals covered by $D^*$ and let $t$ be a terminal in the middle region of $D^*$. Let $F_\one \in \mathcal{F}$ be the unique terminal face that contains $t$. If $F_\one$ is a plural face and $T' \cap T_\one = \emptyset$ (none of the terminals of $T_\one$ are covered by $D^*$), let $\omega'$ be obtained from $\omega$ by setting the weight of each edge of $F_\one$ to $\infty$ and let $t'$ be any terminal of $F_\one$. Otherwise, let $t'=t$ and $\omega'=\omega$. Then \rev{for the triple $(G,T'\cup\{t'\},\omega')$, it holds that}:
\begin{enumerate}[label= \emph{\alph*.}]
\item\label{lem:disentangleA:a} for any finite-weight \rev{multiway cut} $B$ of $(G,T'\cup\{t'\},\omega')$, $B \cup (C-D)$ is a \rev{multiway cut} for $(G,T,\omega)$;
\item\label{lem:disentangleA:b} $D$ is a \rev{minimum multiway cut} for $(G,T'\cup\{t'\},\omega')$, and thus $(G,T'\cup\{t'\},\omega')$ has a finite-weight \rev{multiway cut};
\item\label{lem:disentangleA:c} for any \rev{minimum multiway cut} $B$ for $(G,T'\cup\{t'\},\omega')$, $B^*$ is an internal set of $B^* \cup (C^*-D^*)$.
\end{enumerate}
\end{lemma}

\begin{proof}
 
\begin{enumerate}[label = \alph*.]
\item Let $B$ be a finite-weight \rev{multiway cut} of $(G,T'\cup\{t'\},\omega')$. We claim that $B \cup (C-D)$ is a \rev{multiway cut} for $(G,T,\omega)$. Let $B^*$ be the set of dual edges corresponding to $B$. Since $B$ has finite weight, the construction of $\omega'$ and the definition of $t'$ ensures that either $t=t'$ or the face of $B^*$ that encloses $t'$ also encloses $T_\one$. Hence, in either case there is a face of $B^*$ that encloses $t$ and no other terminal of $T'$.

Consider the enclosing region $f$ of $D^*$. Every terminal of $T' \cup \{t\}$ is enclosed by $f$.
Then the intersections of each face of $B^*$ with $f$ yields a set of regions in the plane that each contain at most one terminal of $T' \cup \{t\}$. Since each terminal of $T-T'-\{t\}$ is contained in a face of $C^*-D^*$ enclosed by the complement of $f$, it follows that each face of $B^* \cup (C^*-D^*)$ contains at most one terminal (note that any edge of $B^*$ not enclosed by $f$ only partitions the faces of $C^*-D^*$ and is not relevant to the feasibility of $B \cup (C-D)$).

\item We now prove that $D$ is a \rev{minimum multiway cut} for $(G,T'\cup\{t'\},\omega')$. We first argue that $D$ is a \rev{multiway cut} for $(G,T'\cup\{t'\},\omega')$ of finite weight. If $p_\one = 1$, then $F_\one$ consists of a single loop and the optimality of $C$ ensures that $C$ (and thus $D$) does not contain this loop. $D$ has finite weight by construction. Moreover, $t'=t$ is in the exclosing region of $D^*$ by definition, and thus, $D$ is a \rev{multiway cut}.

Otherwise, if $p_\one > 1$, consider two cases. If $D^*$ covers a terminal of $T_\one$, then $\omega' = \omega$ and $D$ trivially has finite weight. Moreover, $t'=t$ is in the exclosing region of $D^*$ by definition, $D$ is a \rev{multiway cut}. If none of the terminals of $T_\one$ are covered by $D^*$, then all terminals of $T_\one$ are in the exclosing region of $D^*$. Then $D^*$ does not contain the dual vertex $v_\one$ (and therefore, no edge incident to $v_\one$). Thus, $D^*$ has finite weight. As $T_\one$ is in the exclosing region of $D^*$, so is $t'$. Since all terminals of $T'$ are not in the exclosing region of $D^*$, $D$ is a \rev{multiway cut}.

To complete the argument, let $B$ be a \rev{minimum multiway cut} for $(G,T'\cup\{t'\},\omega')$ of weight strictly smaller than $D$. By the feasibility of $D$, $B$ must have finite weight. Moreover, by the preceding, $B\cup (C-D)$ is a \rev{multiway cut} for $(G,T,\omega)$. Then it has smaller weight (with respect to $\omega$) than $C$, a contradiction. Hence, $D$ is an a \rev{minimum multiway cut} for $(G,T'\cup\{t'\},\omega)$.

\item Next, suppose that $B$ is a \rev{minimum multiway cut} for $(G,T'\cup\{t'\},\omega')$. We prove that $B^*$ is an internal set with respect to $A^*$, where $A^*$ is the set of dual edges corresponding to $A = B \cup (C-D)$. Note that $A$ is a \rev{minimum multiway cut} for $(G,T,\omega)$ by the preceding. We now argue that $B^*$ is enclosed by the enclosing region $f$ of $D^*$ and only possibly shares vertices with $f$. Suppose that $B^*$ is not enclosed by the enclosing region $f$ of $D^*$ or shares edges with the boundary of $f$. As mentioned earlier in the proof, dual edges of $B^*$ that are not enclosed by $f$ are irrelevant to the feasibility of $A$. On the other hand, shared dual edges are effectively counted twice in $A$. Hence, let $X^*$ be the set of dual edges in $B^*$ that are either not enclosed by $f$ or shared with the boundary of $f$. Let $X$ be the corresponding set of edges in $G$. Note that $X^* \not= \emptyset$ by assumption and thus $\omega(X) > 0$. We have just argued that $(B-X) \cup (C-D)$ is still a \rev{multiway cut}. Moreover, $\omega(B) = \omega(D)$ and thus $\omega(B-X) < \omega(D)$. Then $\omega((B-X) \cup (C-D)) < \omega(D \cup (C-D)) = \omega(C)$, a contradiction to the optimality of $C$. Hence, $B^*$ is enclosed by $f$ and possibly shares only vertices with $f$. Hence, we can define $f$ as the enclosing region of $B^*$ and there is a single face of $B^*$ that encloses $A^*-B^*$ that is the exclosing region.

Finally, suppose that for this \rev{minimum multiway cut} $B$ for $(G,T'\cup\{t'\},\omega')$, $B^*$ is not a union of biconnected components of $A^*$. Since $B^*$ itself is trivially a union of biconnected components, this means that the intersection of the enclosing region and exclosing region of $B^*$ with respect to $A^*$ is in fact a union of at least two regions, and not a single middle region. Each of these regions is a face of $A^*$. Since at most one of these regions (faces) can contain a terminal by Lemma~\ref{lem:notwoboundedfaces}, $A$ is not an \rev{minimum multiway cut}, a contradiction. 
\end{enumerate}
\end{proof}

\subsubsection{Algorithm}
We are now ready to develop the algorithm. Let $\mathcal{A}\rev{(G,T,\omega,\mathcal{F})}$ be any algorithm for {\sc Planar Multiway Cut} that always outputs a \rev{multiway cut of $(G,T,\omega)$}, but is only guaranteed to find a \rev{minimum multiway cut} if for \emph{all} \rev{minimum multiway cuts} $C$ it holds that $C^*$ is connected. We show in later sections that we can find such an algorithm $\mathcal{A}$ with the claimed running time bounds. Using it as a black box for now, we can give a recursive algorithm for {\sc Planar Multiway Cut}.

Let $(G,T,\omega, \rev{\mathcal{F}})$ be a transformed instance of {\sc Planar Multiway Cut} with $\mathcal{F}$ a set of $k$ terminal faces. Recall that $\mathcal{F}$ can be partitioned into the set $\mathcal{F}_1$ of singular faces and the set $\mathcal{F}_2$ of plural faces; possibly, one of these sets is empty.

In the first phase of the algorithm, consider two new types of sub-instances. For the first type, if $\mathcal{F}_2 \not=\emptyset$, consider each $F_\two \in \mathcal{F}_2$ and each set $\tilde{T}$ that is the union of the sets $\{T_\one \mid F_\one \in \tilde{\mathcal{F}}\}$ for a subset $\tilde{\mathcal{F}}$ of $\mathcal{F} \setminus \{F_\two\}$, and solve recursively on the transformed version of the sub-instance $(G,T'',\omega'',\rev{\tilde{\mathcal{F}} \cup \{F_\two\}})$, where $T'' = \tilde{T} \cup \{\terminal{\two}{1}\}$, and $\omega''$ is equal to $\omega$, except it is set to $\infty$ for all edges of $F_\two$.
For the second type, consider each set $T''$ that is the union of the sets $\{T_\one \mid F_\one \in \tilde{\mathcal{F}}\}$ for a strict subset $\tilde{\mathcal{F}}$ of $\mathcal{F}$, and solve recursively on the transformed version of the sub-instance $(G,T'',\omega'',\rev{\tilde{\mathcal{F}}})$, where $\omega'' = \omega$.

Let $B$ be the \rev{multiway cut} that is recursively found for any such sub-instance $(G,T'',\omega'',\rev{\mathcal{F}''})$ and let $B^*$ denote the set of corresponding dual edges. In the second phase of the algorithm, let $T'$ be the set of terminals in the outer face of $B^*$. If there is a \rev{single} terminal face $F_\one$ for which $\emptyset \subset T'\cap T_\one \subset T_\one$, then consider the sub-instance $(G,(T' \setminus T_\one) \cup \{\terminal{\one}{1}\},\omega', (\mathcal{F} \setminus \mathcal{F}'') \cup \{F_\one\})$, where $\omega'$ is obtained from $\omega''$ by setting the weight of each edge of $F_\one$ to $\infty$. Otherwise, consider the sub-instance $(G,T',\omega,\mathcal{F} \setminus \mathcal{F}'')$. Let $Z_B$ be the \rev{multiway cut} that is recursively found for such a sub-instance. Among all such combinations \rev{$B \cup Z_B$} that are \rev{a multiway cut for $(G,T,\omega)$}, \rev{and the multiway cut found by} $\mathcal{A}(G,T,\omega,\rev{\mathcal{F}})$, return \rev{one} of minimum weight.

This finishes the description of the algorithm. We prove the following result.

\begin{lemma}\label{lem:reduce-to-connected}
Let $\mathcal{A}\rev{(G,T,\omega,\mathcal{F})}$ be any algorithm for {\sc Planar Multiway Cut} with a given set $\rev{\mathcal{F}}$ of $k$ terminal faces that always outputs a \rev{multiway cut of $(G,T,\omega)$}, but is only guaranteed to find an \rev{minimum multiway cut} if for all minimum multiway cuts $C$ to the instance it holds that $C^*$ is connected. Let $T(n,k)$ be a monotone function that describes the running time of $\mathcal{A}$ on instances with $n$ vertices and $k$ terminal faces. Then {\sc Planar Multiway Cut} can be solved in $O(2^{\OO(k)} T(n,k) \mathrm{poly}(n))$ time.
\end{lemma}
\begin{proof}
\rev{Let $(G,T,\omega, \rev{\mathcal{F}})$ be a transformed instance of {\sc Planar Multiway Cut} with $\mathcal{F}$ the set of $k$ terminal faces.} We prove that the above algorithm solves a given instance of {\sc Planar Multiway Cut} in the claimed running time. Note that the described algorithm always returns a \rev{multiway cut of $(G,T,\omega)$}, because $\mathcal{A}$ \rev{does}. To prove that the algorithm returns a \rev{minimum multiway cut}, let $C^*$ be a \rev{minimum multiway cut} of $(G,T,\omega)$ with the maximum number of connected components. If $C^*$ is connected, then $\mathcal{A}$ will find a \rev{minimum multiway cut}. Otherwise, let $D^*$ be a connected component of $C^*$ for which none of its bounded faces encloses another connected component of $C^*$. Then $D^*$ is an internal set.

We observe that the \rev{triple constructed by} Lemma~\ref{lem:disentangleA} with respect to $D^*$ and $C^*$ \rev{corresponds to} one of the sub-instances considered in the first phase of the algorithm.
Indeed, we note that $D^*$ is merely a connected component of $C^*$ (in particular, not equal to $C^*$) and by Lemma~\ref{lem:notwoboundedfaces}, there is at most one terminal $t$ of the terminal faces covered by $D^*$ such that $t$ is not covered by $D^*$. Hence, the set of terminal faces covered by $D^*$ is a strict subset of $\mathcal{F}$.
If $D^*$ does not cover all terminals of each terminal face covered by $D^*$, then $t$ exists and is the terminal in the middle region of $D^*$. \rev{In this case, the terminal set of the triple constructed by Lemma~\ref{lem:disentangleA} for $D^*$} simply consists of the terminals on the subset of terminal faces covered by $D^*$. \rev{The corresponding set of terminal faces} is a strict subset of $\mathcal{F}$, and thus this triple \rev{corresponds to} a sub-instance \rev{considered} in the second type of instances of phase one of the algorithm.
If $D^*$ does cover all terminals of each terminal face covered by $D^*$, then the terminal $t$ in the middle region of $D^*$ belongs to a terminal face \rev{$F_\one$} not covered by $D^*$. If \rev{$F_\one$} is a singular face, then there is another terminal face not covered by $D^*$ (in the part of the exclosing region of $D^*$ that is not the middle region), \rev{and thus the set of terminal faces covered by $D^*$ is a strict subset of $\mathcal{F}$. Thus}, the \rev{triple constructed by} Lemma~\ref{lem:disentangleA} \rev{for $D^*$ corresponds to a sub-instance considered} in the second type of instances of phase one of the algorithm. If \rev{$F_\one$} is a plural face, then the \rev{triple constructed by} Lemma~\ref{lem:disentangleA} \rev{for $D^*$ corresponds to a sub-instance considered} in the first type of instances of phase one of the algorithm.

Consider the \rev{minimum multiway cut} $B$ found by the algorithm for the \rev{sub-instance that corresponds to the triple} constructed by Lemma~\ref{lem:disentangleA} with respect to $D^*$ and $C^*$. Let $A = B \cup (C-D)$. Let $B^*$ and $A^*$ denote the sets of corresponding dual edges. \rev{By Lemma~\ref{lem:disentangleA}\ref{lem:disentangleA:b}, this sub-instance has a finite-weight multiway cut, and thus $\omega(B) = \omega(D)$. By Lemma~\ref{lem:disentangleA}\ref{lem:disentangleA:a}, $A$ is a \rev{minimum multiway cut} of $(G,T,\omega)$. By Lemma~\ref{lem:disentangleA}\ref{lem:disentangleA:c}, $B^*$ is an internal set of $A^*$}. Hence, $C^*-D^*$ is an internal set of $A^*$ by Proposition~\ref{prp:internalcomp}. \rev{This fact, combined with the fact that $A$ is a multiway cut  of $(G,T,\omega)$, implies that there is at most one terminal $t$ of the terminal faces covered by $C^*-D^*$ that is not covered by $B^*$. In particular, there is at most one terminal face $F_\one$ for which $\emptyset \subset T' \cap T_\one \cap T_\one$, where $T'$ is the set of terminals in the outer face of $B^*$.} Repeating the same arguments as in the previous paragraph, we see that the \rev{triple constructed by} Lemma~\ref{lem:disentangleA} with respect to $C^*-D^*$ and $A^*$ \rev{corresponds to a sub-instance} considered in phase two of the algorithm. 

Consider the \rev{minimum multiway cut} $Z_B$ found by the algorithm for the \rev{sub-instance that corresponds to the triple} constructed by Lemma~\ref{lem:disentangleA} with respect to $C^*-D^*$ and $A^*$. \rev{By Lemma~\ref{lem:disentangleA}\ref{lem:disentangleA:b},\ref{lem:disentangleA:a}}, $\omega(Z_B) = \omega(C-D)$ and thus $Z_B \cup B$ is a \rev{minimum multiway cut} for $(G,T,\omega)$. Hence, the algorithm will return a \rev{minimum multiway cut for $(G,T,\omega)$}.
Here we note that if the sub-instance is of the first type, then $F_\one = F_\two$ in the second phase by the fact that $B^*$ is an internal set.

Finally, to see the running time, note that each sub-instance is formed by a set of terminals $T_1 \cup T_2$, where $T_1$ is the union of the terminals of a subset of the terminal faces of $\mathcal{F}$ and $T_2$ is the set of first terminals of a subset $\mathcal{F}'$ of the plural faces of $\mathcal{F}$, such that these subsets of terminal faces are disjoint. The weight of all edges of the terminal faces in $\mathcal{F}'$ is set to $\infty$ and then the instance is transformed. It follows that there are at most $3^k$ distinct sub-instances ever considered by the algorithm. We also observe that for each sub-instance, compared to the input instance, the number of terminal faces decreases or else a plural face becomes singular. Hence, instead of the above top-down recursive algorithm, we can apply a bottom-up dynamic programming algorithm. Here we look at the instances in order of increasing total number of terminal faces plus number of plural faces. To fill each table entry, we need to consider at most $(k+1) 2^k$ table entries in the first phase and exactly one in the second phase. For each combination, we need polynomial time for feasibility check and to run $\mathcal{A}$. Hence, the total running time is as claimed.
\end{proof}

Following Lemma~\ref{lem:reduce-to-connected}, we need to develop the required algorithm $\mathcal{A}$. Hence, in the remainder, we assume that the dual of any optimum solution to our instance of {\sc Planar Multiway Cut} is connected.

\section{Augmented Planar Multiway Cut Dual: Skeleton and Nerves}\label{sec:Skeleton and Nerves}
\rev{As before, we are given an instance $(G,T,\omega,\mathcal{F})$. Recall that we assume that the instance is transformed and that the dual of any optimum solution is connected.}
An important challenge for our algorithm is to find the part of an optimum solution that cuts the many terminal pairs incident to a single terminal face (of course, in concert with cutting terminal pairs on other faces). In the case that $k=1$, Chen and Wu~\cite{Chen-Wu} proposed the notion of an augmented dual in order to reduce to an instance of {\sc Planar Steiner Tree} wherein all the terminals lie on a single face. This, in turn, can be solved in polynomial time~\cite{Erickson,Bern}. Such a simple reduction does not work in our case, but we do borrow and extend the notion of an augmented dual.

\begin{definition}
The \defi{augmented dual} graph $G^+$ of $G$ with respect to $\mathcal{F}$ is constructed as follows. Starting with the planar dual of $G$, for each face $F_{\one} \in \mathcal{F}$, replace the corresponding dual vertex $v_\one$ by a set of vertices $T^+_{\one}= \{\augvertex{\one}{1}, \augvertex{\one}{2}, \ldots, \augvertex{\one}{p_\one}\}$ (called the \defi{augmented terminals}), \rev{each being an end point of some edge that was incident} to $v_\one$. For $1 \leq j \leq p_\one$, each $\augvertex{\one}{j}$ and its incident edge represents the edge $(^\one t_{j-1}, ^\one t_{j})$ on the boundary of $F_{\one}$. The weight function on the edges of $G^+$ is denoted by $\omega^+$; all the edges in the augmented dual graph have the same weight as their corresponding edge in $G$ under $\omega$.
\end{definition}

\begin{figure}[t!]
    \centering
    \includegraphics[width= 0.75\textwidth]{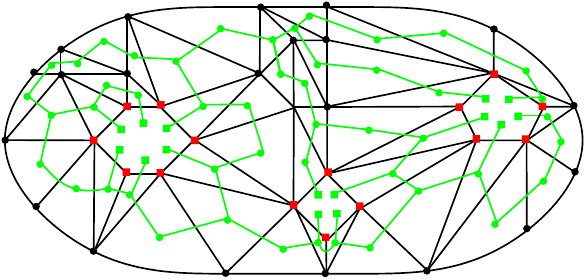}
    \caption{The figure shows the graph $G$ in black (\rev{for the sake of this illustration, we do not show parallel edges}). The red blocks represent terminals, while the green blocks depict the augmented vertices. The green dots represent the dual vertices for each non-terminal face in the augmented dual. The green curves are the edges of $C^+$.}\label{fig:Solution dual}
\end{figure}

Observe that there is still a one-to-one correspondence between edges of $G$ and edges of \rev{of the augmented dual} $G^+$, as there is between edges of $G$ and edges of \rev{the dual} $G^*$. \rev{We call the edges of the $G^*$ \emph{augmented dual edges} and the vertices of $G^*$ \emph{augmented dual vertices}.}
We also note that an alternative construction of $G^+$ would be to subdivide each dual edge of $G^*$ incident to the dual vertex $v_\one$ corresponding to a terminal face $F_\one \in \mathcal{F}$ and subsequently removing $v_\one$.

\begin{lemma} \label{lem:augmenteddualdegree}
\rev{Every augmented dual vertex has degree at most~$3$.}
\end{lemma}
\begin{proof}
\rev{Since the instance is transformed, each face of $G$ that is not in $\mathcal{F}$ has length at most~$3$. By construction, each of the augmented terminals has degree~$1$. Hence, each augmented dual vertex has degree at most~$3$.}
\end{proof}

\rev{The set of edges of $G^+$ corresponding to a minimum multiway cut $C$ of $(G,T,\omega)$ is denoted by $C^+$.} For practical purposes, any time we consider a set $C^+$ of edges of the augmented dual, we also denote by $C^+$ the subgraph of the dual induced by the edges in $C^+$. Figure~\ref{fig:Solution dual} illustrates the structure of $C^+$.

\begin{lemma}\label{lem:augmentedconnected}
Let $C$ be any minimum multiway cut of $(G,T,\omega)$ and $C^+$ the set of corresponding edges of the augmented dual.
Then $C^+$ is a connected planar graph with $k$ faces, one per terminal face in $\mathcal{F}$. 
\end{lemma}
\begin{proof}
Let $C^*$ be the set of corresponding dual edges. By assumption and Proposition~\ref{prp:dualfaces}, we know that $C^*$ is a connected subgraph of the dual, each face of which encloses exactly one terminal.

First, we show that $C^+$ is connected. Let $v_\one \in V(C^*)$ be the dual vertex representing the terminal face $F_\one$, for any $F_\alpha \in \mathcal{F}$. Suppose that there exists a path $P$ in $C^*$ from vertex $a^*$ to $b^*$, such that $P$ passes through $v_\one$. Let $x^*$ and $y^*$ be the vertices on $P$ adjacent to $v_\one$. Note that in $C^+$, $x^*$ and $y^*$ are the sole neighbors of two of the augmented terminals of $F_\one$. Without loss of generality, let these augmented terminals be $\augvertex{\one}{j}$ and $\augvertex{\one}{k}$. \rev{Since $v_{\one}$ lies at the intersection of the cycles in $C^*$ enclosing each terminal on $F_{\one}$, there still exists a path from $\augvertex{\one}{j}$ and $\augvertex{\one}{k}$ in $C^+$. Moreover, this path does not use the dual vertex of any other terminal face.} To go from $a^*$ to $b^*$, we can follow the path from $a^*$ to $\augvertex{\one}{j}$, then go along the path from $\augvertex{\one}{j}$ to $\augvertex{\one}{k}$ and finally from $\augvertex{\one}{k}$ to $b^*$. Since \rev{$a^*$, $b^*$, and $v_\one$} were arbitrary, this holds for any path passing through the dual vertex of any terminal face. Since no modification was made to any other path of $C^*$, this proves that $C^+$ is connected.

Next, we show that $C^+$ contains exactly $k$ faces, each enclosing a terminal face of $\mathcal{F}$. For all \rev{$F_\one \in \mathcal{F}$}, observe that its dual vertex $v_\one$ lies at the intersection of face boundaries enclosing each terminal on the boundary of $F_\one$. None of these face boundaries intersect a dual vertex corresponding to any face $F_\two$ , for $\two \neq \one$, or else $C^*$ would enclose more than one terminal in that face, contradicting its feasibility. After splitting $v_\one$ into augmented terminals, by the preceding paragraph, we get a connected graph. In particular, each pair of consecutive augmented terminals is connected in $C^+$ by a path. The union of these paths bounds a region in the plane enclosing the face $F_\one$. Therefore, $C^+$ has exactly $k$ faces.
\end{proof}
 
\subsection{Skeleton}
Considering Lemma~\ref{lem:augmentedconnected}, we note in particular that $C^+$ has $k$ faces. However, each of these $k$ faces can be highly complex. By `zooming out', we can show that $C^+$ in fact has a very nice global structure with $k$ faces that are much easier to describe. To this end, we prove the following.

We define the action of \emph{dissolving} a vertex $v$ of degree~$2$ with \emph{distinct} neighbors $u$ and $w$, as removing $v$ from the graph and adding an edge from $u$ to $w$. If we drop the constraint that $u,w$ are distinct, we speak of \emph{strongly dissolving}. We explicitly mention that while dissolving a vertex we retain any parallel edges that may arise. Moreover, no self-loops may arise while dissolving a vertex, while this is possible when \rev{strongly} dissolving a vertex.

\begin{definition}
Let $C$ be any multiway cut of $(G,T,\omega)$ and $C^+$ the set of corresponding edges of the augmented dual. \rev{Then the subgraph of $C^+$ that remains after exhaustively removing all the edges incident to a vertex of degree~$1$ of $C^+$ is called the \defi{skeleton} of $C^+$, denoted by $S^+$. The graph that remains of $S^+$ after dissolving all vertices of degree~$2$ from $S^+$ is called the \defi{shrunken skeleton} of $C^+$.}
\end{definition}

\begin{lemma}\label{lem:skeleton}
Let $C$ be any minimum multiway cut of $(G,T,\omega)$ and $C^+$ the set of corresponding edges of the augmented dual.
\rev{Then the skeleton and shrunken skeleton of $C^+$ have minimum degree~$2$ and maximum degree~$3$.}
\rev{Moreover,} the shrunken skeleton of $C^+$ is a connected planar multi-graph without self-loops and with $k$ faces, such that each of its faces strictly encloses exactly one terminal face of $\mathcal{F}$.
\end{lemma}
\begin{proof}
\rev{By Lemma~\ref{lem:augmenteddualdegree}, each augmented dual vertex of $C^+$ has maximum degree~$3$. The same holds for the skeleton and shrunken skeleton of $C^+$ by construction. The construction of the skeleton ensures that both the skeleton and shrunken skeleton have minimum degree~$2$.}

Following Lemma~\ref{lem:augmentedconnected}, $C^+$ has $k$ faces.
Recall that for any bridge of a planar graph, the faces on both sides of the bridge are the same. Hence, removing a bridge does not decrease the number of faces. It follows immediately that the skeleton $S^+$ has $k$ faces, one per terminal face in $\mathcal{F}$. Since each augmented terminal has degree one, its incident edge is a bridge of $G^+$. Hence, each face of $S^+$ strictly encloses exactly one terminal face of $\mathcal{F}$. Since $C^+$ is connected by Lemma~\ref{lem:augmentedconnected}, by construction, $S^+$ is also connected. Note that, while dissolving vertices of degree~$2$, we can maintain essentially the same embedding as of $S^+$. Hence, the shrunken skeleton has the same properties.

Finally, we consider the possibility of self-loops. Since self-loops cannot arise by dissolving vertices of degree~$2$ or removing edges incident to a vertex of degree~$1$, this means that $C^+$ and thus even $C^*$ must have a self-loop. This self-loop encloses a terminal by the minimality of $C$. Then the edge of $G$ corresponding to this self-loop is a bridge of $G$. However, $G$ is bridgeless by \rev{the fact that the instance is transformed}.
\end{proof}

It is crucial to note that while the skeleton is connected, it is not necessarily bridgeless. This has important algorithmic consequences discussed later.

\begin{definition}
\rev{Let $C$ be any minimum multiway cut of $(G,T,\omega)$ and $C^+$ the set of corresponding edges of the augmented dual.} Let $F_\one \in \mathcal{F}$ and let $f_\one$ be the corresponding face of the skeleton of $C^+$. We call the set of edges of the augmented dual of $f_\one$ the \defi{spine} of $F_\one$ with respect to $C^+$.
\end{definition}
Note that a spine is not necessarily isomorphic to a cycle, but only to a closed walk, because the skeleton is not necessarily bridgeless.

\begin{definition}
Let $F_\one \in \mathcal{F}$. An \defi{enclosing cut} is a cut that separates $T_\one$ from $T \setminus T_\one$ and also separates each terminal in $T_\one$ from every other terminal.
\end{definition}
\begin{remark}
Chen and Wu~\cite{Chen-Wu} use the term \defi{island cut} to denote the $(T_\one, T\setminus T_\one)$-cut for each $F_\one \in \mathcal{F}$. Since they find a 2-approximate solution instead of an exact one, it suffices to separate the graphs into islands and then find a multiway cut for each island.
\end{remark}

\begin{corollary}\label{cor:skeleton}
Let $C$ be any minimum multiway cut of $(G,T,\omega)$ and $C^+$ the set of corresponding edges of the augmented dual. Consider a face $f_\one$ of the skeleton of $C^+$ and let $F_\one \in \mathcal{F}$ be the single terminal face strictly enclosed by it. Then the set of edges of $G$ corresponding to the spine of $F_\one$ is a $(T_\one, T\setminus T_\one)$-cut and the set of edges of $G$ corresponding to the edges of the augmented dual of $C^+$ enclosed by $f_\one$ is an enclosing cut.
\end{corollary}
\begin{proof}
This is immediate from the fact that, by Lemma~\ref{lem:skeleton}, each face of the skeleton strictly encloses exactly one terminal face of $\mathcal{F}$.
\end{proof}

The following is proved as in Dahlhaus \etal~\cite[Lemma~2.2]{DJPSY94} with some modifications.

\begin{lemma}\label{lem:skeleton-size}
Let $C$ be any minimum multiway cut of $(G,T,\omega)$ and $C^+$ the set of corresponding edges of the augmented dual.
Then the shrunken skeleton of $C^+$ has at most $4k$ vertices and at most $12k$ edges.
\end{lemma}
\begin{proof}
\rev{Consider the shrunken skeleton of $C^+$. It has no vertices of degree~$1$ and no self-loops, is connected, and has $k$ faces by Lemma~\ref{lem:skeleton}. However, it might have parallel edges and, therefore, vertices of degree~$2$. Let $F$ be a face with exactly two vertices on its boundary, namely $u$ and $v$. Let $e_1,e_2$ be the two parallel edges that bound $F$. We add an additional edge $uv$ and embed it inside $F$. This creates two faces $F_1$ and $F_2$. Now, we subdivide $uv$, $e_1$, and $e_2$; this create three new vertices, say $x$, $y$, and $z$ respectively. Add the edges $xy$ and $xz$ and embed them inside $F_1$ and $F_2$ respectively. Note that this modification strictly reduces the number of parallel edges and that the newly created vertices have degree at least~$3$. By iteratively modifying each face with exactly two vertices on its boundary in the same way, we ensure that there are no parallel edges and no vertices of degree~$2$ in the shrunken skeleton. This modification at most doubles the number of existing faces and edges.}

\rev{Let $v$, $e$, and $f$ be the number of vertices, edges, and faces of the after applying this modification to the shrunken skeleton. This graph has no vertices of degree~$1$ and $2$; hence, $e \geq 3v/2$. By Euler's formula, $v-e+f = 2$, and thus $v-3v/2+f \geq 2$. Then $v \leq 2(f-2) \leq 2f$. We also recall that any simple connected planar graph satisfies $e \leq 3v-6$ by Euler's formula, and thus $e \leq 6f$. Since our modification at most doubles the original number of faces, $f \leq 2k$. Hence, $v \leq 4k$ and $e \leq 12k$. Since the modifications only increase the number of vertices and edges, these bounds also hold for the original shrunken skeleton (prior to modification).}
\end{proof}

\subsection{Single Face}
Per Corollary~\ref{cor:skeleton}, for each face $F_\one \in \mathcal{F}$ and corresponding face $f_\one$ of the skeleton, the spine of $F_\one$ is an island cut, and the set of edges of the augmented dual of $C^+$ enclosed by $f_\one$ is an enclosing cut. We show that the enclosing cut has a very specific structure.
\rev{Intuitively, it is a cycle with trees attached to them that separate the terminals of $T_\one$. These terminals form intervals along $F_\one$, a notion that we define rigorously first. Then, we prove the structure of more detail, showing that the trees are Steiner trees on an interval and a vertex on the bounding cycle, called nerves. Understanding this structure is helpful when we later design a dynamic programming algorithm that finds a solution with this structure.}

\subsubsection{Intervals}
We need the following definitions.

\begin{definition}\label{def:interval}
Let $F_\one \in \mathcal{F}$.
An \defi{interval} of augmented terminals of $T^+_\one$ is the set $\{\augvertex{\one}{i} \mid i \in I\}$, where $I \subseteq \{1,\ldots,p_\one\}$ is any set such that either $|I| = p_\one$ or for each $i \in \{1,\ldots,p_\one\}$ except exactly one, it holds that there is a $j \in I$ where $j=i+1 \mod p_\one$. We say that the interval is \emph{between $i$ and $j$} ($1 \leq i,j \leq p_\one$) if $i \leq j$ and $I = \{i, \ldots, j\}$, or $i > j$ and $I = \{i,\ldots,p_\one,1,\ldots,j\}$.
\end{definition}

It is important to note that we treat intervals as ordered sets of terminals. In particular, the intervals $\{\augvertex{\one}{1},\ldots,\augvertex{\one}{p_\one}\}$ and $\{\augvertex{\one}{2},\ldots,\augvertex{\one}{p_\one},\augvertex{\one}{1}\}$ are distinct.

This definition invites several additional definitions that will prove useful in later parts of the paper.
Two intervals $I,I'$ of the same terminal face $F_\one$ are \emph{consecutive} if they are disjoint and there is a third (possibly empty) interval $I''$, disjoint from $I$ and $I'$, such that $I \cup I' \cup I'' = \augset{\one}$.

Given two consecutive intervals $I= \{\augvertex{\one}{1},\ldots,\augvertex{\one}{a}\}$ and $I'= \{\augvertex{\one}{a+1},\ldots,\augvertex{\one}{b}\}$ of the same terminal face $F_\one$, a terminal $t \in T_\one$ is \emph{inbetween} $I$ and $I'$ if $t =$ $\terminal{\one}{a}$. Note that there is exactly one terminal inbetween two consecutive intervals, unless $I \cup I' = \augset{\one}$, in which case there are exactly two.
A terminal $t \in T_\one$ is \emph{between} $I = \{\augvertex{\one}{1},\ldots,\augvertex{\one}{a}\}$ if $t \in \{t_1,\ldots,t_{a-1}\}$. 

A \emph{subinterval} of an interval $I = \{\augvertex{\one}{1},\ldots,\augvertex{\one}{a}\}$ is any subset $\{\augvertex{\one}{b},\ldots,\augvertex{\one}{c}\}$ of $I$ such that $1 \leq b \leq c \leq a$ or is the empty interval. A \emph{prefix} of an interval $I = \{\augvertex{\one}{1},\ldots,\augvertex{\one}{a}\}$ is any subinterval of $I$ containing $\augvertex{\one}{1}$ or is the empty interval. A \emph{suffix} of an interval $I = \{\augvertex{\one}{1},\ldots,\augvertex{\one}{a}\}$ is any subinterval of $I$ containing $\augvertex{\one}{a}$ or is the empty interval.

(Note that we have started the intervals at $\augvertex{\one}{1}$ only for simplicity and to avoid modulo-calculus, but the definitions extend in the obvious manner.)

\subsubsection{Enclosing Cut on your Nerves}
Using the notion of intervals, we can describe the enclosing cut. To this end, we rely on the following result.

\begin{theorem}[{Chen and Wu~\cite[Lemma~3]{Chen-Wu}}]\label{thm:chen-wu-full}
Let $(G,T,\omega)$ be an instance of {\problemEMWC} where $G$ is a biconnected plane graph such that all terminals of $T$ lie on the outer face. Then a set of edges is a minimal multiway cut for this instance if and only if the corresponding edges of the augmented dual form a minimal Steiner tree in the augmented dual on the augmented terminals.
\end{theorem}
In this result, the biconnectivity of $G$ ensures that the augmented dual $G^+$ is connected. For the conclusion of the above theorem to hold, it suffices that $G^+$ is connected.

\begin{corollary}\label{cor:connected aug dual sufficient}
Let $(G,T,\omega)$ be an instance of {\problemEMWC} where $G$ is a plane graph such that all terminals of $T$ lie on the outer face and $G^+$ is connected. Then a set of edges is a minimal multiway cut for this instance if and only if the corresponding edges of the augmented dual form a minimal Steiner tree in $G^+$ on the augmented terminals.
\end{corollary}
We can also obtain the following corollary.

\begin{corollary}\label{cor:chen-wu}
Let $(G,T,\omega)$ be an instance of {\problemEMWC} where $G$ is a plane graph such that the outer face is bounded by a cycle and all terminals of $T$ lie on this cycle. Then a set of edges is a minimal multiway cut for this instance if and only if the corresponding edges of the augmented dual form a minimal Steiner tree in the augmented dual on the augmented terminals.
\end{corollary}
\begin{proof}
We observe that in $G^*$, the vertex representing the outer face is not a cut vertex of the solution by Lemma~\ref{lem:dualcutvertex}. Hence, $G^+$ is connected and the result follows from Corollary~\ref{cor:connected aug dual sufficient}.
\end{proof}
It is important to note that the mentioned cycle might be two parallel arcs or a self-loop.

These results suggest that Steiner trees in the augmented dual are important for a multiway cut. This is indeed the case, although the situation is substantially more involved when $|\mathcal{F}|>1$.

\begin{definition}\label{def:nerve}
\rev{Given a plural face $F_\one \in \mathcal{F}$, $v \in V(G^+)$, and integers $i,j$ with $1 \leq i,j \leq p_\one$, a \defi{nerve} for $F_\one$ and $(v,i,j)$ is a minimum Steiner tree in $G^+$, with the terminal set being the attachment point $v$ and the interval of augmented terminals between $i$ and $j$, that has a single augmented dual edge incident to $v$. The vertex $v$ is called the \defi{attachment point} of the nerve.}
\end{definition}
For a nerve $(v,i,j)$, we may speak of the interval between $i$ and $j$ as simply the interval of the nerve. In this way, we can extend the definitions of consecutive, between, and in between, originally defined for intervals, to nerves in the straightforward manner. Note that nerves are not necessarily unique for a triple $(v,i,j)$; we deal with this issue later.

A Steiner tree, as in Definition~\ref{def:nerve} might not necessarily exist, but we argue that our solution is built only from such Steiner trees.

\begin{lemma}\label{lem:nerves}
Let $C$ be any minimum multiway cut of $(G,T,\omega)$ and $C^+$ the set of corresponding edges of the augmented dual. Consider a face $f_\one$ of the skeleton $S^+$ of $C^+$ and let $F_\one \in \mathcal{F}$ be the single terminal face strictly enclosed by it. \rev{Let $X^+$ be the set of edges of the augmented dual enclosed by $f_\one$. Then $X^+$ is a union of nerves, such that the union of all intervals of the nerves is $T^+_\one$.}
\end{lemma}
\begin{proof}
\rev{Let $Y^+$ be the spine of $F_\one$.} 
By Corollary~\ref{cor:skeleton}, $G \setminus Y$ does not contain any path connecting the terminals of $T_\one$ to terminals of $T \setminus T_\one$. Therefore, the region of $G^+$ bounded by $Y^+$ and the boundary of $F_\one$ does not contain any terminal of $T\setminus T_\one$. The edges in $X$ thus correspond to a multiway cut of $T_\one$ in $G \setminus Y$. Moreover, this multiway cut is minimal in $G \setminus Y$: any edge that could be removed from $X$ while it remains a multiway cut in $G\setminus Y$ can also be removed from $C$ while it remains a multiway cut in $G$, which would contradict the minimality of $C$.

Now consider $G \setminus Y$. Note that the spine does not contain any augmented terminals, as those have degree~$1$ in $C^+$ and thus are not part of the skeleton by definition. Hence, $Y$ does not contain any edges of $F_\one$ and therefore, $F_\one$ persists in $G\setminus Y$. 
Let $Q$ be the component of $G \setminus Y$ that contains $F_\one$.
Since the instance is preprocessed by adding an edge between every pair of consecutive vertices on any terminal face of the input graph $G$, $F_\one$ is a simple cycle in $Q$. 
Moreover, $X$ is a minimal multiway cut for $T_\one$ in $Q$, as argued above.
Hence, by Corollary~\ref{cor:chen-wu}, $X^+$ is a minimal Steiner tree in the augmented dual of $Q$ on the terminal set $T^+_\one$. It is in fact of minimum weight; otherwise, we could replace $X^+$ by a minimum-weight Steiner tree on $T^+_\one$ and obtain a minimum multiway cut for $G$ of smaller weight using Corollary~\ref{cor:chen-wu} and the fact that $Y$ forms an island cut.

Now we consider what $X^+$ looks like in the augmented dual of $G$. We note that all edges of the augmented dual of $G \setminus Y$ also appear in the augmented dual of $G$. Hence, we can think of $X^+$ as a set of edges in $G^+$. Let $X$ be the corresponding set of edges of $G$. Deleting the set of edges in $Y$ from $G$ is equivalent to contracting the edges of $Y^+$ in the augmented dual. Let $y^+$ be the vertex formed by contracting all the edges in $Y^+$. Since $C^+$ is connected by Lemma~\ref{lem:augmentedconnected}, it follows that at least one edge of $X^+$ is incident to $y^+$.

Deleting $y^+$ from $X^+$ creates $\deg(y^+)$ many connected components of $X^+$. We treat each of these connected components as subsets of edges, which includes the unique edge incident to $y^+$. We claim that each of these components \rev{is a nerve, defined by $y^+$ and some interval of $T^+_\one$ in its vertex set}. Clearly, each of the components is a Steiner tree on $y^+$ and some subset of $T^+_\one$. We first show that the augmented terminals contained in each component form an interval of $T^+_\one$. Due to the planarity of $G^+$, no two tree edges cross each other. Therefore, if the terminals of any two connected components were to cross each other on $T^+_\one$, then they must intersect in some vertex. However, since each connected component is maximally connected, we reach a contradiction. Hence, the terminals form an interval.
 
 \begin{figure}[t]
     \centering
     \includegraphics[width=0.5\textwidth]{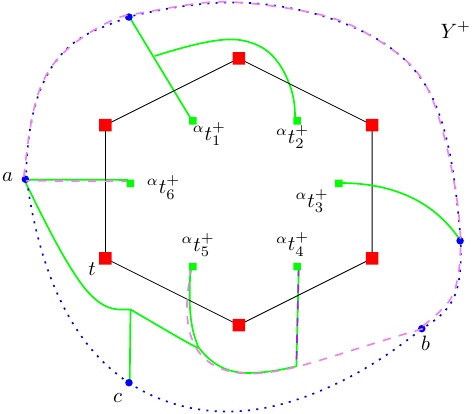}
     \caption{The boundary of the face $F_\one$ is shown in black. The terminals of $T_\one$ lying on the boundary are drawn as red boxes. The spine $Y^+$ of the minimum enclosing cut dual is drawn in blue dotted lines. The green trees are the minimum Steiner trees of $X^+ \setminus \{y^+\}$ in $G^+/Y^+$. The violet dashed lines represent the \mst~in $G^+$ on the terminal set $\{\augvertex{\one}{4}, \augvertex{\one}{5}, \augvertex{\one}{6}\}$, in which the path between $\augvertex{\one}{5}$ and $\augvertex{\one}{6}$ contains the subpath of $Y^+$ between the vertices $a$ and $b$ that does not enclose $t$.}\label{fig:no expose}
 \end{figure}
 
Due to the minimality of $X^+$ in $G^+/Y^+$, with respect to the terminal set $T^+_\one \cup \{y^+\}$, we claim that each of the components is a \mst~in $G^+/Y^+$ containing its respective interval of $T^+_\one$. Suppose that one of these trees is not a minimum Steiner tree in $G^+$. Then, there must exist in $G^+$ a Steiner tree $W^+$ of lower weight on the same terminal set, such that it either contains some edges of $Y^+$ or crosses $Y^+$. In the latter case, there is a subpath of this Steiner tree that intersects $Y^+$ in at least two vertices, leading to the creation of a cycle enclosing no terminal of $T_\one$. That would contradict the minimality of $X^+$. In the former scenario, suppose that we replace the concerned component of $X^+$ with $W^+$. We claim that all the previously enclosed terminals by the component of $X^+$, remain enclosed after the replacement. If for some terminal $t \in T_\one$ in the interval, $W^+$ does not contain a subpath enclosing it, then the path between the augmented vertices flanking it contains edges of $Y^+$. This path in $W^+$ crosses all the paths from $t$ to  $T_\one \setminus \{t\}$, and the segment of $Y^+$ between the points of intersection crosses all the paths from $t$ to terminals in $T \setminus T_\one$. Therefore, the region bounded by $W^+$ and $Y^+$ that contains $t$ cannot contain any other terminal of $T_\one$, and $t$ remains enclosed. The replacement would, thus, yield an enclosing cut of $T_\one$, strictly smaller than $X^+ \cup Y^+$. This, again, contradicts the minimality of $X^+$. Figure~\ref{fig:no expose} illustrates the argument.

Finally, suppose that some \mst~intersects the spine in more than one vertex. Then $X^+ \cup Y^+$ would contain a cycle that encloses no terminals, which contradicts Proposition~\ref{prp:dualfaces}. \rev{Following the above arguments, that we may assume that a single edge is incident to the attachment points.}
\rev{Therefore, $X^+$ is a union of nerves, such that the union of all intervals of the nerves is $T^+_\one$.}
\end{proof}

\rev{The proof of the previous lemma immediately implies the following:}

\begin{corollary} \label{cor:forest}
\rev{Let $C$ be any minimum multiway cut of $(G,T,\omega)$ and $C^+$ the set of corresponding edges of the augmented dual. Consider a face $f_\one$ of the skeleton $S^+$ of $C^+$ and let $F_\one \in \mathcal{F}$ be the single terminal face strictly enclosed by it. Let $X^+$ be the set of edges of the augmented dual enclosed by $f_\one$. Let $Z^+$ be any other set of nerves on the same attachment point and interval combinations as the nerves of $X^+$. Then $(C \setminus X) \cup Z$ is a minimum multiway cut of $(G,T,\omega)$.}
\end{corollary}

We now argue that nerves can be computed in polynomial time.

\begin{lemma}\label{lem:nerves-compute}
Given \rev{an augmented dual vertex} $v$ and an interval of a terminal face $F_\one$ between $i$ and $j$, we can compute in polynomial time a nerve $(v,i,j)$, if it exists.
\end{lemma}
\begin{proof}
It follows from Bern~\cite{Bern} (or Kisfaludi-Bak \etal\cite{KNL20}) that a minimum Steiner tree on a set of terminals with at most two faces can be computed in polynomial time. We first compute a minimum Steiner tree $Z^+$ in $G^+$ on $v$ and the augmented terminals in the interval between $i$ and $j$. Then, for each neighbor $u$ of $v$, we compute a minimum Steiner tree $Z^+_u$ in $G^+ \setminus \{v\}$ on $u$ and the augmented terminals between $i$ and $j$. By the preceding, this can be done in polynomial time. Then we compare $\omega(Z^+)$ with $\omega(Z^+_u) + \omega(v,u)$ for all $u$. If there is a $u$ for which they are equal, then $Z^+_u$ combined with the edge $(v,u)$ is a nerve. Moreover, any nerve can be constructed in this way.
\end{proof}

\rev{Recall that nerves} are not necessarily unique for a triple $(v,i,j)$. From now on, we associate a unique minimum Steiner tree with each triple (if it exists), namely the one computed by Lemma~\ref{lem:nerves-compute}. We call this the \emph{unique nerve}; \rev{by abuse of notation, we denote this by $(v,i,j)$}. In the remainder, whenever we talk about a nerve, we mean the unique nerve on the same interval and attachment point. Following Corollary~\ref{cor:forest}, we may assume that minimum-weight multiway cuts are built from a skeleton and (unique) nerves.

\begin{definition}
Let $C$ be any minimum multiway cut of $(G,T,\omega)$ and $C^+$ the set of corresponding edges of the augmented dual. Then $C^+$ (and by extension, $C$) is called \defi{nerved} if the edges of $C^+$ that are not part of the skeleton each belong to a \rev{(unique)} nerve.
\end{definition}
Observe that following Corollary~\ref{cor:forest} and the discussion above, we can assume that \rev{we consider a minimum-weight multiway cut that is nerved}.

\section{Bones and Homotopy}\label{sec:Bones and Homotopy}
\rev{As before, we are given an instance $(G,T,\omega,\mathcal{F})$. Recall that we assume that the instance is transformed and that the dual of any optimum solution is connected.}
We now wish to describe the structure of the paths between the attachment points of nerves and to branching points of the skeleton. To this end, we start with the following definition.

\begin{definition}
\rev{Let $C$ be any multiway cut of $(G,T,\omega)$ and $C^+$ the set of corresponding edges of the augmented dual.} Each edge of the shrunken skeleton \rev{of $C^+$} is called a \defi{shrunken bone} and corresponds to a path of the skeleton, called a \defi{bone}. Any vertex of the shrunken skeleton \rev{of $C^+$} is called a \defi{branching point}. 
\end{definition}

By Lemma~\ref{lem:skeleton}, each (shrunken) bone is incident to one or two faces of the (shrunken) skeleton, and thus separates at most two faces of $\mathcal{F}$, called the \emph{separated terminal faces} of the (shrunken) bone. Note that both sides of a (shrunken) bone might be incident to the same face, as the shrunken skeleton can contain bridges.

\subsection{Orienting Nerves}

To further discussions in the remainder of the paper, it helps to define an orientation of the nerves, which specifies where a nerve goes with respect to a bone.
Let $C$ be any minimum multiway cut of $(G,T,\omega)$ and $C^+$ the set of corresponding edges in the augmented dual. For any augmented dual vertex $x$, there is a small ball centered on $x$ that contains $x$ and points of its incident augmented dual edges \rev{but no other point of the augmented dual graph}. Let $e$ and $e'$ be two edges of the augmented dual incident to $x$. Then this small ball is split into two parts by the union of $e$ and $e'$. An augmented dual edge is \emph{west} of $x$ with respect to $e$ and $e'$ if it is contained in the part of the ball clockwise from $e'$ to $e$ and \emph{east} otherwise. 

Now consider two edges of the augmented dual $e$ and $e'$ of the skeleton of $C^+$ that appear in counterclockwise order on the boundary of a face $f_\one$ of the skeleton of $C^+$ (so $e$ appears after $e'$). Let $x$ be their shared endpoint. \rev{If the boundary of $f_\one$ is not a simple cycle}, then we only consider the first appearance of an augmented dual edge to define the ordering.  Let $F_\one \in \mathcal{F}$ be the terminal face corresponding to $f_\one$. 
We then say that a nerve $(x,h,i)$ for $F_\one$ with $1 \leq h,i \leq p_\one$ \emph{extends west} of $x$ if it holds that the augmented dual edge incident to $x$ that belongs to the nerve is west of $x$. If the augmented dual edge incident to $x$ that belongs to the nerve is east of $x$, then the nerve \emph{extends east}. By definition, a nerve has only a single incident edge, and thus it either extends east or it extends west.

\subsection{Homotopy and the Optimum Solution}
We now want to describe the structure of bones. While it might seem that these are just shortest paths between certain branching and attachment points, this does not account for their role in separating pairs of terminals. In order to specify this role, we need assistance from homotopy. 
 
To facilitate the discussion of homotopy, we first need some definitions. First, we need the notion of crossings between two paths $P$ and $Q$ in a planar graph. \rev{We say that there is a \emph{crossing} between $P$ and $Q$ if, after contracting any common edges of $P$ and $Q$, there are two distinct edges of the remainder of $P$ and two distinct edges of the remainder of $Q$, all incident on the same vertex, such that a clockwise rotation around the vertex will see the two edges of $P$ and $Q$ alternately. We say that $P$ \emph{crosses} $Q$ if there is at least one crossing between them.}

Now we construct a \emph{cut graph} $K$ for $G^*$, which is a Steiner tree in $G^*$ with the set $V^*_{\mathcal{F}}$ of dual vertices of the terminal faces in $\mathcal{F}$ as terminal set. Here we follow Frank and Schrijver~\cite[Proposition~1]{Homotopy}. First, we compute a shortest path between each pair of dual vertices in $V^*_{\mathcal{F}}$. Let $\mathcal{P}$ denote the resulting set of paths. By a slight modification of the weights for the sake of this computation, we can assume all shortest paths in $G^*$ are unique. Then we can assume that any pair $P,Q$ of shortest paths in $\mathcal{P}$ either \rev{has at most one crossing}, or has a common endpoint and does not cross. Now build an auxiliary complete graph on $V^*_{\mathcal{F}}$ with the lengths of the paths in $\mathcal{P}$ as weights and find a minimum spanning tree in this auxiliary graph. The paths in $\mathcal{P}' \subseteq \mathcal{P}$ corresponding to this spanning tree form the cut graph $K$.

As explained by Frank and Schrijver~\cite[Proposition~1]{Homotopy}, we may assume the paths in $\mathcal{P}'$ do not cross. We call the paths in $\mathcal{P}'$ the \emph{spokes} of the cut graph. Note that $K$ has $k-1$ spokes, each of which can be associated with a unique identifier. We also orient each spoke in an arbitrary direction. Then any oriented path $P$ that crosses a spoke $Q$ can be said to cross $Q$ in a particular direction, depending on whether $P$ goes towards the west or east side of $Q$.

\begin{definition}
The \defi{crossing sequence} or \defi{homotopy string} of a path $P$ in $G^+$ is an ordered sequence of the spokes of $K$ crossed by it, along with an indication of the orientation of each crossing. Let $P$ and $P'$ be two paths in $G^+$. We say that $P$ is \emph{homotopic} to $P'$ if they have the same endpoints as well as the same crossing sequence with respect to $K$.
\end{definition}

Note that we perform a slight abuse here by considering crossings of a path in $G^+$ with paths in $G^*$. However, we will never consider crossings of paths in $G^+$ that have an endpoint in an augmented dual terminal, alleviating any possible cause for confusion.

Using this precise definition of homotopy, we now fix a particular solution to the instance. Among all possible nerved minimum multiway cuts $C$ of $(G, T, \omega)$, we assume henceforth that $C$ is one of which the skeleton of $C^+$ has the minimum number of crossings with $K$\rev{, where $C^+$ is again the set of corresponding edges of the augmented dual. Note that $C$ is inclusion-wise minimal. Moreover,} the skeleton contains no vertices of the augmented dual, as those have degree~$1$ in $G^+$. 
From now on, we fix this solution as `the' optimal solution or `the' optimum.

Our main goal in the next subsections will be to show that (a sufficient part of) the bones of the skeleton of the optimum solution have short homotopy strings and that replacing such parts by \rev{short} parts with the same homotopy string leads to another optimum solution.

\subsection{Nerve Path}
We now consider a subpath of a bone that starts and ends at attachment points of two distinct nerves and contains attachment points only of nerves that extend towards the same direction. In particular, the nerves extend towards the same terminal face. We call such a subpath a \defi{nerve path}. We aim to prove the following crucial property of nerve paths, namely that any nerve path has a homotopy string of bounded length.

\begin{definition}
\rev{Let $P$ be a path in $G^+$ that does not contain any augmented terminals and has some nerves attached to it that only extend towards terminal faces $F_\one,F_\two \in \mathcal{F}$ (possibly $\one=\two$).
Let $N_1,N_2$ be two nerves that attach to $u,v$ respectively. Let $f_\one, f_\two$ be the faces of the skeleton separated by $b$ and let $F_\one, F_\two \in \mathcal{F}$ be the corresponding terminal faces. 
Then $P$ is a \emph{nerve path with bookends $N_1,N_2$} if all nerves that attach to $P\setminus\{u,v\}$, as well as $N_1,N_2$, extend towards the same terminal face (say $F_\two$); if $\one=\two$, then moreover, all these nerves extend towards the same direction.}
\end{definition}

\rev{Note that the bookends are not uniquely determined by the ends of a nerve path, as there may be several attaching to $u$ or $v$ towards the same face. However, in our upcoming analysis, there is only ever one possibility.}

\begin{lemma}\label{lem:singular face bone}
\rev{Any nerve path in $C^+$ that is a subpath of a bone of the skeleton of $C^+$} crosses any path of $K$ \Oh{k} times.
\end{lemma}

\begin{proof}
We modify the proof of~\cite[Lemma 5.2]{CDV} to prove the claim above. Let $K$ be the graph defined above and $\gamma$ be a path in $K$. We treat the dual vertices of the terminal faces as obstacles \rev{(see Section~\ref{sec:topology})}. Let $u,v$ be two endpoints of a nerve path $P$ of $C^+$ \rev{that is a subpath of a bone $b$ of the skeleton of $C^+$}. Since $P$ is a path in the augmented dual $G^+$, $P$ does not pass through any vertex of $K$ that corresponds to a terminal face; in particular, it does not pass through the endpoints of $\three$.

Let $f_\one, f_\two$ be the faces of the skeleton separated by $P$ and let $F_\one, F_\two \in \mathcal{F}$ be the corresponding terminal faces. Denote the bookends of $P$ by $N_1$ and $N_2$. Note that these are uniquely defined: at most one nerve can attach to each augmented dual vertex on $P$, because $u,v$ have degree at least~$2$ in the skeleton by Lemma~\ref{lem:skeleton} and each augmented dual vertex has degree at most~$3$ by Lemma~\ref{lem:augmenteddualdegree}. Without loss of generality, we assume that both $N_1$ and $N_2$ (and all the other nerves on $P$, if they exist) extend towards $F_\two$. If $\one=\two$, then without loss of generality, we assume that all the nerves extend towards the west of $b$.
   
\rev{We now consider the following topological construction:} we push all the crossings of $P$ with $\three$ together to a single point \rev{$p$}. \rev{Precisely, we contract the dual edges on the maximal subpath of $\three$ that contains intersection points of $P$ and $\three$, into the point $p$}. Then the subpaths of $P$ between consecutive crossings of $P$ with $\three$ form cycles containing \rev{$p$}, which we call \emph{loops}. The corresponding system of loops is denoted by $P^*$. Then, $P^*$ is a set of pairwise disjoint loops $L$ intersecting at \rev{$p$}. The contracted point is referred to as the base point of the loops. A face of $L$ is called a \defi{monogon} if it is homeomorphic to an open disc and has only one copy of the base point on its boundary. Likewise, a \defi{bigon} is an open disc with two copies of \rev{$p$} on its boundary. In the bone $P$, a bigon occurs as a strip bounded by two subpaths of the path $\three$ and two subpaths of $P$. \rev{Since monogons and bigons are homeomorphic to an open disc, they enclose no obstacles and thus, in particular, enclose no terminal faces.} \rev{See Figure~\ref{fig:mono-bi-example} for an example of monogons or bigons}. To show that the number of crossings of $P$ with $\three$ is \Oh{k}, we need to show that the degree of \rev{$p$} is \Oh{k} in $P^*$. 
 
  \begin{figure}[tp]
    \centering
    \includegraphics[width=0.5\textwidth]{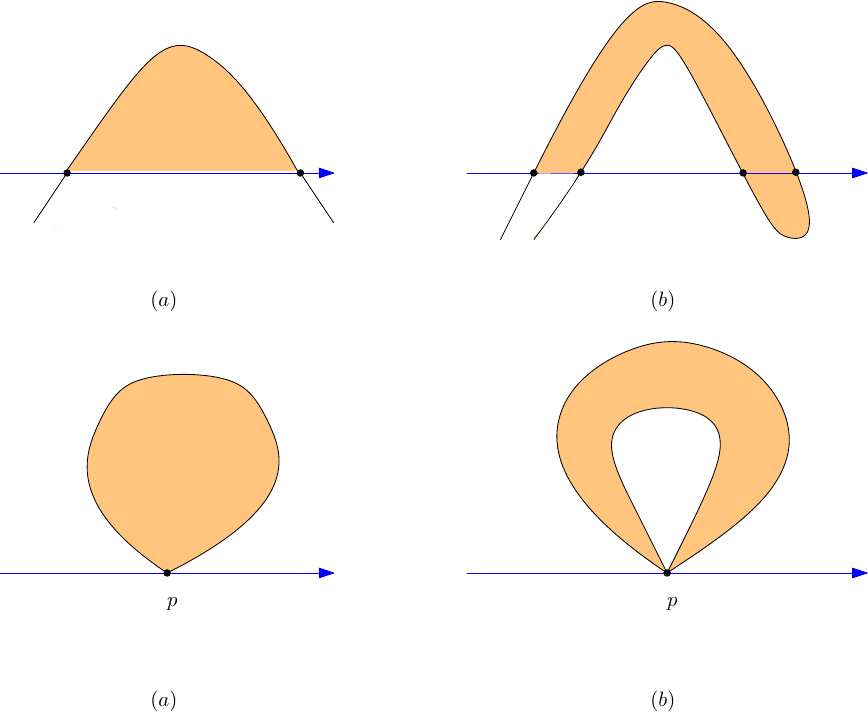}
    \caption{\rev{The figure shows an example of (a) a monogon and (b) a bigon formed by $P$ crossing $\three$.}}
    \label{fig:mono-bi-example}
  \end{figure}

 \begin{figure}[tp]
    \centering
    \includegraphics[width=\textwidth]{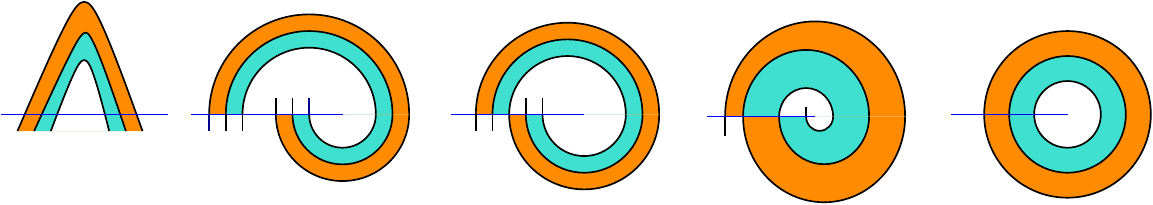}
    \caption{The figure shows all possible ways in which a loop in $L$ can be incident to two bigons. The two bigons are shaded in orange and turquoise. The blue line depicts $\three$.}
    \label{fig:CDV}
 \end{figure}

 Figure~\ref{fig:CDV}~(same as~\cite[Figure~2]{CDV}) shows an exhaustive list of configurations of any path crossing an edge of the cut graph $K$, which lead to the occurrence of a loop incident to two bigons in $L$. Note that the right-most configuration does not occur in our setting since $P$ is a path. 
 
 \rev{As a first step in our proof, we show} that no loop in $L$ can be incident to two bigons. Let us assume the contrary and \rev{consider two bigons $I_1,I_2$ that are incident on the same loop}. \rev{At the outset, we exclude some basic configurations (shown in Figure~\ref{fig:narrow}) by providing shortcuts. The following makes this rigorous:}

 \begin{figure}[tp]    
    \includegraphics[scale=0.7]{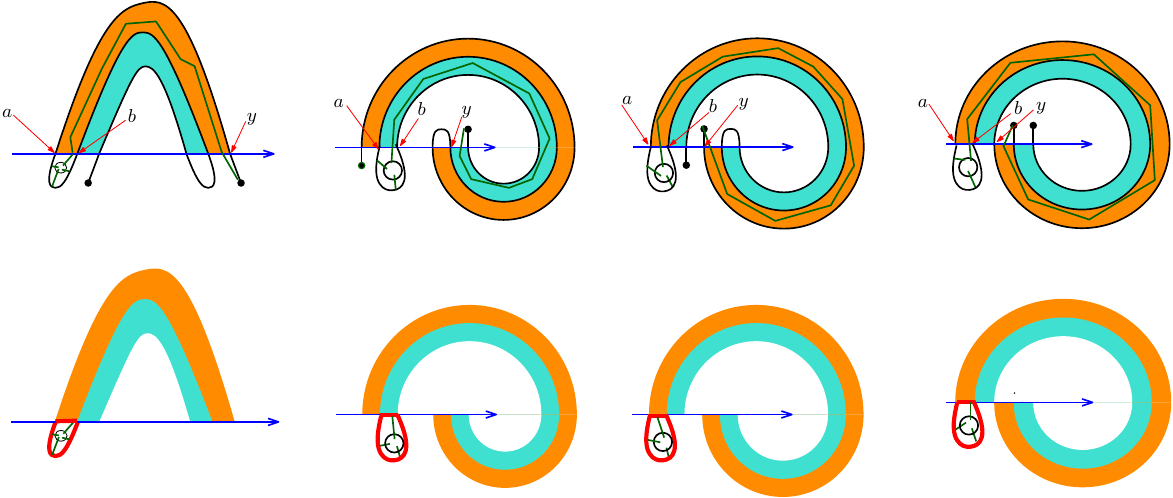}
    \caption{\rev{The figure shows four examples of the configuration described in Claim~\ref{clm:narrow} (top) and the result of the corresponding shortcut described in the claim (bottom).}}
    \label{fig:narrow}
 \end{figure}

\begin{claim}\label{clm:narrow}
    \rev{Let $a$ and $b$ be consecutive crossings of $P$ with $\three$ such that $\three[a,b]$ lies on the boundary of some bigon. Suppose $P[a,b]$ is not incident to $I_1$ or $I_2$ (i.e.\ does not contain any of the three loops defining the bigons). If $F_\two$ is contained in the region bounded by $\three[a,b]\cup P[a,b]$ and there exists a subpath $P[a,y]$, incident to at least one of $I_1,I_2$ such that $\three[a,b] \subseteq \three[a,y]$, then we contradict our choice of $C^+$.}
\end{claim}
\begin{claimproof}
    Any nerve reaching $F_\two$ must either cross $\three[a,b]$ or have its attachment point on $P[a,b]$. Therefore, we can remove $P[a,y]$ from $C^+$ and replace it with $\three[a,b]$. Since $\three[a,y]$ is a shortest path in $G^*$ and therefore in $G^+$, it has length at most the length of $P[a,y]$ and since $\three[a,b]$ is a subpath of $\three[a,y]$, we get a solution of weight at most that of $C^+$ with strictly fewer crossings. This is a contradiction.
\end{claimproof}

Consider the two bigons $I_1,I_2$, depicted as orange and turquoise strips in Figure~\ref{fig:CDV}. There are three loops that together bound the union of the two bigons, namely, the top-most loop incident to the orange bigon only, the middle loop incident to both the bigons and the bottom-most loop incident to the turquoise bigon only. As neither of the bigons contain any terminal and due to Claim~\ref{clm:narrow}, the terminal face $F_\two$ can be contained in the following two regions:
\begin{itemize}
    \item Region A : the region bounded by the bottom-most loop incident to the turquoise bigon and the subpath of $\three$ between its endpoints.
    \item Region B: the unbounded region incident to the union of top-most loop incident only to the orange bigon and the subpath of $\three$ between its endpoints.
\end{itemize} 
 
 We consider different configurations of $P$ creating a loop incident to two bigons and in each such configuration, we present an exchange argument to replace $P$ by a shorter path in $C^+$ that preserves the feasibility of $C^+$. The replacement is based on the intuition that any nerve that must cross $\three$ to reach $F_\two$ can be shortcut to have its attachment point on $\three$. However, there can be nerves that never cross $\three$. Then, we must preserve all subpaths of $P$ containing the attachment points of such nerves. 

 We say that a nerve \emph{traverses} a bigon if any path of the nerve from its attachment point to its leaves crosses each of the two distinct subpaths of $\three$ bounding the bigon an odd number of times.

\begin{figure}[pt]    
        \centering
        \includegraphics[width=0.9\textwidth]{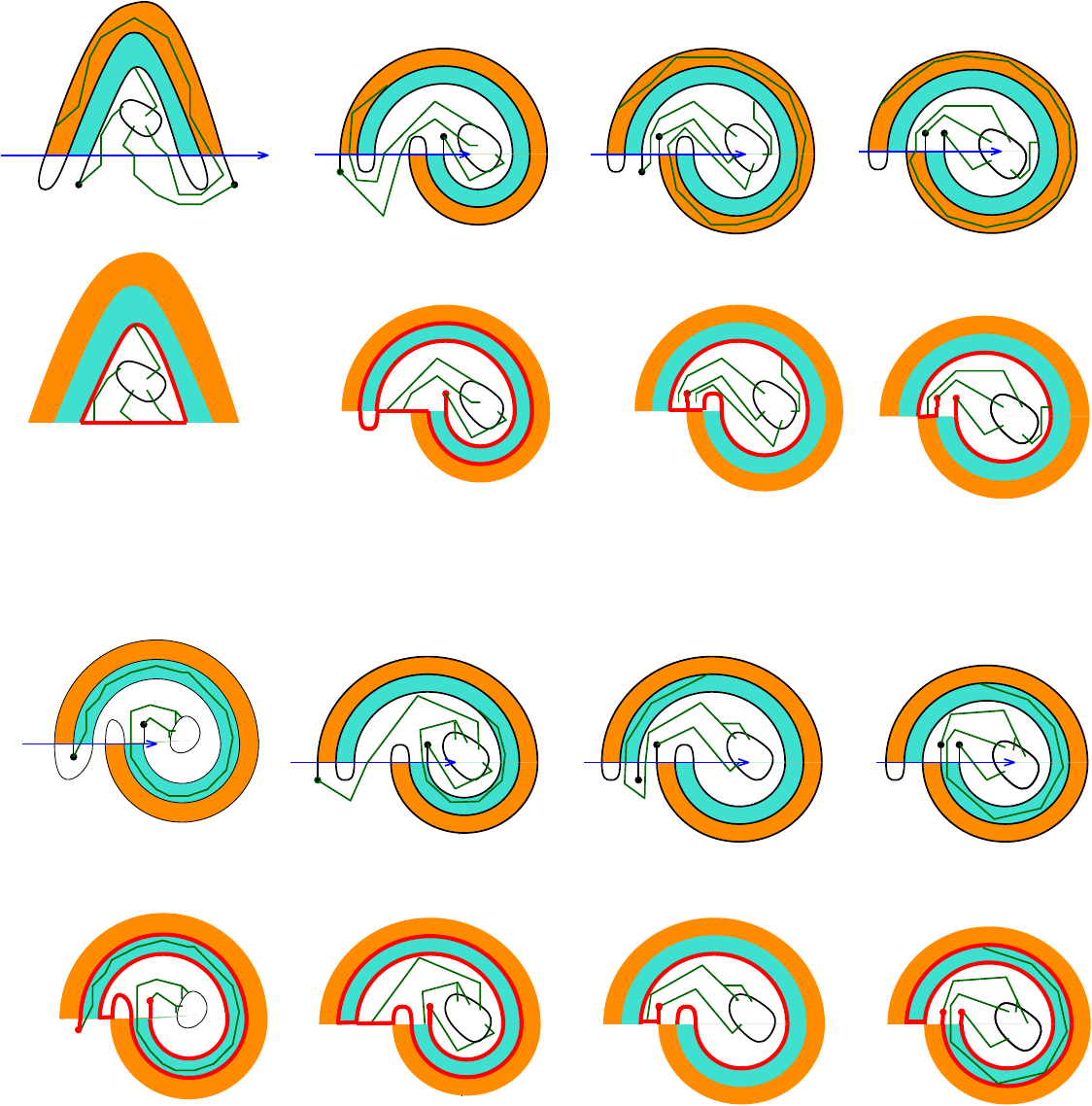}       
        \caption{Case 1: The terminal face is in the region A. The black curve shows the nerve path $P$. The directed blue line is the path $\three$. In dark green lines, we depict the nerves attached to $P$. The face $F_\two$ is drawn as a black splinegon in the outer face of $P^*$. In the top half of the figure, we consider the case where the nerves extend to the side of the orange bigon. In the bottom half of the figure, we consider the configurations where the nerves extend to the side of the turquoise bigon. \rev{For each half, the top row indicates a topological situation, and the bottom row the corresponding shortcut that contradicts the choice of $C^+$.}}\label{fig:Case 1}   
 
\end{figure}

\medskip

\noindent \textbf{Case 1: $F_\two$ is contained in Region A.}\\
Here, we consider two subcases, namely, whether subpaths of  nerves are contained in the orange or the turquoise bigon. Figure~\ref{fig:Case 1} shows the specific replacement for every configuration possible in this case (modulo symmetric configurations with the same replacement).

\medskip 
\noindent\textbf{Case 2: $F_\two$ is contained in Region B.}\\
We shall consider two subcases here:\\
\noindent\textbf{Case 2a.} \emph{Either $N_1$ or $N_2$ traverses one of the two bigons.}\\Suppose that $N_1$ traverses a bigon. Figure~\ref{fig:case2a} depicts this case. Suppose that the leftmost crossing of $N_1$ with $\three$ is $x$ and the rightmost crossing is $y$. Then $\three[x,y]$ is at most as long as $N_1[x,y]$, allowing us to shortcut the nerve.\\
\noindent\textbf{Case 2b.} \emph{Neither $N_1$ nor $N_2$ can traverse any bigon.}\\
Figure~\ref{fig:Case2b} depicts this case. The top figure shows the cases where the nerves extend toward the orange bigon whereas the bottom figure depicts the cases where the nerves extend towards the turquoise bigon.
\medskip

\begin{figure}[t!]
    \centering   
    \includegraphics[scale=0.7]{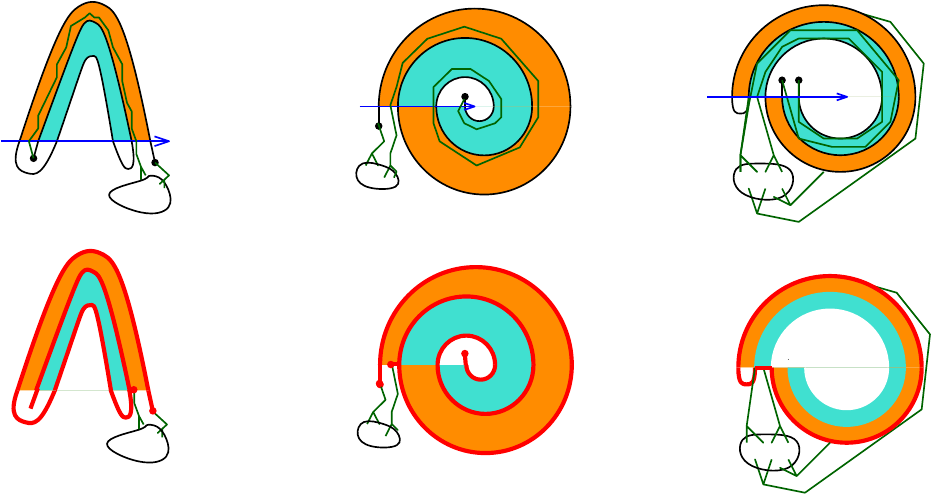}
     \caption{Case 2a: In the top row, we show configurations where $F_\two$ is contained in the outer face of $P^*$. Either $N_1$ or $N_2$ traverses one of the two bigons. The bottom row shows the \rev{shortcut that contradicts the choice of $C^+$} as bold red curves.}\label{fig:case2a}
\end{figure}

\begin{figure}[tp]
    \centering
    \includegraphics[scale=0.7]{ 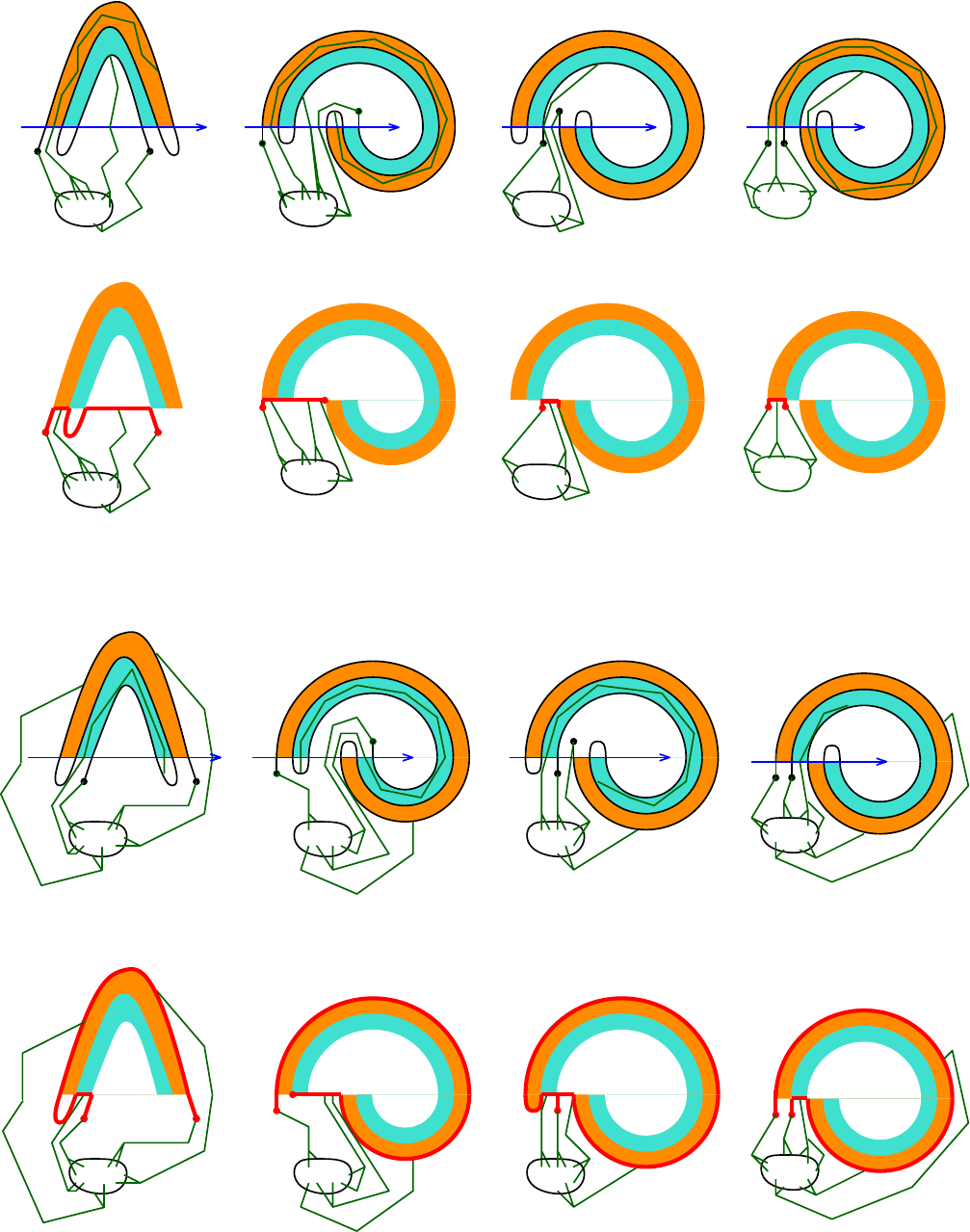}
    \caption{Case 2b: The top half of the figure depicts the case when the nerves attached to $P$ extend towards the orange bigon. The bottom half of the figure depicts the case when they extend towards the turquoise bigon. \rev{For each half, the top row indicates a topological situation, and the bottom row the corresponding shortcut that contradicts the choice of $C^+$.}}
    \label{fig:Case2b}
\end{figure}

Since, in each of the cases, \rev{we can obtain a multiway cut with weight smaller than $C$ or a multiway cut for which the skeleton of $C^+$ has fewer crossings with $K$}, we obtain a contradiction to the minimality of $C^+$. Therefore, no loop of $L$ can be incident to two bigons. \rev{This completes the first step of our proof.}

\bigskip

\rev{For the second step of our proof}, we deal with the monogons. Some monogons may have nerves attached to them that extend towards the nerve face without ever crossing $\three$. These are the monogons that we cannot replace by the subpath of $\three$ bounding them on one side. However, in what follows, we argue that only \Oh{k} of these monogons actually exist in $C^+$.

We orient $\three$ and $P$ arbitrarily. \WLOG~assume that the nerves extend to the west of $P$. Consider the system of loops $P^*$. Let \defi{good loops} be those that are present to the east (or west) of $\three$, possibly have nerves attached to them extending to the west (or east) of $P$, and border a monogon or enclose a bigon that it borders. We can remove the good loops from $C^+$ without contradicting its feasibility by shortcutting (see Figure~\ref{fig:good monogon}). Loops that enclose any terminal face (the end point of a path of $K$) in their interior are called \defi{obstacle loops}. Now consider loops that form monogons.\ \defi{Bad monogons} are those that are present on the west (or east) of $\three$ and contain nerves attached to them that extend towards the west (or east) of $P$. Figure~\ref{fig:good monogon} shows two bad monogons flanking a good loop. In what follows, we argue that between two obstacle loops, there can be at most four bad monogons.
  
  \begin{figure}[t]
      \centering
      \includegraphics{ 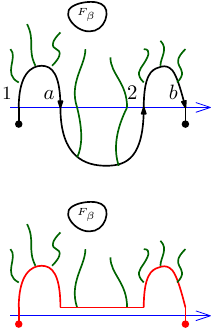}
      \caption{The blue line is the path $\three$ and the arrow indicates its orientation. The black curve shows $P$. The dark green curves are the nerves attached to $P$. The first and third monogons are bad, and the middle loop is good. The figure at the bottom shows the shortcut of the good loop as a thick red curve.}\label{fig:good monogon}
  \end{figure}
  \begin{figure}[t]
      \centering
      \includegraphics{ 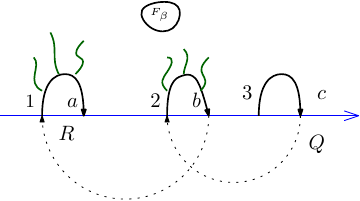}
      \caption{We assume that no two of the three consecutive bad monogons appear in the same order as on $P$.}\label{fig:wrong order}
  \end{figure}

  The \emph{direction} of a bad monogon depends on how its loop is oriented with respect to $\three$: either the head or the tail of the loop comes first on $\three$. 
  Two bad monogons are \emph{consecutive} if there exists no other bad monogon 
  between them. Consider any five consecutive bad monogons on $P$ that have no obstacle loop between them. \WLOG we assume that three of them are bad monogons in the same direction. If any two of these appear in the same order along $\three$ as their order on $P$, then we can shortcut $P$ as in Figure~\ref{fig:good monogon}, because the subpath of $P$ between them does not contain an obstacle loop and thus contains a good loop. So, we assume that no two of them appear in the same order. We denote the crossing points of these three monogons with $\three$ by $1,a,2,b,3,c$ in the order along $\three$. Let $Q$ be the subpath connecting $c$ to $2$ and $R$ be the subpath of $P$ connecting $b$ to $1$. Figure~\ref{fig:wrong order} illustrates the configuration. 
  
  We consider the following cases:
  
  \medskip

  \noindent\textbf{Case 1: $Q$ crosses $\three$ for the first time to the right of $c$.}\\
  In this case, we can shortcut the loop of $Q$ to the east of $\three$ from $c$ to the crossing. We can do so as this loop is good; indeed, it suffices to observe it is not an obstacle loop, as it is between two of the chosen bad monogons.
  
\begin{figure}[t]
    \centering
    \includegraphics{ 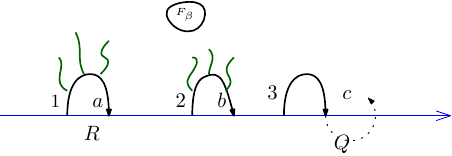}
    \caption{Case 1: $Q$ crosses $\three$ for the first time to the right of $c$.}\label{fig:mon_case1}
\end{figure}
  
  \noindent\textbf{Case 2: $Q$ crosses $\three$ between $b$ and $3$ and crosses $\three$ a second time to the right of the first crossing.}\\ As illustrated in Figure~\ref{fig:mon_case 2}, we either reach the next monogon of the sequence (as in the right figure) or reach a point $e$ to the right of $c$, and must cross $\three$ again to reach $2$. In the latter case, we follow the subpath of $Q$ that starts at $e$ and restart the case analysis from $e$ onwards. Then, in either of the two scenarios, due to the planarity of $C^+$, the loop $R$ from $b$ to $1$ must cross $\three$ to the right of $b$ to avoid crossing $Q$. Then, we get a good loop, and can shortcut $P$.
  
\begin{figure}[tp]
    \centering
    \includegraphics[width=\textwidth]{ 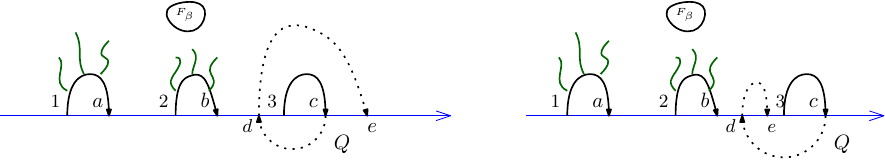}
    \caption{Case 2: $Q$ crosses $\three$ between $b$ and $3$ for the first time and to the right of the first crossing after that.}\label{fig:mon_case 2}
\end{figure}
  
 \noindent\textbf{Case 3: $Q$ crosses $\three$ between $b$ and $3$ and again to the left of its crossing point.}\\
  The second crossing could lie in between $b$ and the first crossing of $Q$, or to the left of $2$. In the former case, $Q$ forms a good loop, and we can shortcut it by the piece of $\three$ between the two crossings. In the latter case, there must not be any obstacle in the region enclosed by the loop $2$-$b$ and the one formed by $Q$, or else $Q$ would contain an obstacle loop. Therefore, we get a good loop formed by the loop of $Q$ between the two crossings. This is a contradiction to the loop between $2$ and $b$ being a bad monogon, and we can shortcut $P$.

  \begin{figure}[tp]
      \centering
      \includegraphics[width=\textwidth]{ 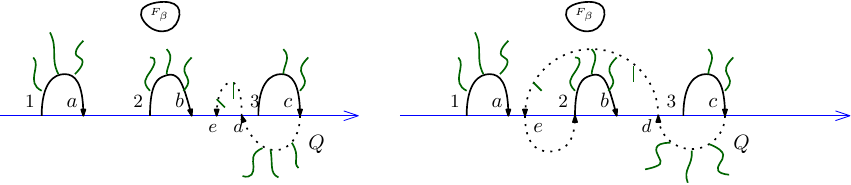}
      \caption{Case 3: $Q$ crosses $\three$ for the first time between $b$ and $3$ and then to the left of the first crossing.}\label{fig:mon_case3}
  \end{figure}
  
  \smallskip
  \noindent\textbf{Case 4: $Q$ crosses $\three$ the first time to the left of $2$.}\\
  Figure~\ref{fig:mon_case4} illustrates all the configurations of $Q$. In $A$ and $C$, we get a good loop and we can shortcut $P$ by the piece $\three[d,e]$ in $A$ and the piece of $\three$ between the crossing to left of $e$ and the right of $d$ in case $C$. In configurations $B$ and $D$, we treat $e$ as the first crossing of $Q$ with $\three$, and we land in the cases~$2$ or~$3$ based on where $Q$ crosses $\three$ after $e$. 

  \begin{figure}[tp]
      \centering
      \includegraphics[width=\textwidth]{ 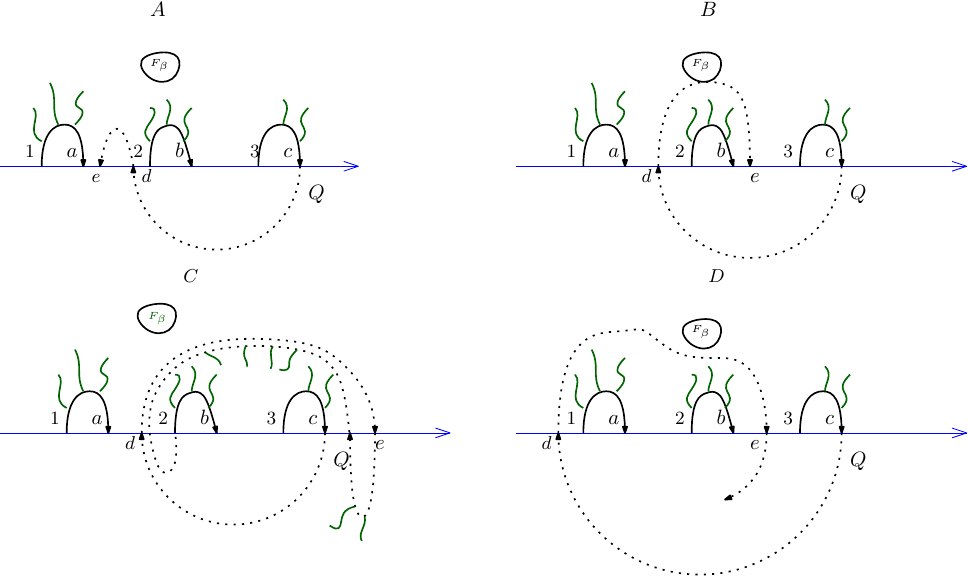}
      \caption{Case 4: First crossing of $Q$ is to the left of $2$.}\label{fig:mon_case4}
  \end{figure}
  \medskip
  
 We explicitly mention that these transformations hold good even if $\one = \two$ and $b$ is a bridge of the shrunken skeleton. So, we have argued that there are no good loops in $C^+$ and at most four consecutive bad monogons between two obstacle loops. 
 
 Now, we deal with the remaining bigons. Whenever we encounter a face of $L$ that is a bigon, we remove one of the two loops incident to it, and iterate until no bigons remain. Therefore, in the proof of~\cite[Lemma~2.1]{Chambers}, which essentially counts the number of obstacle loops, we can multiply the total number of crossings by 10 to get an upper bound on the total number of crossings of $P$ with $K$. The number of obstacle loops remaining in $L$ is at most $ 6\ell+2g - 3$, where $g$ is the genus of the surface on which $C^+$ is embedded and $\ell$ is the number of obstacles of the surface. We treat the obstacles we place on the dual vertices of the terminal faces as obstacles. Since there are $k$ terminal faces, and therefore $k$ ``obstacles'', there are $k$ obstacles in all. Therefore $\ell= k$ and $g=0$. This proves that the number of loops, and thus the degree of $p$ is \Oh{k}.
 \end{proof}

\subsection{Alternating Nerves}
Consider a shrunken bone $b$ of the shrunken skeleton of $C^+$. Then $b$ separates two terminal faces $F_\one,F_\two \in \mathcal{F}$ (possibly $F_\one = F_\two$ if $b$ is a bridge of the shrunken skeleton). While the previous section builds sufficient understanding \rev{about} the path corresponding to $b$ between consecutive nerves that extend towards the same terminal face, this does not help if two consecutive attachment points have nerves towards different terminal faces, or towards the same terminal face, but one extends east and one extends west. \rev{In particular, there may be many such nerves.} To this end, we need more elaborate methods.
\rev{Consider the following notion:}

\begin{definition}
\rev{Let $P$ be a path in $G^+$ that does not contain any augmented terminals and has some nerves attached to it that only extend towards terminal faces $F_\one,F_\two \in \mathcal{F}$ (possibly $\one=\two$).}
Let $x$ and $y$ be consecutive attachment points on $P$ (possibly $x = y$) such that $x$ is the attachment point for a nerve that extends towards $F_\one$ and $y$ is the attachment point for a nerve that extends towards $F_\two$. \rev{If $\one = \two$, then one of these nerves should extend west and the other east.} Then we call these nerves an \defi{alternating pair of nerves} \rev{with respect to $P$}.
\end{definition}
\rev{Note that the nerves are not necessarily unique in this definition (for example, there could be two nerves towards the same face in the same direction attached to $x$). However, in our upcoming analysis, there is only ever one possibility.}

\rev{Let $P$ denote the path in the skeleton of $C^+$ corresponding to the shrunken bone $b$ and let $u,v$ be the endpoints of $P$. Observe that at most one nerve can attach to each augmented dual vertex on $P$, as $u,v$ have degree at least~$2$ in the skeleton by Lemma~\ref{lem:skeleton} and each augmented dual vertex has degree at most~$3$ by Lemma~\ref{lem:augmenteddualdegree}. Conversely, the attachment point of a nerve on $P$ has degree~$3$.}

\rev{Assume first that there is an alternating pair of nerves with respect to $P$.
Let $N_u$ denote the alternating pair of nerves with respect to $P$ with attachment points $x_u,y_u$ such that $y_u$ is closest to $u$ along $P$ among all such alternating pairs of nerves and $x_u$ is on $P[u,y_u]$. Note that $x_u$ and $y_u$ must be distinct by the preceding. Let $N_v$ similarly denote the alternating pair of nerves with respect to $P$ with attachment points $x_v,y_v$ such that $y_v$ is closest to $v$ along $P$ among all such alternating pairs of nerves and $x_v$ is on $P[v,y_v]$. Again, $x_v$ and $y_v$ must be distinct.}

\rev{We first consider the case that $P[u,y_u]$ and $P[v,y_v]$ do not have a vertex in common. Then there may be many nerves attached to $P[y_u,y_v]$, both towards $F_\one$ and $F_\two$. However, we can show that this structure can be easily described and found. To that end, we first need to deal with one exceptional situation:}

\begin{figure}[t]
    \centering
    \includegraphics[width=\textwidth, keepaspectratio]{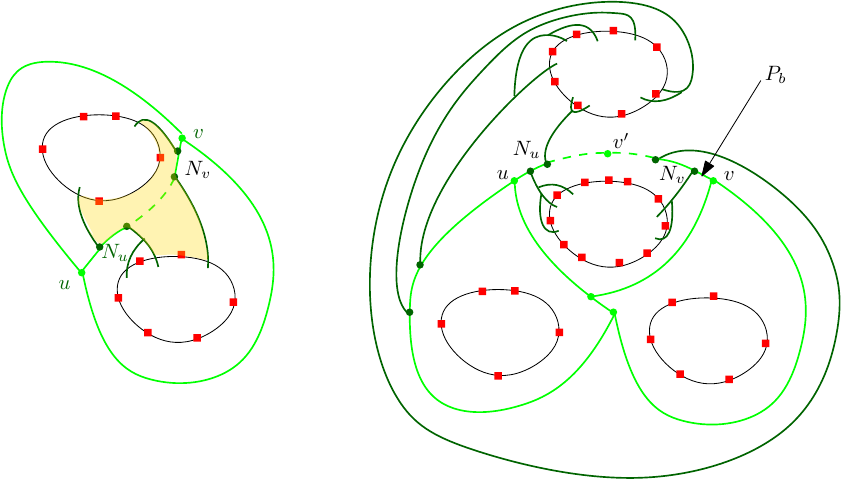}
    \caption{The figure shows the two situations when two alternating pairs of nerves occur. \rev{The situation in the right illustration is handled by Remark~\ref{rem:problem}. The situation in the left illustration is handled by Lemma~\ref{lem:MSTinside}. In the left illustration, the region $B$ exists and is shaded in yellow.}}
    \label{fig:problem}
\end{figure}

\begin{remark}\label{rem:problem}
\rev{Suppose that a subpath of $P[x_u,x_v]$ is incident to the outer face of $C^+$. Let $Q$ denote the maximal subpath of $P[x_u,x_v]$ incident to the outer face of $C^+$. Then we make one arbitrary augmented dual vertex $v'$ of $Q$ that has degree~$2$ an ``artificial'' branching point of the shrunken skeleton of $C^+$. Observe that such an augmented dual vertex of degree~$2$ exists. Indeed, $Q$ must start and end at the attachment points of two nerves, which have degree~$3$ in the augmented dual and thus correspond to a face of length~$3$ in $G$. However, since the instance is transformed, such a face only neighbors faces of length~$2$. Thus, there is an augmented dual vertex on $Q$ that has degree~$2$ in the augmented dual and is not an endpoint of $Q$. 
Adding $v'$ to the shrunken skeleton splits $P$ into two new bones. See the right illustration in Figure~\ref{fig:problem}. Note that this can happen for exactly one shrunken bone, and for the resulting bones, we cannot get into the current case (where $P[x_u,x_v]$ is incident to the outer face).}
\end{remark}

\rev{Having dealt with this exception, we consider the contraction of} the edges in $N_u$ and $N_v$ to form the vertices $u^*$ and $v^*$ respectively. \rev{This} is equivalent to deleting the corresponding edges in $G$. \rev{Since $P[y_u,y_v]$ is not incident to the outer face of $C^+$, this} forms a \rev{bounded}, connected region in the plane with pieces of the boundary of $F_\one$ and $F_\two$ bounding it on two sides. We denote this region by $B$ \rev{(see the left illustration in Figure~\ref{fig:problem})}. 

\begin{lemma}\label{lem:MSTinside}
The augmented minimum multiway cut dual $C^+$ restricted to $B$ comprises a \mst~containing the vertices $u^*, v^*$, and the augmented terminals corresponding to terminals on the pieces of the boundary of $F_\one$ and $F_\two$ bounding $B$.
\end{lemma}
\begin{proof}
The terminals embedded in the region $B$ are not connected to any terminal outside. They lie on the boundary of the outer face of $B$. Also note that the augmented dual restricted to the region $B$ is connected (because the restriction of $C^+$ to $B$ is connected) and equal to the augmented dual of the subgraph of $G$ restricted to $B$. By Corollary~\ref{cor:connected aug dual sufficient}, we know that the minimum multiway cut restricted to $B$ forms a \mst~in the augmented dual graph. The terminals of the \mst~are the augmented terminals corresponding to the intervals on the segments of $F_\one$ and $F_\two$ bounding $B$, along with $u^*$ and $v^*$. The vertices $u^*$ and $v^*$ serve as augmented vertices representing the intervals on the boundary of $B$ connecting $F_\one$ and $F_\two$. 
\end{proof}

We now consider the structure of the bone for the parts that are not in the region $B$, or what \rev{the bone looks} like when \rev{$P[u,y_u]$ and $P[v,y_v]$ do have a vertex in common or even that} there are no two alternating pairs of nerves. In the previous section, we already discussed what happens to nerve paths. Hence, we only need to consider paths between attachment points of alternating nerves and between an attachment point of a nerve and a branching point. The following result is immediate from~\cite[Lemma~5.2]{CDV} with the substitution $g=0$ and $t=k$.

\begin{lemma}\label{lem:homstring}
Each path on a bone of $C^+$ between the attachment points of a pair of alternating nerves on a bone, as well as the one from a branching point to its closest attachment point, crosses any path of $K$ \Oh{1} times.
\end{lemma}

Finally, we argue that replacing certain subpaths of a bone by a homotopically equivalent path still yields an optimal solution. The proof of this lemma is reminiscent of the proof of~\cite[Lemma~7.2]{CDV}.

\begin{lemma}\label{lemma:shortest homotopic path}\label{lem:enclosinghomotopy}
Let $b$ be any shrunken bone of the shrunken skeleton of $C^+$ and let $P$ be the path that forms the corresponding bone. Let $x$ and $y$ be two vertices of the augmented dual on $P_b$ such that $P_b[x,y]$ contains no attachment points nor branching points, except possibly $x$ or $y$. 
Then there exists a minimum multiway cut dual which contains a shortest path in the plane homotopic to $P[x,y]$.
\end{lemma}
\begin{proof}
Let $P'$ be a shortest path  with the same crossing sequence as $P=P_b[x,y]$. It is sufficient to prove that for any $t_1, t_2 \in T$, $P'$ disconnects all $t_1$-$t_2$ paths that are disconnected by $P$ and that are not disconnected by any other part of $C^+$.
 
Suppose that $R$ is a $t_1$-$t_2$ path that is crossed by $P$ but not by any other part of $C^+ \setminus P$. Let $f_1$ and $f_2$ be the faces of $C^+$ that enclose $t_1$ and $t_2$ respectively. Note that $P$ lies in the intersection of the boundaries of $f_1$ and $f_2$. Also observe that $R$ is fully contained in the region bounded by the union of $f_1$, and $f_2$. Moreover, since $t_1$ is in $f_1$ and $t_2$ is in $f_2$, $R$ must cross $P$ an odd number of times. 
Since $P$ and $P'$ have the same endpoints, the union of their drawings in the plane decomposes the plane into a set $X$ of two or more regions. Since $P$ and $P'$ are homotopic and both are disjoint from the boundaries of $F_\one$ and $F_\two$, neither $t_1$ nor $t_2$ lies in any bounded region of $X$. Hence, if $R$ enters a bounded region of $X$, then $R$ must also exit it, and thus $R$ intersects any bounded region of $X$ an even number of times. However, by the preceding argument, $R$ crosses $P$ an odd number of times. Hence, $R$ must cross $P'$. 
\end{proof}

\section{Towards an \texorpdfstring{$n^{\OO(k)}$}{n to the power O(k)}-time Algorithm}\label{Algorithms}
With the structure of the optimum solution in place, we are now ready to develop the algorithm. 
\rev{As before, we are given an instance $(G,T,\omega,\mathcal{F})$. Recall that we assume that the instance is transformed and that the dual of any optimum solution is connected. Moreover, when we speak of \emph{the} optimal solution, we mean the optimal solution specified in Section~\ref{sec:Bones and Homotopy}}.
Our guiding light will be the topology of the optimum solution, which consists of its skeleton and a compact description of its bones and is defined more formally below. As we do not know any optimum solution, we enumerate all possible topologies and argue that we find a feasible, minimum-weight solution for the topology that corresponds to an optimum solution.

\subsection{Topology}
We first define a topology and then show that all topologies can be efficiently enumerated.

\begin{definition}
A \defi{topology} consists of a triple $(S,s,h)$ \rev{for which}:
\begin{itemize}
\item $S$ is a connected plane multigraph (possibly containing parallel edges but no self-loops) with $k$ faces \rev{and at most $4k+1$ vertices}. This is the  \rev{\emph{shrunken skeleton}} of the topology. We direct the edges of the \rev{shrunken skeleton} in an arbitrary fashion to obtain a \rev{\emph{directed shrunken skeleton}}. \rev{By abuse of notation, we call the edges of $S$ the \emph{shrunken bones} of $S$};
\item \rev{there is a bijection that maps to each face of $S$, a unique corresponding face of $\mathcal{F}$};
\item for each shrunken bone $b$ of \rev{$S$}, separating faces $f_\one,f_\two$ of the skeleton (possibly $\one=\two$), a structural description, describing:
\begin{itemize}
\item a (possibly empty) ordered multisubset $s(b)$ of $\{\one,\one,\two,\two\}$ such that $s(b)[i]$ and $s(b)[i+1]$ are not equal for any $1 \leq i < |s(b)|$, except possibly for $i=2$ when $|s(b)|=4$;

\item for any $1 \leq i \leq 2|s(b)|+1$, $h(b)[i]$ is a homotopy string of length $O(k)$. When $|s(b)|=4$, then $h(b)[5]$ is unspecified.
\end{itemize}
\end{itemize}
\end{definition}
If $\one = \two$, then we use $\one$ or $\two$ to denote east and west respectively in $s(b)$.

The intuition behind the definition of a topology is as follows. The \rev{shrunken skeleton of the toplogy} guides the overall structure of the solution that will be found by our algorithm. \rev{By enumerating all possible topologies, we will encounter one that for which its shrunken skeleton} is equivalent to the shrunken skeleton of the optimal solution. \rev{By Lemma~\ref{lem:skeleton}, the shrunken skeleton of the optimal solution has $k$ faces and is connected, and by Lemma~\ref{lem:skeleton-size} and Remark~\ref{rem:problem}, it has at most $4k+1$ vertices. We thus require the same for the topology}. Moreover, each face of the shrunken skeleton \rev{of the toplogy} corresponds to a unique face in $\mathcal{F}$. The orientation of the shrunken skeleton of the topology will prove useful in later definitions and algorithms.

\rev{We seek to expand the shrunken skeleton of the topology into a multiway cut for $(G,T,\omega)$ that has the same shrunken skeleton and bijection with faces of $\mathcal{F}$ as the topology prescribes}. The multiset $s(b)$ describes the direction of maximal sets of consecutive nerves along the bone \rev{of this expansion} that have the same direction. By this intuition, it makes sense to force the directions to be alternating, as in the definition. We also note that we do not need to specify more than four such sets, because when we have four sets, we obtain two alternating pairs of nerves and can apply Lemma~\ref{lem:MSTinside} to the region between them.

For each set of consecutive nerves as described by $s(b)$, we use $h(b)$ to denote the homotopy of the (up to four) subpaths of the bone \rev{of the expansion} to which the nerves attach. We also use $h(b)$ to describe the homotopy of the (up to five) subpaths of the bone \rev{of the expansion} between these sets of nerves. Following Lemma~\ref{lem:singular face bone}, each of those homotopy strings has length $O(k)$ \rev{in the optimal solution}. When $|s(b)| = 4$, we do not (need to) specify the homotopy of the middle subpaths of the bone \rev{of the expansion}, as we handle this part as in Lemma~\ref{lem:MSTinside}.

Recall that $C$ is the optimal solution and $C^+$ is the set of corresponding edges of the augmented dual, \rev{fixed as discussed in Section~\ref{sec:Bones and Homotopy}}. 
We call a topology $(S,s,h)$ \emph{optimal} if the shrunken skeleton $S^+$ of $C^+$ is equivalent to (the underlying graph of) $S$, each face of $\mathcal{F}$ is assigned to the same face in $S^+$ and $S$, and for each \rev{shrunken} bone $b$ of the shrunken skeleton $S^+$, directed in an arbitrary fashion:
\begin{itemize}
\item there are $|s(b)|$ nonempty maximal sets of nerves, where nerves in the $i$-th set are towards $F_{s(b)[i]}$, and their attachment points are consecutive and uninterrupted on \rev{the bone of $C^+$ corresponding to} $b$; the nerves preceding a nerve from the $i$-th set belong to the $i$-th set or the $i-1$-th set, or if $|s(b)| = 4$ and $i=3$, can be arbitrary nerves to $F_\one$ or $F_\two$, until a nerve of the $i-1$-th set is hit;
\item the part of the bone \rev{of $C^+$ corresponding to $b$} between the nerves at the ends of the $i$-th set of nerves has homotopy string $h(b)[i]$;

\item the part of the bone \rev{of $C^+$ corresponding to $b$} between the $i$-th and $i+1$-th sets of nerves, for $0 \leq i \leq |s(b)|$ where we pretend the $0$-th nerve is the one end of the bone and the $|s(b)|+1$-th nerve is the other end of the bone, has homotopy string $h(b)[2i+1]$. When $|s(b)| = 4$, we do not have to satisfy this for $i=2$.
\end{itemize}
Note that since we fixed $C$, there is a unique optimal topology, which effectively describes $C$. We may thus speak of `the' optimal topology.

\rev{Recall that}, for each topology, we aim to \rev{expand this topology into a multiway cut for $(G,T,\omega)$, as mentioned above}. By considering all topologies, we ensure that we consider the optimal topology at some point. Then we argue that we find a multiway cut \rev{for $(G,T,\omega)$} of minimum weight. Before proceeding with that algorithm, we bound the number of topologies and show how to enumerate them. We require the following auxiliary lemma.

\begin{lemma}\label{lem:plane-size}
There are $2^{\OO(k \log k)}$ connected plane multigraphs with $O(k)$ vertices, $k$ faces and without self-loops. Moreover, they can be enumerated in the same time.
\end{lemma}
\begin{proof}
A trivial bound shows that there are $2^{\OO(p \log p)}$ labeled simple planar graphs on $p$ vertices (although more precise bounds are known~\cite{GimenezN2009}). 
We can extend this to planar multigraphs by guessing the subset of edges of the simple planar graph that are parallel and partitioning the number of parallel edges over them. Note that a planar multigraph with $O(k)$ vertices and $k$ faces must have $O(k)$ edges. Since there are $2^{\OO(q)}$ partitions of the integer $q$, there are $2^{\OO(k \log k)}$ possible planar multigraphs with $p=O(k)$ vertices and $k$ faces and without self-loops. All such planar graphs can be trivially enumerated in $2^{\OO(k \log k)}$ time by enumerating all graphs and testing each for planarity~\cite{HopcroftT1974}.

We now consider embeddings. Recall that two embeddings of a planar graph are equivalent if the ordering of the edges around each of the vertices is the same. In our case, there are $2^{\OO(k \log k)}$ different such orderings. Hence, there are $2^{\OO(k \log k)}$ plane multigraphs with $O(k)$ vertices and $k$ faces and without self-loops. Since each such planar graph can be enumerated in $2^{\OO(k \log k)}$ time and embeddings can be enumerated in time linear in their number~\cite{Cai93}, the bound of $2^{\OO(k \log k)}$ follows.
\end{proof}

\begin{lemma}\label{lem:topology-enumerate}
There are $2^{\OO(k^2 \log k)}$ different topologies. Moreover, they can be enumerated in the same time.
\end{lemma}
\begin{proof}
By Lemma~\ref{lem:plane-size}, there are $2^{\OO(k \log k)}$ different \rev{shrunken} skeletons for a topology. \rev{The number of bijections between the faces of a shrunken skeleton and $\mathcal{F}$ is $2^{\OO(k \log k)}$ as well.}
For each \rev{shrunken} bone $b$ of a skeleton, there are $O(1)$ choices for $s(b)$. Moreover, each homotopy string has length $O(k)$ and there are at most $8$ of them for every bone, meaning length $O(k)$ for each bone and $O(k^2)$ in total. Each entry of a homotopy string has $O(k)$ possible values by the definition of a cut graph. Hence, there are $2^{\OO(k^2 \log k)}$ topologies.

We note that by Lemma~\ref{lem:plane-size}, all \rev{shrunken} skeletons \rev{for a topology} can be enumerated in $2^{\OO(k \log k)}$ time. It is immediate from the preceding that all topologies can be enumerated in $2^{\OO(k^2 \log k)}$ time.
\end{proof}

\subsection{Broken Bones and Splints}
We now become more formal about our interpretation of topologies \rev{and their expansion}. Let $(S,s,h)$ be a topology. We define the following notion, which states how we interpret $s$ and $h$ for a shrunken bone $b$.

\begin{definition}\label{def:brokenbone}
Let $(S,s,h)$ be a topology. Let $b$ be a shrunken bone of the shrunken skeleton $S$, separating the faces $f_\one$ and $f_\two$, and directed from $u$ to $v$. Let $s(b) = \{\three_1,\ldots, \three_{|s(b)|}\}$, where each $\three_j \in \{\one, \two \}$. Let $F_\one$ and $F_\two$ be the terminal faces enclosed by the faces $f_\one$ and $f_\two$, respectively. A \defi{broken bone} is a tuple that consists of intervals $I_\one$ and $I_\two$ of augmented terminals lying on $F_\one$ and $F_\two$ respectively, vertices of the augmented dual $x^+_1,\ldots,x^+_{|s(b)|+1}$ and $y^+_0,\ldots,y^+_{|s(b)|}$ that are not augmented dual terminals, and nerves $N_1,\ldots,N_{2|s(b)|}$ towards $F_\one$ and $F_\two$, where:
\begin{itemize}
\item for each $1 \leq i \leq |s(b)|$, nerves $N_{2j-1}$ and $N_{2j}$ extend towards $\three_j$ and have root $x^+_j$ and $y^+_j$ respectively. These nerves might be the same, but are non-empty if $I_j$ is non-empty.
\item the interval of a nerve among $N_1,\ldots,N_{2|s(b)|}$ towards $F_\one$ (resp.~$F_\two$) is a subinterval of $I_\one$ (resp.~$I_\two$). Moreover, these subintervals appear in the order indicated by their indices on $I_\one$ and $I_\two$ (but do not necessarily cover $I_\one$ and $I_\two$ completely);
\item if $I_\one$ is non-empty, then a prefix of $I_\one$ is the interval of a nerve among $N_1,\ldots,N_{2|s(b)|}$. The same holds for a suffix of $I_\one$ (possibly, this is the same nerve). The same holds with respect to $I_\two$ if it contains at least two augmented terminals.
\end{itemize}
\end{definition}

Note that in the definition of a broken bone, it is not strictly necessary to specify $x^+_1,\ldots,x^+_{|s(b)|}$ and $y^+_1,\ldots,y^+_{|s(b)|}$, as their definition is implied by $N_1,\ldots,N_{2|s(b)|}$. Similarly, we do not actually need all the nerves we have specified when $|s(b)| = 4$. We still write this to streamline and simplify the definition (which admittedly is already quite complex).

We now define the notion of an optimal broken bone. We refer back to the definition of an optimal topology to recall the $|s(b)|$ maximal sets of nerves defined there. For the optimal solution $C$, the optimal topology, and a shrunken bone $b$ \rev{of the optimal topology}, we call a broken bone $x_1^+,\ldots,x^+_{|s(b)|+1}$, $y^+_0,\ldots,y^+_{|s(b)|}$, $N_1,\ldots,N_{2|s(b)|}$, $I_\one, I_\two$ \emph{optimal} if for the bone of $C^+$ corresponding to $b$, its ends are $y^+_0$ and $x^+_{|s(b)|+1}$ in the direction of the directed \rev{shrunken} skeleton, the $i$-th set of nerves has the nerves $N_{2i-1}$ and $N_{2i}$ at the ends of the set (where $N_{2i-1}$ and $N_{2i}$ are possibly equal), and the intervals $I_\one$ and $I_\two$ correspond to the union of the intervals of all nerves that have their attachment point on the bone. \rev{For nerves of $C^+$ that are attached to a branching point (this is possible if the branching point has degree~$2$ in the augmented dual), we need to decide in which broken bone we include it. We break such ties arbitrarily.}

The latter may lead to a nerve that attaches to a branching point and its corresponding interval to be assigned to two broken bones (for the two bones that meet at the branching point and whose shrunken bones bound the same face of the skeleton). In that case, we break ties arbitrarily, but in such a way that all nerves that attach to this branching point are assigned to the same bone and no terminal face is enclosed by the union of any two nerves assigned to the same bone in this way. This ensures that each augmented terminal is in some interval of some optimal broken bone and moreover, no terminal face and its broken bones `interrupt' the sequence of nerves of the bone.

\rev{We now describe how to `fix' a broken bone.}

\begin{definition}\label{def:split}
\rev{Let $(S,s,h)$ be a topology. Let $b$ be a shrunken bone of the shrunken skeleton $S$, separating the faces $f_\one$ and $f_\two$, and directed from $u$ to $v$. Let $s(b) = \{\three_1,\ldots, \three_{|s(b)|}\}$, where each $\three_j \in \{\one, \two \}$. Let $F_\one$ and $F_\two$ be the terminal faces enclosed by the faces $f_\one$ and $f_\two$, respectively. Consider a broken bone $I_\one$, $I_\two$, $x^+_1,\ldots,x^+_{|s(b)|+1}$, $y^+_0,\ldots,y^+_{|s(b)|}$, $N_1,\ldots,N_{2|s(b)|}$. Then a \defi{splint}, fixing this broken bone}, is a subset $D^+$ of the edges of the augmented dual graph $G^+$ with the following properties: 
\begin{itemize}
\item for each terminal $t$ between $I_\one$ (or between $I_\two$), there is a unique bounded face of $D^*$ that encloses $t$, where $D^*$ is the set of dual edges corresponding to the edges of $D^+$.
\item there is a path $P_b$ in $D^+$ between $y^+_0$ and $x^+_{|s(b)|+1}$. These two vertices of the augmented dual are called the \emph{ends} of the splint.
\item $x^+_1,\ldots,x^+_{|s(b)|+1}$ and $y^+_0,\ldots,y^+_{|s(b)|}$ are in $D^+$ and appear in the order $y^+_0$, $x^+_1$, $y^+_1$, $x^+_2$, $y^+_2$,~\ldots,~$x^+_{|s(b)|+1}$ on $P_b$.

\item for each $0 \leq i \leq |s(b)|$ (except $i=2$ when $|s(b)| = 4$), the (possibly empty) subpath of $P_b$ from $y^+_j$ to $x^+_{i+1}$ consists only of vertices that have degree $2$ in $D^+$ and has homotopy string equal to $h(b)[2j+1]$.

\item for each $1 \leq i \leq |s(b)|$, the only vertices on the subpath of $P_b$ between $x^+_j$ and $y^+_j$ that have degree more than~$2$ in $D^+$ are the attachment points of nerves towards $F_{\three_j}$; these nerves include $N_{2j-1}$ and $N_{2j}$. The subpath has homotopy string equal to $h(b)[2j]$.
\item if $|s(b)|=4$, then for the vertices on the subpath of $P_b$ between $y^+_2$ and $x^+_3$, the only vertices that have degree more than~$2$ in $D^+$ are the attachment points of nerves towards $F_\one$ or $F_\two$.
\item the intervals of all the aforementioned nerves jointly partition $I_\one$ and $I_\two$ with the exception that for each $1 \leq i \leq |s(b)|$, the nerves $N_{2j-1}$ and $N_{2j}$ are not necessarily distinct (but are distinct from all other nerves).
\end{itemize}
\end{definition}

We refer to Figure~\ref{fig:splinting} for an illustration of a splint and its broken bone.

\begin{figure}[t]
    \centering
    \includegraphics[width=0.85\textwidth]{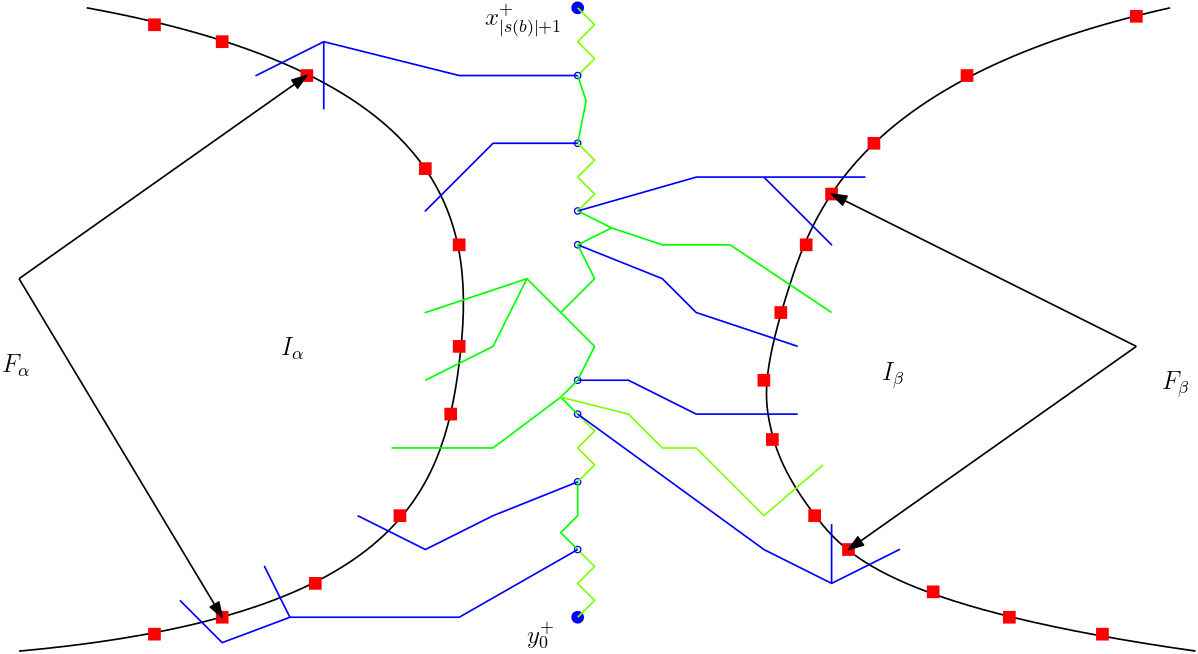}
    \caption{In this example, $s(b) = \{ \one, \two, \two, \one\}$. The blue lines depict the ``guessed'' nerves attached to the broken bone. In green, we depict the splint that is used to fix the broken bone. The arrows point to the first and last terminals of the intervals $I_\one$ and $I_\two$.}\label{fig:splinting}
\end{figure}

Finally, we call a splint \emph{optimal} if it fixes an optimal broken bone and all its nerves are the nerves of $C^+$ that were defined to belong to the bone. Again, an optimal broken bone and splint are unique with respect to $C$, so we may speak of `the' optimal broken bone and splint of a bone.

For a given topology and a particular shrunken bone of the topology, our goal is to enumerate all broken bones for this shrunken bone and find a minimum-weight splint for it. By enumerating all broken bones, we ensure that we consider the optimal broken bone. Then the minimum-weight splint that we find will have the same weight and structure of the optimal splint. We discuss that algorithm in a moment, but first argue that we can enumerate all possible broken bones efficiently.

\begin{lemma} \label{lem:enumeratebones}
There are $O(n^{26})$ distinct broken bones. \rev{They can be enumerated in the same time.}
\end{lemma}
\begin{proof}
Using Euler's formula, $G^+$ has $O(n)$ vertices of the augmented dual.
We note that any interval can be specified by two vertices of the augmented dual and (thus) any nerve by three vertices of the augmented dual. Since $|s(b)| \leq 4$, there are at most $8$ nerves and two endpoints per broken bone. It follows that there are $O(n^{26})$ distinct broken bones. \rev{Enumerating them in this time is straightforward.}
\end{proof}

\subsection{Splinting Algorithm}\label{sec:splinting}

We now describe an algorithm that, given a topology $(S,s,h)$ and a broken bone $x_1^+,\ldots,x^+_{|s(b)|+1}$, $y^+_0,\ldots,y^+_{|s(b)|}, \\ N_1,\ldots,N_{2|s(b)|}, I_\one, I_\two$ for a \rev{shrunken} bone $b$ of $S$, finds a splint of minimum weight. 

For any $0 \leq i \leq |s(b)|$ (except $i=2$ when $|s(b)|=4$), to find the path from $y^+_i$ to $x^+_{i+1}$, we invoke the algorithm by Frank and Schrijver~\cite[Section~5]{Homotopy} to find the shortest path with the homotopy string $h(b)[2i+1]$. These paths do not contain any attachment points of nerves.

The following algorithm computes the weight of the part of a minimum weight splint between the vertices $x^+_i$ and $y^+_{i}$ of the augmented dual, for all $1 \leq i \leq |s(b)|$. Note that this part must be a nerve path. For simplicity, we consider only $i= 1$, as the other cases are similar. \WLOG~$s(b)[1] = \one$. We find the part of the splint starting at $x^+_1$, with its corresponding nerve $N_{1}$, and ending at $y^+_{1}$, with its corresponding nerve $N_{2}$. Also, we assume that $N_1$ spans the interval $\{\augvertex{\one}{\ell}, \ldots, \augvertex{\one}{\ell'}\}$ and $N_2$ spans the interval $\{\augvertex{\one}{j}, \ldots, \augvertex{\one}{j'}\}$. Then the interval covered by all nerves attaching to this nerve path is $I_1 = \{^\one t_{\ell}, \ldots, ^\one t_{j'}\}$.

We compute this part of the splint as follows. By $c[x^+, a, a']$, we denote the weight of the unique nerve on $x^+ \in V(G^+)$ and the interval $\{\augvertex{\one}{a}, \ldots, \augvertex{\one}{a'}\}$, which can be computed by Lemma~\ref{lem:nerves-compute}. By $c'[x^+, a, a', e]$ we denote the weight of a partial splint passing through the vertex $x^+ \in V(G^+)$, that encloses every terminal between the augmented terminals $\{\augvertex{\one}{\ell}, \ldots, \augvertex{\one}{a'}\}$, contains the unique nerve $(x^+,a,a')$, and the partial nerve path of the homotopy given by the prefix of $h(b)[2]$ of length $e$. \rev{We wish to compute $c'$.}

The dynamic programming algorithm \rev{to compute $c'$} is given in Algorithm~\ref{alg:splinting} below. We use $d_{h(b)[2](e',e]}(x'^+,x^+)$ to denote the length of a shortest path from $x'^+$ to $x^+$ with homotopy string equal to the substring of $h(b)[2]$ between indices $e'$ and $e$ (not including the symbol on index~$e'$). In other words, this substring is equal to the prefix of length $e$ minus the prefix of length $e'$. This shortest path length can again be computed by the algorithm by Frank and Schrijver~\cite[Section~5]{Homotopy}.

We also use $\stackrel{*}{\min}$ to denote that the minimum is only allowed over certain combinations. In particular, in Line~\ref{eq:recurrence}, the nerve $(x^+,a,a')$ being considered in combination with nerve $(x'^+,z,a-1)$ and the path between $x'^+$ and $x^+$ of homotopy string $h(b)[2](e',e]$ must define a region that only encloses $\terminal{\one}{a-1}$ (the terminal inbetween the two nerves). Moreover, the nerve $(x^+,a,a')$ itself must create a region for every terminal in $\{\terminal{\one}{a},\ldots,\terminal{\one}{a'-1}\}$. A similar constraint holds in Line~\ref{eq:return}. Furthermore, if $x^+$ lies on a path $\three$ of the graph $K$, then in Line~\ref{eq:recurrence}, we find the minimum over values of $e'< \three$ and enforce the next index on the homotopy string $h(b)[2]$ to be $\three$.

\begin{algorithm}[tb!]
    \caption{Splinting Algorithm}\label{alg:splinting}
    $c'[x^+_1, \ell, \ell', 0] = c[x^+_1, \ell, \ell']$\;
    $c'[x^+,a, a',e] = \infty$ for all $x^+ \not= x^+_1$, $e \not= 0$, $a \not= \ell$, or $a' \not= \ell'$\;
    \For{$\ell' < a' < j$}{
    \For{$\ell' < a \leq a'$}{
    \For{$x^+ \in V(G^+)$}{
    \For{$0 \leq e \leq |h(b)[2]|$}
    {
    $\displaystyle
       c'[x^+, a, a', e] =  \min_{\substack{\ell \leq z < a \\ x'^+ \in V(G^+)\\ 0 \leq e' \leq e}}^*
       \Biggl\{c'[x'^+, z, a-1, e'] + c[x^+, a, a']
        + d_{h(b)[2](e',e]}(x'^+, x^+)\Biggr\} $
    \label{eq:recurrence}
     
      }
     }
    }
    }
    
    \KwRet{$\displaystyle
    \min_{\substack{\ell \leq z < j \\ x^+ \in V(G^+) \\ 0 \leq e \leq |h(b)[2]|}}^* \Biggl\{ c'[x^+,z,j-1,e] + d_{h(b)[2](e,|h(b)[2]|]}(x^+, y^+_1) + c[y_1^+,j,j']\Biggr\}$ \label{eq:return}
    }
\end{algorithm}
The same algorithm is used to compute all the other parts of the splints, too. Note that even though the algorithm only computes an optimal value, it can be easily modified to return the optimal solution (nerves and nerve path).

Finally, if $|s(b)| = 4$, then we know from Lemma~\ref{lem:MSTinside} that embedded in the region bounded by the nerves attached to $y^+_1$ and $x^+_2$, and $y^+_3$ and $x^+_4$, along with the paths between them, as well as the segments of $F_\one$ and $F_\two$ bounding the region on either side is a \mst~with its terminals on the boundary of the region. Using the algorithm of Erickson~\etal\cite{Erickson} (see also Bern~\cite{Bern}), we find the \mst. This finishes the description of the splinting algorithm.

\begin{lemma}\label{lem:splinting-algo}
Given a topology $(S, s, h)$, a bone $b$ of $S$, and the optimal broken bone for $b$, we can fix the broken bone by the splint found through the splinting algorithm. \rev{This splint has weight equal to the optimal splint.} Moreover, the splinting algorithm runs in time $n^{\OO(1)}$.
\end{lemma}
\begin{proof}
The correctness and optimality of finding the shortest path with the homotopy strings $h(b)[2i+1]$ between $y^+_i$ and $x^+_{i+1}$ for $0 \leq i \leq |s(b)|$ (except $i=2$ when $|s(b)|=4$) follows from Lemma~\ref{lemma:shortest homotopic path}. Moreover, if $|s(b)| = 4$, the correctness of finding a \mst~between the alternating nerves follows from Lemma~\ref{lem:MSTinside}.

Then, we show that the algorithm finds a feasible splint. This is immediate by the definition of $\stackrel{*}{\min}$, which ensures that the nerves cover all terminals between their intervals as well as inbetween the nerves.

Next, we argue that the algorithm finds a splint for the optimal broken bone that has the same weight as the optimal splint. We again only consider the part of the bone between $x^+_i$ and $y^+_i$ for $i=1$ and $s(b)[1] = \one$; the other cases are similar. Let $(x^+,a,a')$ be any nerve that is part of $C^+$ and where $x^+$ lies on $P_b$ between $x^+_1$ and $y^+_1$. From Lemma~\ref{lem:nerves} (see also Figure~\ref{fig:no expose}), it follows that it creates a set of regions that each enclose only a single terminal, each between the interval of the nerve.  Consider the prefix of the homotopy string of the subpath of $P_b$ between $x^+_1$ and $x^+$ and let it have length $e$. If $N_1 = (x^+,a,a')$ and (thus) $e=0$, then $c'$ contains the optimum weight of a partial splint for $x^+,a,a',e$, namely the weight of $N_1$. Otherwise, let $(x'^+,z,a-1)$ be the nerve preceding $(x^+,a,a')$ on $P_b$ and let $e'$ be the length of the homotopy string of the subpath of $P_b$ from $x^+_1$ to $x'^+$. Then the subpath of $P_b$ from $x'^+$ to $x^+$ has homotopy string $h(b)[2](e',e]$. By Lemma~\ref{lem:enclosinghomotopy}, replacing this subpath by any other path with the same homotopy string, still yields an optimal solution. Hence, the algorithm will consider and allow $z, x'^+, e'$ in the minimization for $x^+,a,a',e$ in Line~\ref{eq:recurrence} or~\ref{eq:return}. Using induction on $x'^+,z,a-1,e'$, it follows that $c'$ contains the optimum weight of a partial splint for $x^+,a,a',e$. Moreover, the algorithm returns the value of a minimum-weight splint.

It remains to argue the running time. 
Finding shortest paths with a specified homotopy using Frank and Schrijver's~\cite{Homotopy} algorithm takes $n^{\OO(1)}$ time. Finding all nerves takes polynomial time through Lemma~\ref{lem:nerves-compute}. Finding a \mst~in a planar graph, when all the terminals appear on the boundary of a single face, takes another $n^{\OO(1)}$ time using the algorithm of Erickson~\etal\cite{Erickson} (see also Bern~\cite{Bern}). Finally, we loop over all \Oh{n^5\cdot k^2} choices for $x^+,a,a',e,x'^+,z,e'$ to compute $c'$ in Line~\ref{eq:recurrence} of the algorithm and \Oh{n^2 \cdot k} choices in Line~\ref{eq:return}. Therefore, our algorithm runs in $n^{\OO(1)}$ time.  
\end{proof}

For the sake of intuition, we note that we have now gathered sufficient ideas to give a $2^{\OO(k^2 \log k)} n^{\OO(k)}$ time algorithm for {\sc Planar Multiway Cut}. In particular, we can enumerate all topologies in $2^{\OO(k^2 \log k)}$ time \rev{by Lemma~\ref{lem:topology-enumerate}} and then enumerate all \rev{combinations of} broken bones for any shrunken bone \rev{of the topology} in $n^{\OO(k)}$ time total \rev{by Lemma~\ref{lem:enumeratebones}}. \rev{For each combination of broken bones, we then compute a minimum-weight splint for each broken bone and take the union of these splints. Among all such unions that form a multiway cut for $(G,T,\omega)$, we return one of minimum weight.}
For the combination of the optimal topology and the optimal broken bones for each shrunken bone, the splinting algorithm delivers \rev{a splint for each of the optimal broken bones that has the same weight as the optimum splint by Lemma~\ref{lem:splinting-algo}}. \rev{The union of these splints forms an optimal multiway cut for $(G,T,\omega)$, since the homotopy of each nerve path and each path between branching points and the starts/ends of nerve paths are maintained by Lemma~\ref{lem:splinting-algo}. This implies a $2^{\OO(k^2 \log k)} n^{\OO(k)}$ time algorithm for {\sc Planar Multiway Cut}.} In the next section, we argue how to reduce the $n^{\OO(k)}$ factor down to $n^{\OO(\sqrt{k})}$.

\section{Algorithm Using Sphere-Cut Decomposition}\label{sec:sc}
\sectionmark{Sphere-cut decomposition}
\rev{As before, we are given an instance $(G,T,\omega,\mathcal{F})$. Recall that we assume that the instance is transformed and that the dual of any optimum solution is connected. Consider the optimal solution specified in Section~\ref{sec:Bones and Homotopy}. Now consider the optimal topology and the optimal broken bones and consider the splint computed by Lemma~\ref{lem:splinting-algo} for each optimal broken bone. Since the homotopy of each nerve path and each path between branching points and the starts/ends of nerve paths are maintained by Lemma~\ref{lem:splinting-algo}, it follows that we can replace each bone and attached nerves of the optimum solution by the minimum-weight splint computed by the algorithm for the corresponding optimal broken bone. In particular, the homotopy strings ensure that the number of crossings of the bones with the cut graph $K$ stays minimum. Throughout this section, we let $C$ denote this optimal solution (after the replacement by splints) and let $C^+$ denote the set of corresponding augmented dual edges. We call this the \emph{splinted} solution.}

We are now ready to discuss how we go from a topology to a multiway cut. If the topology is optimal, we argue that we find a minimum-weight multiway cut that has the same structure as the optimal multiway cut $C$. Our aim is to employ Theorem~\ref{thm:sc} on the shrunken skeleton to get a small sphere-cut decomposition and then apply a dynamic program. \rev{This dynamic program use as state (or indices) the tuples that describe parts of the broken bones, but only for those shrunken bones that are incident to the noose of the current edge of the sphere-cut decomposition. See e.g.\ Figure~\ref{fig:partial-stretcher}.} 

However, as already noted a few times, the shrunken skeleton might have bridges. Hence, Theorem~\ref{thm:sc} cannot be applied directly, but only on the bridge blocks of the shrunken skeleton. Therefore, we first develop a dynamic program that combines solutions of the bridge blocks of the shrunken skeleton, effectively reducing the problem.

\subsection{Reduction to Bridge Blocks}\label{sec:bridge-block-reduction}
When we want to combine solutions of different bridge blocks, it is natural to use the bridge block tree in a dynamic program. However, if we do this naively, we immediately run into the issue that in order to compute a solution for a non-trivial bridge block, we need to know the solutions for all bridge blocks contained in its bounded faces. This can be resolved by enforcing an ordering on the computation of the bridge blocks that depends on the embedding. To that end, we proposed the embedding-aware bridge block (EABB) tree in Section~\ref{sec:eabb}. We now show how to use it.

Let $(S,s,h)$ be a topology. Let $L = L(S)$ be the EABB tree for $S$ and let $\mathcal{B}$ denote the set of bridge blocks of $S$. For a BB-node $l$ of $L$, let $\mathcal{B}(l)$ denote the bridge block corresponding to $l$. Extending this notation, for a subtree $L'$ of $L$, we use $\mathcal{B}(L')$ to denote the set of all bridge blocks corresponding to BB-nodes in $L'$. For a node $l$ of $L$, we use $L_l$ to denote the subtree of $L$ rooted at $l$. In particular, $L_{\ell(L)} = L$, where we recall that $\ell(L)$ is the root of $L$.

\begin{lemma}\label{lem:eabbt:internal}
Let $l$ be any node of $L$. Let $v$ be the cut vertex corresponding to $l$'s closest $C$-node ancestor in $L$ (which may be $l$). Then $\mathcal{B}(L_l)$ is an internal set of $S$.
\end{lemma}
\begin{proof}
By Lemma~\ref{lem:bridge-block:precp}, if there is a bridge block $B$ and a bridge block $B' \in \mathcal{B}(L_l)$ such that $B \precp B'$, then $B \in \mathcal{B}(L_l)$. Hence, any bridge block not in $\mathcal{B}(L_l)$ is in the outer face of each of the blocks in $\mathcal{B}(L_l)$. Hence, there is a single face of $\mathcal{B}-\mathcal{B}(L_l)$ that encloses $\mathcal{B}(L_l)$ and there is a single face of $\mathcal{B}(L_l)$ that encloses $\mathcal{B}-\mathcal{B}(L_l)$. \rev{These faces meet at $v$}. It follows that $\mathcal{B}(L_l)$ is an internal set.
\end{proof}

We now perform a bottom-up dynamic programming with respect to $L$. The crux here is to understand the relation that a child $l$ has with its parent. This is formed by two parts. The first and easier part is the cut vertex shared by neighboring bridge blocks. The second, more complicated part, is the terminal face $F_l$ in the middle region induced by the internal set $\mathcal{B}(L_{l})$. The terminals in $T_l$ are covered jointly by $\mathcal{B}(L_{l})$, the blocks induced by siblings of $l$ in $L$, and by the block induced by $l$'s parent. Using a similar argument as in the proof of Lemma~\ref{lem:nerves}, we can see that each of these \rev{sibling blocks} is responsible for a single interval of $T_l$, \rev{while the parent block is responsible for all intervals in between those of the siblings}. To help in the computations, we additionally consider the first nerves that cover the prefix and suffix of this interval. We now expand on this intuition of the dynamic program and define the table more formally.

Let $l$ be a node of $L$. Define $w = w(l)$ as follows: if $l$ is a C-node, then let $w$ be the corresponding cut vertex; if $l$ is a BB-node and $l$ has a parent in $L$, then this parent is a C-node and we let $w$ be its corresponding cut vertex; otherwise, let $w$ be any vertex of $\mathcal{B}(l)$. Let $F_l \in \mathcal{F}$ denote the unique terminal face in the middle region of $\mathcal{B}(L_l)$.

\begin{definition}\label{def:stretcher}
Given a vertex $w^+$ of the augmented dual, an interval $I_l$ of $F_l$, and two (possibly empty) nerves $N^1_l,N^2_l$ towards $F_l$, a \defi{stretcher} is a set $D^+$ of edges of the augmented dual such that:
\begin{enumerate}[label= \emph{(\roman*)}]
\item\label{def:stretcher:intervals} the interval of $N^1_l$ is a prefix of $I_l$ and the interval of $N^2_l$ is a suffix of $I_l$. $N^1_l$ and $N^2_l$ are either both empty or both non-empty and cannot be empty if there is at least one terminal between $I_l$.
\item there is a (possibly empty) set of nerves in $D^+$  towards $F_l$ whose augmented terminal sets are intervals that jointly partition $I_l$. $N^1_l$ and $N^2_l$ are among those nerves; 
\item for any terminal $t \in T_l$ between $I_l$, there is a bounded face of $D^*$ that encloses only $t$ and no other terminals. Here $D^*$ is the set of dual edges corresponding to the edges of the augmented dual in $D^+$;

\item for any terminal $t \in T_\one$ of a bounded face $f_\one$ in the subgraph of $S$ induced by the bridge blocks of $L_{l}$, there is a bounded face of $D^*$ that encloses only $t$ and no other terminals;
\item $D^+$ contains a splint for each bone in $\mathcal{B}(L_l)$ and $N^1_l$ and $N^2_l$ are included in these splints;

\item $w^+$ is the vertex of $D^+$ that is an end of all splints for the bones in $\mathcal{B}(L_l)$ that are incident to $w$.
\end{enumerate}
We call $w^+$, $I_l$, $N^1_l$, and $N^2_l$ a \defi{binder} of the stretcher and a \emph{binder} for $l$. We say a binder is \emph{valid} if it adheres to~\ref{def:stretcher:intervals}.
\end{definition}

\begin{proposition}\label{prp:binders:number}
For each node $l$ of $L$, there are $n^{\OO(1)}$ binders. These can be enumerated in the same time.
\end{proposition}
\begin{proof}
It suffices to observe that a binder has one augmented dual vertex, an interval of a terminal face, and two nerves, each of which can be described by a constant number of vertices of the augmented dual.
\end{proof}

Consider any vertex $l$ of the embedding-aware block-cut tree $L$ of the shrunken skeleton of an optimum solution $C^+$. Let $I_l$ be the union of intervals of the nerves in $C^+$ towards $F_l$ that attach to a block in $\mathcal{B}(L_l)$ and let $N^1_l$ and $N^2_l$ be the nerves whose intervals are a prefix and suffix of $I_l$ respectively. \rev{Since $C^+$ is splinted}, $C^+$ is a union of splints for each of the shrunken bones of $S$. Let $w^+$ be the vertex of $C^+$ that is an end of all splints for the bones in $\mathcal{B}(L_l)$ that are incident to $w =w(l)$. Then we call $w^+$, $I_l$, $N^1_l$, and $N^2_l$ an \emph{optimal binder}. Note that for each optimal binder, a stretcher does exist, which we call an \emph{optimal stretcher}. We may speak of `the' optimal binder and stretcher, because the minimum-weight solution and topology are uniquely defined.

Then for each binder $I_l,N^1_l,N^2_l,w^+$, define $A_l[I_l,N^1_l,N^2_l,w^+]$ as a minimum-weight stretcher for this binder. If no such stretcher exists, then we define $A_l[I_l,N^1_l,N^2_l,w^+]$ to be the set of all edges of the augmented dual. We argue that the table $A$ can be computed in a dynamic programming fashion.

In the next section, we prove the following.

\begin{lemma}\label{lem:binders:partial-dp}
For any BB-node $l$ of $L$ and the optimal binder $B$ for $l$, we can compute a minimum-weight stretcher for $l$ when given minimum-weight stretchers for all optimal binders of all children of $l$. Moreover, it can be computed in $n^{\OO(\sqrt{|\mathcal{B}(l)|})}$ time.
\end{lemma}

We now describe how to compute a table entry for $I_l$ and $w^+$ if $l$ is a C-node. Let $l_1,\ldots,l_q$ denote the children of $l$ in $L$. Note that the children of $l$ are all BB-nodes. We assume that the bridge blocks appear in this order around $w$ in $S$. Then we compute $A_l[I_l,N^1_l,N^2_l,w^+]$ as the minimum-weight union of $A_{l_j}[I_{l_j},N^1_{l_j},N^2_{l_j},w^+]$ over all families $I_{l_1},\ldots,I_{l_q}$ of (possibly empty) intervals of $F_l$ that form a partition of $I_l$ and appear in this order on $I_l$ and over all nerves $N^1_{l_j},N^2_{l_j}$ whose interval is a prefix respectively suffix of $I_{l_j}$. During this computation, we discard any union that does not form a stretcher. If all unions are discarded in this way, we set $A_l[I_l,N^1_l,N^2_l,w^+]$ equal to the set of all edges of the augmented dual.

\begin{lemma}\label{lem:binders:dp}
For any C-node $l$ of $L$ and the optimal binder $B$ for $l$, we can compute a minimum-weight stretcher when given minimum-weight stretchers for all optimal binders of all children of $l$. Moreover, the table $A$ will store a stretcher of minimum weight for the optimal binder. Finally, it can be computed in $\OO(qn^{\OO(1)})$ time
\end{lemma}
\begin{proof}
Consider an optimal binder for $l$ and optimal binders for $l_1,\ldots,l_q$. By assumption, $A$ contains a minimum-weight stretcher for the optimal binder for $l_1,\ldots,l_q$. We only need to be concerned with the terminals $t$ inbetween consecutive non-empty intervals $I_{l_j}$ and $I_{l_{j'}}$, where $j < j'$. That is, $I_{l_j}$ and $I_{l_{j'}}$ are non-empty intervals and there is no $j < j'' < j'$ such that $I_{l_{j''}}$ is a non-empty interval. The remainder follows by definition. So consider such a terminal $t$. Consider the nerves $N^2_{l_j}$ and $N^1_{l_{j'}}$ of the optimal binders for $l_j$ and $l_{j'}$. Now follow the nerve paths of the splints belonging to bones bordering $f_j$ contained in $\mathcal{B}_{l_{j''}}$ for all $j \leq j'' \leq j'$ of the stretcher stored in $A$. This path has the same homotopy as in the optimal solution and goes between the same vertices as in the optimal solution, namely the roots of $N^2_{l_j}$ and $N^1_{l_{j'}}$. Together with the sides of $N^2_{l_j}$ and $N^1_{l_{j'}}$, which are as in the optimum, this path thus yields a region that encloses $t$ and no other terminals by Lemma~\ref{lemma:shortest homotopic path}. Hence, the algorithm computes a stretcher for the optimal binder.

Finally, it remains to argue that the computed stretcher has minimum weight. We note that the optimal stretcher (for the optimal binder) can be decomposed into optimal stretchers (for the optimal binders) for each of the children $l_1,\ldots,l_q$ by the definition of optimality. Then, it follows by the description of the algorithm that $A$ stores a minimum-weight stretcher for the optimal binder.

Note that a trivial implementation would compute this minimum in $\OO(n^{\OO(q)})$ time by enumerating all partitions of $I_l$ and all nerves. However, using a simple dynamic program using partial unions, this can be reduced to $\OO(qn^{\OO(1)})$ time.

\end{proof}

\subsection{Algorithm for a Bridge Block}\label{sec:bridge-block}
We now set out to prove Lemma~\ref{lem:binders:partial-dp}. 
Consider any BB-node $l$ of $L$. It is either a single edge or a connected and bridgeless graph without self-loops. Hence, it has a sphere-cut decomposition by Theorem~\ref{thm:sc}. Assuming that the table $A$ has been computed for all children of $l$ in $L$ (which are C-nodes), we compute the table entry for $l$. By Lemma~\ref{lem:eabbt:internal}, $\mathcal{B}(L_l)$ is an internal set. Let $F_l$ be the terminal face in the middle region of $\mathcal{B}(L_l)$. Consider a binder $I_l,N^1_l,N^2_l,w^+$ for $l$. We assume that the binder is valid. We now wish to compute a stretcher for this binder by a dynamic program over the sphere-cut decomposition.

An important part of the dynamic program is how to incorporate the solutions for the children of $l$ in $L$. Since we effectively consider $\mathcal{B}(l)$ as a collection of (shrunken) bones, we need to associate a bone with each child to ensure this. We now make this more formal.
For each C-node of $L$ that is child $l'$ of $l$, corresponding to a cut vertex $c$, consider the middle region $f_{l'}$ of $\mathcal{B}(L_{l'})$ and let $b,b'$ be two of the bones on the boundary of $\mathcal{B}(l)$ that are incident to $c$. Since $l'$ is a child of $l$, it follows from the definition of an EABB tree that $b$ and $b'$ are well defined (and possibly $b=b'$). Pick one of $b, b'$ in a consistent manner (say $b$) and associate $l'$ with this bone. We call $l'$ an \emph{associated child} of the bone $b$, the cut vertex $c$, and the middle region $f_{l'}$. Later, when we propose the algorithm, we will discuss how to specify the binders for associated children.

Now let $(R,\eta,\delta)$ be a sphere-cut decomposition of $\mathcal{B}(l)$;  refer back to Section~\ref{sec:sphere-cut} for the definitions. 
For every pair $x,y$ of adjacent vertices of $R$ such that $y$ is the parent of $x$, consider the noose $\noose = \delta(x,y)$. Let $\mathcal{F}_{\noose} \subseteq \mathcal{F}$ be the set of terminal faces for which the corresponding face of $S$ is intersected by $\noose$ and let $\mathcal{F}_{\enc(\noose)} \subseteq \mathcal{F}$ be the set of terminal faces for which the corresponding face of $S$ is enclosed by $\enc(\noose)$. We include in $\mathcal{F}_{\enc(\noose)}$ the bridge blocks of any associated children of any bone enclosed by $\enc(\noose)$. Note that $\mathcal{F}_{\enc(\noose)}$ is possibly empty, but $\mathcal{F}_{\noose}$ never is.

We now describe a bottom-up dynamic programming algorithm that aims to compute a stretcher of minimum weight for the given binder. To develop this algorithm, we first need a notion of what the partial solution is that we compute during the algorithm.

To this end, we need the notion of a partial binder and a partial stretcher. The intuition is that a partial binder is a dynamic programming state for the intersection of a noose $\noose$ with the hypothetical solution prescribed by the topology. For each vertex of the topology intersected by $\noose$, the state stores a corresponding vertex of the augmented dual. For each terminal face $F_\one$ in $\mathcal{F}_{\noose}$, we only see part of the hypothetical solution, which is described by an interval of $T_\one$ and the (possibly empty) first and last nerves that (possibly together with other nerves) cover the interval. Then a corresponding partial stretcher is the entry stored in the dynamic programming table for the partial binder, which essentially stores a solution for all terminal faces in $\mathcal{F}_{\enc(\noose)}$ and a partial solution for all terminals between the mentioned interval.

An important situation occurs when $F_l$, the face in the middle region of the internal set $\mathcal{B}(L_l)$, is in $\mathcal{F}_{\noose}$. Recall that the binder specifies an interval of $T_l$ (and certain nerves) that must be covered by the stretcher that we are computing. Hence, in this case, the partial solution needs to satisfy (at least, partially) these demands as well. This constraint is met in part~\ref{def:partial-stretcher:intervals} of Definition~\ref{def:partial-stretcher} below.

We now define this intuition more formally. Throughout, if $w_1,\ldots,w_q$ are the $q$ vertices of $\mathcal{B}(l)$ intersected by a noose $\noose$, then we assume that the faces of $\mathcal{F}_{\noose}$ are numbered $F_1,\ldots,F_q$, where $F_j$ is the terminal face corresponding to the face of $\mathcal{B}(l)$ intersected by the part of $\noose$ between $w_j$ and $w_{j+1}$ (or $w_1$ if $j=q$).

\begin{figure}[t]
    \centering
    \includegraphics[width=0.75\textwidth]{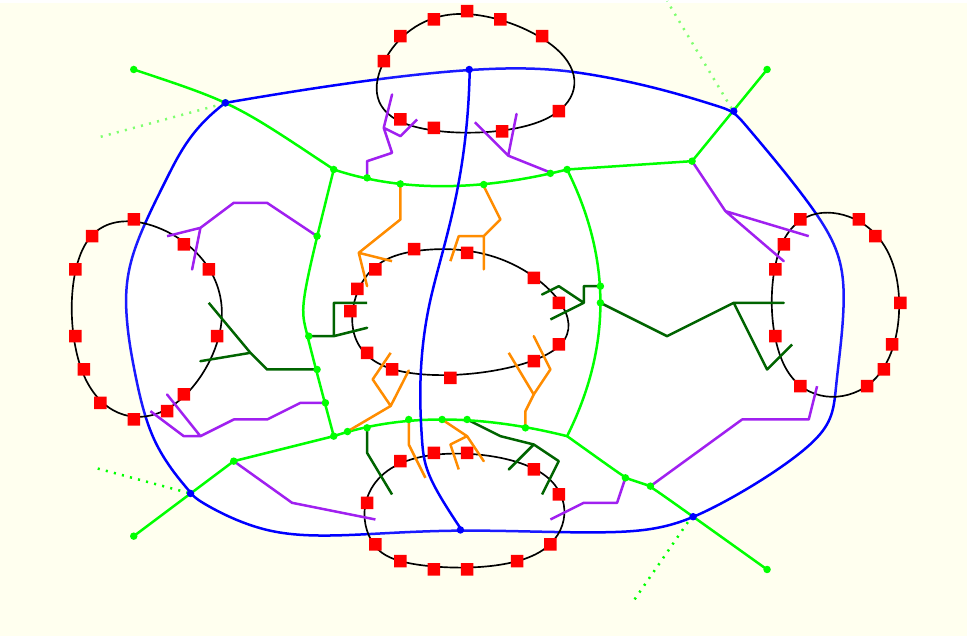}
    \caption{\rev{The figure illustrates an instance of \problemPMWC, with five terminal faces shown (the terminals are red squares). A part of the skeleton of the optimum solution is drawn in green. The blue cycles denote the nooses of the sphere-cut decomposition: for an edge $(x,y)$ of the decomposition with child edges $(y_1,x)$ and $(y_2,x)$, the outer blue cycle denotes the noose $\delta(x,y)$, the left blue cycle (say) the noose $\delta(y_1,x)$, and the right blue cycle the noose $\delta(y_2,x)$. For the noose $\delta(x,y)$, some of the information stored in a partial binder for it are shown: the outermost nerves $N_1$ and $N_2$ of each face intersected by the noose are shown in purple. For the nooses $\delta(y_1,x)$ and $\delta(y_2,x)$, the outermost nerves unique to their partial binders are shown in orange. The nerves in dark green are not part of partial binders of any of the three nooses.}}
    \label{fig:partial-stretcher}
\end{figure}

\begin{definition}\label{def:partial-stretcher}
Given a binder $I_l,N^1_l,N^2_l,w^+$ for $l$, the noose $\noose = \delta(x,y)$, vertices of the augmented dual $w^+_1,\ldots,w^+_q$, where $q = |\mathcal{F}_{\noose}|$, intervals $I_1,\ldots,I_q$ of $F_1,\ldots,F_q$ respectively, and (possibly empty) nerves $N^1_1,N^2_1,\ldots,N^1_q,N^2_q$ towards $F_1,\ldots,F_q$ respectively, a \defi{partial stretcher} for $(x,y)$ is a set $D^+$ of edges of the augmented dual such that:
\begin{enumerate}[label= \emph{(\Roman*)}]
\item\label{def:partial-stretcher:intervals} the intervals of $N^1_1,\ldots,N^1_q$ are a prefix of $I_1,\ldots,I_q$ (if non-empty) and the intervals of $N^2_1,\ldots,N^2_q$ are a suffix of $I_1,\ldots,I_q$ (if non-empty) respectively. For each $1 \leq i \leq q$, $N^1_i$ and $N^2_i$ are either both empty or both non-empty and cannot be empty if $I_i$ has more than one terminal. If $F_l \in \mathcal{F}_{\noose}$, say $F_l=F_j$, then: $I_j$ is a subinterval of $I_l$; if $I_j$ is a prefix of $I_l$, then $N^1_j = N^1_l$; if $I_j$ is a suffix of $I_l$, then $N^2_j = N^2_l$;
\item there is a (possibly empty) set of (possibly empty) nerves in $D^+$ towards $F_1,\ldots,F_q$ whose corresponding intervals jointly partition $I_1,\ldots,I_q$. $N^1_1,\ldots,N^1_q$ and $N^2_1,\ldots,N^2_q$ are among those nerves;
\item for each terminal $t$ between $I_1$, or $\ldots$, or $I_q$, there is a bounded face of $D^*$ that encloses only $t$ and no other terminals of $T$. Here $D^*$ is the set of dual edges corresponding to the edges of the augmented dual in $D^+$;
\item for any terminal $t \in T_\one$ of a terminal face $F_\one \in \mathcal{F}_{\enc(\noose)}$, there is a bounded face of $D^*$ that encloses only $t$ and no other terminals of $T$;

\item $D^+$ is the union of a set of splints, one for each bone enclosed by $\enc(\noose)$. $N^1_1,N^2_1,\ldots,N^1_q,N^2_q$ are included in these splints;

\item if $w$ is in $\enc(\noose)$, then $w^+$ is the vertex of $D^+$ that is an end of all splints for the bones in $\enc(\noose)$ that are incident to $w$;
\item $w^+_1,\ldots,w^+_q$ are the ends of splints in $D^+$ that correspond to $w_1,\ldots,w_q$.
\end{enumerate}
We call $w^+_1,\ldots,w^+_q$, $I_1,\ldots,I_q$, and $N^1_1,\ldots,N^1_q,N^2_1,\ldots,N^2_q$ a \defi{partial binder} of the partial stretcher and of $(x,y)$. We say a partial binder is \emph{valid} if it adheres to Property~\ref{def:partial-stretcher:intervals}.
\end{definition}
\rev{The definition of a partial binder is illustrated in Figure~\ref{fig:partial-stretcher}.} 

\begin{proposition}\label{prp:partial-binders:number}
For any binder and noose $\noose = \delta(x,y)$, there are $n^{\OO(q)}$ distinct partial binders, where $q = |\mathcal{F}_{\noose}|$. These can be enumerated in the same time.
\end{proposition}
\begin{proof}
It suffices to observe that a partial binder has $\OO(q)$ vertices of the augmented dual, intervals of terminal faces, and nerves, each of which can be described by a constant number of vertices of the augmented dual.
\end{proof}

For the dynamic program, we define a function $Z$ that assigns to any partial binder a minimum-weight partial stretcher. If such a partial stretcher does not exist, then $Z$ assigns the set of all edges of the augmented dual. Below, we show how we compute the table $Z$ and argue that a minimum-weight stretcher is computed. To that end, we perform dynamic programming on the sphere-cut decomposition and describe two cases: a base case on a leaf of the branch decomposition, and an inductive case.

\paragraph{Dynamic Program: Base Case.}
Consider the base case, when $x$ is a leaf. Since $\mathcal{B}(l)$ is bridgeless and connected, $\midset(x,y)$ consists of the endpoints of a single bone $b = (w_1,w_2)$, enclosed by the noose $\noose = \delta(x,y)$. Let $w^+,I_l, N^1_l,N^2_l$ be a valid binder and $w^+_1,w^+_2$, $I_1,I_2$, and $N^1_1,N^1_2,N^2_1,N^2_2$ be a partial binder. We assume that the partial binder is valid, \ie~it satisfies property~\ref{def:partial-stretcher:intervals} of Definition~\ref{def:partial-stretcher}. We also assume that if $w$ is in $\enc(\noose)$, then $w^+ = w^+_j$ when $w = w_j$ for $j\in\{1,2\}$. Let $F_\one$ and $F_\two$ be the two terminal faces separated by $b$.

An important consideration here is how to deal with associated children of $b$. 

The partial binder makes the partial solution for $b$ responsible for covering the terminals between $I_1$ and $I_2$, but part of this responsibility can be delegated to the associated children of $b$. We only need to specify which subintervals of $I_1$ and $I_2$ is covered by the associated children by splitting it. We now formalize this intuition.

Suppose $l'$ is an associated child of $b$. Say it is associated with $w_1$ and face $F_\one$. Without loss of generality, $b$ follows $w_1$ in the face ordering. Then we enumerate all possible intervals $I_1'$ and $I_1''$ of $F_\one$ that partition $I_1$ (with $I_1'$ preceding $I_1''$) and enumerate all possible (possibly) nerves $N_1'$ and $N_1''$ towards $F_\one$, such that $w_1^+$, $I_1'$, $N^1_1$, and $N_1'$ is a valid binder for $l'$ (called a \emph{split-parameterized binder}) and $w^+_1,w^+_2$, $I_1'',I_2$, and $N_1'',N^1_2,N^2_1,N^2_2$ is a valid partial binder (a \emph{split partial binder}). Then we call $I_1'$, $I_1''$, $N_1'$, and $N_1''$ a \emph{splitter} for the partial binder with respect to the associated child $l'$. Tying this to our earlier intuitive understanding, the splitter effectively specifies which part of $I_1$ is covered by the bone and which part by the associated child. This specification leads to a binder for the associated child (the split-parameterized binder) and a new partial binder for $b$ (the split partial binder). We will later consider every possible splitter and take the best solution we find.

From now on, we assume that the partial binder is split with respect to all associated children of $b$ (we might call it split even if $b$ has no associated children). By abuse of notation, we still use the same variables for it.

We now enumerate all broken bones for $b$ for which $I_\one = I_1$, $I_\two = I_2$, $y_0^+ = w^+_1$, $x^+_{|s(b)|+1} = w^+_2$, and nerves $N_1,\ldots,N_{2|s(b)|}$ that correspond to $N^1_1, N^1_2$ and $N^2_1, N^2_2$. Here we mean by `correspond' that for the smallest $j \in \{1,\ldots,|s(b)|\}$ for which $s(b)[j] = \one$, it holds that $N_{2j-1} = N^1_1$ and for the largest $j \in \{1,\ldots,|s(b)|\}$ for which $s(b)[j] = \one$, it holds that $N_{2j} = N^2_1$; a similar condition holds with respect to $\two$. We say that such a broken bone \emph{specializes} the partial binder. Conversely, the partial binder \emph{generalizes} the broken bone. Note that there is a unique partial binder generalizing a broken bone, while this is not true for the converse. Then we use Algorithm~\ref{alg:splinting} to compute a minimum-weight splint for each broken bone that specializes the split partial binder. Finally, set $Z$ of the partial binder to be equal to a minimum-weight union of the splint found in this manner for the split partial binder and the stretchers stored in $A$ for the split-parameterized binders. This minimum is optimized over all splitters for which the aforementioned union forms a partial stretcher. If no such union yields a partial stretcher, set $Z$ to be equal to the set of all edges of the augmented dual.

For the optimal topology and the optimal binder for a node $l$ of $L(S)$, a partial binder is \emph{optimal} if $C^+$ contains a partial stretcher for this partial binder where each of the splints of the partial stretcher (mentioned in part V) are optimal. Since $C^+$ is splinted, we may assume that any optimum solution indeed consists of splints. We then call this partial stretcher \emph{optimal} as well. It follows that, for a bone $b$, a split partial binder is \emph{optimal} if it generalizes the optimal broken bone for $b$.
A splitter is \emph{optimal} if the resulting split-parameterized binders and split partial binder are optimal.

\begin{proposition}\label{prp:partial-binders:base}
Let $x$ be a leaf and $b$ the corresponding bone whose endpoints are in $\midset(x,y)$. Consider any broken bone specializing a valid partial binder $B$ for $(x,y)$ and if $w$ is in $\enc(\noose)$, then $w^+ = w^+_j$ when $w = w_j$ for $j\in\{1,2\}$. If $B$ is split or $b$ has no associated children, then any splint for the broken bone is a partial stretcher for $B$. If $b$ has associated children, $B$ is optimal, and the splitter is optimal, then the union of a splint for the broken bone and the stretchers stored in $A$ for the optimal split-parameterized binders is a partial stretcher.
Finally, for the optimal partial binder, the table entry $Z(B)$ will store a partial stretcher of minimum weight.
\end{proposition}
\begin{proof}
We verify that all properties of Definition~\ref{def:partial-stretcher} hold. Property~\ref{def:partial-stretcher:intervals} holds by definition. Property~(II) and (III) follow by the definition of a splint for a broken bone specializing the partial binder. Property~(IV) follows from the definition of a splint. Property~(V) follows from the definition of a splint for a broken bone specializing the partial binder. Property~(VI) follows by assumption and the definition of a splint. Finally, Property~(VII) follows by the definition of a splint for a broken bone specializing the partial binder.

If $b$ has associated children, then we augment the argument for Property~(III) and~(IV). Indeed, Property~(III) is satisfied for all terminals except possibly the terminal $t$ inbetween $I_1'$ and $I_1''$ of the optimal splitter. However, $w_1^+$, $N_1'$, and the attachment point of $N_1''$ are as in the optimum and thus the path between $w_1^+$ and the attachment point of $N_1''$ has the same homotopy as in the optimum (by the definition of a splint) and goes between the same vertices as the optimum. Together with the flanks of $N_1'$ and $N_1''$, which are as in the optimum, this path thus yields a region that encloses $t$ and no other terminals by Lemma~\ref{lemma:shortest homotopic path}. Property~(IV) follows by the definition of a partial stretcher applied to the associated children.

For the final part, we note that the preceding establishes that for the optimal (original) partial binder and optimal splitter, the algorithm yields a partial stretcher. Note that the optimal broken bone is among the broken bones specializing the optimal partial binder by definition. Then, by the optimality of the splinting algorithm, the resulting splint for the optimal split partial binder is optimal. Then, by the optimality of the table $A$ for each associated child and of the splinting algorithm, the resulting splint and partial stretcher will have minimum weight.
\end{proof}

\paragraph{Dynamic Program: Inductive Case.}
Now consider the inductive case, when $x$ is an internal vertex. Let $w^+, I_l, N^1_l,N^2_l$ be a valid binder. Let $w^+_1,\ldots,w^+_q$, $I_1,\ldots,I_q$, and $N^1_1,\ldots,N^1_q,N^2_1,N^2_q$ be a partial binder $B$ for $(x,y)$. We assume that the partial binder satisfies property~\ref{def:partial-stretcher:intervals} of Definition~\ref{def:partial-stretcher}. We also assume that if $w$ is on $\noose$, then $w^+ = w^+_j$ when $w = w_j$ for $j\in\{1,\ldots,q\}$.
Let $v^+_1,\ldots,v^+_r$, $J_1,\ldots,J_r$, and $M^1_1,\ldots,M^1_r,M^2_1,\ldots,M^2_r$ be a valid partial binder $B_1$ for $(y_1,x)$ and let
$u^+_1,\ldots,u^+_s$, $K_1,\ldots,K_s$, and $L^1_1,\ldots,L^1_s,L^2_1,\ldots,L^2_s$ be a valid partial binder $B_2$ for $(y_2,x)$. Let $\nooseone = \delta(y_1,x)$ and $\noosetwo = \delta(y_2,x)$. Note that the terminal faces in $\mathcal{F}_{\noose}$, $\mathcal{F}_{\nooseone}$, and $\mathcal{F}_{\noosetwo}$ can be categorized into four classes: those that appear in all three (of which there are at most two), those that appear in $\mathcal{F}_{\noose}$ and $\mathcal{F}_{\nooseone}$ but not $\mathcal{F}_{\noosetwo}$, those that appear in $\mathcal{F}_{\noose}$ and $\mathcal{F}_{\noosetwo}$, but not in $\mathcal{F}_{\nooseone}$, and those that appear in $\mathcal{F}_{\nooseone}$ and $\mathcal{F}_{\noosetwo}$, but not in $\mathcal{F}_{\noose}$.

We say that these partial binders \defi{match} if
\begin{itemize}
\item for each terminal face $F_j \in \mathcal{F}_{\noose}$ that is in $\mathcal{F}_{\nooseone}$ (say it is also numbered $j$ in $\mathcal{F}_{\nooseone}$) but not in $\mathcal{F}_{\noosetwo}$, the state for this face is the same in $B$ and $B_1$. Formally, $I_j = J_j$, $w^+_j = v^+_j$, $w^+_{j+1} = v^+_{j+1}$, $N^1_j = M^1_j$, and $N^2_j = M^2_j$;
\item for each terminal face $F_j \in \mathcal{F}_{\noose}$ that is in $\mathcal{F}_{\noosetwo}$ (say it is also numbered $j$ in $\mathcal{F}_{\noosetwo}$) but not in $\mathcal{F}_{\nooseone}$, the state for this face is the same in $B$ and $B_1$. Formally, $I_j = K_j$, $w^+_j = u^+_j$, $w^+_{j+1} = u^+_{j+1}$, $N^1_j = L^1_j$, and $N^2_j = L^2_j$;
\item for each terminal face $F_j \in \mathcal{F}_{\noose}$ that is in both $\mathcal{F}_{\nooseone}$ and $\mathcal{F}_{\noosetwo}$ (say it is also numbered $j$ in both $\mathcal{F}_{\nooseone}$ and $\mathcal{F}_{\noosetwo}$), the state for this face in $B$ is the `union' of the states stored in $B_1$ and $B_2$. Formally, $I_j = J_j \uplus K_j$, $w^+_j = v^+_j$, $v^+_{j+1} = u^+_j$, $w^+_{j+1} = u^+_{j+1}$, $N^1_j = M^1_j$, and $N^2_j = L^2_j$;
\item for each terminal face $F_j$ that is in both $\mathcal{F}_{\nooseone}$ and $\mathcal{F}_{\noosetwo}$ (say it is numbered $j$ in both $\mathcal{F}_{\nooseone}$ and $\mathcal{F}_{\noosetwo}$) but not in $\mathcal{F}_{\noose}$, the state for this face in $B_1$ and $B_2$ jointly covers all terminals. Formally, $J_j \uplus K_j = T^+_j$, $v^+_{j+1} = u^+_j$, and $v^+_j = u^+_{j+1}$.
\end{itemize}
To simplify the presentation, we did not consider index overflow, \eg, when $j+1 > |\mathcal{F}_{\noose}|$, we should use index~$1$ instead. The following is immediate from the definition above and that of optimal partial binders.

\begin{proposition}\label{prp:partial-binders:match}
The triple of optimal partial binders for $(x,y)$, $(y_1,x)$, and $(y_2,x)$ match.
\end{proposition}

The algorithm now does the following. For each valid partial binder $B$ for $(x,y)$, we enumerate all matching, valid partial binders $B_1,B_2$ for $(y_1,x)$ and $(y_2,x)$ respectively and set $Z$ to be equal to the minimum-weight union of the corresponding partial stretchers stored in $Z$ for $B_1$ and $B_2$ that form a partial stretcher for $B$. If no such union forms a partial stretcher for $B$, then we set $Z$ to be equal to the set of all edges of the augmented dual.

\begin{proposition}\label{prp:partial-binders:inductive}
Consider the triple of optimal partial binders $B,B_1,B_2$ for $(x,y)$, $(y_1,x)$, and $(y_2,x)$ respectively. The union of any partial stretcher for $B_1$ and any partial stretcher for $B_2$ forms a partial stretcher for $B$. Moreover, the table entry $Z(B)$ will store a partial stretcher of minimum weight.
\end{proposition}
\begin{proof}
Consider partial stretchers $P_1$ and $P_2$ for $B_1$ and $B_2$ respectively. It is important to remember the fact that the set of edges (bones) in $\enc(y_1,x)$ and in $\enc(y_2,x)$ are disjoint by the definition of a sphere-cut decomposition. Property~(VI) and~(VII) now follow using the fact that each vertex in $\midset(x,y)$ is in $\midset(y_1,x)$ or $\midset(y_2,x)$ or matched in the partial binders. It is also clear that $P_1 \cup P_2$ is a union of a set of splints, almost verifying Property~(V). Moreover, for each face in $\mathcal{F}_{\enc(\nooseone)}$ or $\mathcal{F}_{\enc{(\noosetwo)}}$, Property~(IV) is satisfied. We now verify the remaining (parts of the) properties.

By Proposition~\ref{prp:partial-binders:match}, the partial binders match. We then consider each of the cases:
\begin{itemize}
\item for each terminal face $F_j \in \mathcal{F}_{\noose}$ that is in $\mathcal{F}_{\nooseone}$ (say it is also numbered $j$ in $\mathcal{F}_{\nooseone}$) but not in $\mathcal{F}_{\noosetwo}$, $P_1$ contains edges of the augmented dual to satisfy Property~(II),~(III), and~(V) with respect to $F_j$;
\item for each terminal face $F_j \in \mathcal{F}_{\noose}$ that is in $\mathcal{F}_{\noosetwo}$ (say it is also numbered $j$ in $\mathcal{F}_{\noosetwo}$) but not in $\mathcal{F}_{\nooseone}$, $P_2$ contains edges of the augmented dual to satisfy Property~(II),~(III), and~(V) with respect to $F_j$;
\item for each terminal face $F_j \in \mathcal{F}_{\noose}$ that is in both $\mathcal{F}_{\nooseone}$ and $\mathcal{F}_{\noosetwo}$ (say it is also numbered $j$ in both $\mathcal{F}_{\nooseone}$ and $\mathcal{F}_{\noosetwo}$), $P_1$ and $P_2$ jointly contain edges of the augmented dual to satisfy Property~(II),~(III), and~(V) with respect to $F_j$. The only exception is the terminal inbetween $J_j$ and $K_j$. If such a terminal $t$ indeed exists, then follow the nerve paths stored in the splints of $P_1$ for the bones of face $f_j$ in $\enc(\nooseone)$ from the root of nerve $M^2_j$ (which must be non-empty since $t$ exists) to $v^+_{j+1}=u^+_{j}$, and then the nerve paths stored in the splints of $P_2$ for the bones of face $f_j$ in $\enc(\noosetwo)$ from $v^+_{j+1}=u^+_{j}$ to the root of nerve $L^1_j$ (which must be non-empty since $t$ exists). This path has the same homotopy as in the optimum, by the optimality of the topology and the splints. It also goes between the same vertices as in the optimum, by the optimality of the partial binders. Together with the sides of $M^2_j$ and $L^1_j$, which are as in the optimum, this path thus yields a region that encloses $t$ and no other terminals by Lemma~\ref{lemma:shortest homotopic path}.
\item for each terminal face $F_j$ that is in both $\mathcal{F}_{\nooseone}$ and $\mathcal{F}_{\noosetwo}$ (say it is numbered $j$ in both $\mathcal{F}_{\nooseone}$ and $\mathcal{F}_{\noosetwo}$) but not in $\mathcal{F}_{\noose}$, we note it is contained in $\mathcal{F}_{\enc{(\noose)}}$, but not in $\mathcal{F}_{\enc{(\nooseone)}}$ or $\mathcal{F}_{\enc{(\noosetwo)}}$. To verify Property~(IV), we note that it holds for all terminals in $T_j$ by Property~(III) applied to $P_1$ and $P_2$, except for the one or two terminals inbetween $J_j$ and $K_j$. Let $t$ be such a terminal. Then follow the nerve paths stored in the splints of $P_1$ for the bones of face $f_j$ in $\enc(\nooseone)$ from the root of nerve $M^1_j$ (or from $L^1_j$ if the nerve is empty, or from $v^+_{j+1}$ if both are empty) to $v^+_j=u^+_{j+1}$, and then the nerve paths stored in the splints of $P_2$ for the bones of face $f_j$ in $\enc(\noosetwo)$ from $v^+_j=u^+_{j+1}$ to the root of nerve $L^2_j$ (or to $M^2_j$ if this nerve is empty or to $u^+_{j}$ if both nerves are empty). This path has the same homotopy as in the optimum, by the optimality of the topology and the splints. It also goes between the same vertices as in the optimum, by the optimality of the partial binders. Together with the sides of $M^1_j$ (or of $L^1_j$) and $L^2_j$ (or of $M^2_j$), which are as in the optimum, this path thus yields a region that encloses $t$ and no other terminals by Lemma~\ref{lemma:shortest homotopic path}.
\end{itemize}
It follows that the union of $P_1$ and $P_2$ is a partial stretcher for $B$.

For the final part, we note that by the preceding, the table entry is indeed a partial stretcher. It remains to argue that it has minimum weight. It suffices to argue that an optimum partial stretcher for $B$ can be decomposed into (disjoint) partial stretchers for $B_1$ and $B_2$; to this end, we note that the optimum partial stretcher for $B$ is decomposed into the optimum partial stretchers for $B_1$ and $B_2$ by definition of optimality and matching. Then, it follows by induction and Proposition~\ref{prp:partial-binders:base} that $Z$ stores a minimum-weight partial stretcher for $B_1$ and $B_2$, and thus $Z$ will store a minimum-weight partial stretcher for $B$. 
\end{proof}

Finally, we observe that for the root of the sphere-cut decomposition, any partial binder is equal to the given binder and any partial stretcher for the partial binder is a stretcher for the given binder. This is immediate from the definitions. Thus, we can set $A$ for the given binder as the entry stored in $Z$ for the corresponding (equal) partial binder.

\begin{proof}[Proof of Lemma~\ref{lem:binders:dp}]
By Theorem~\ref{thm:sc}, $\mathcal{B}(l)$ has a sphere-cut decomposition of width $\OO(\sqrt{|\mathcal{B}(l)|})$. We then compute the table $Z$. By Proposition~\ref{prp:partial-binders:number}, we can enumerate all partial binders in $n^{\OO(\sqrt{|\mathcal{B}(l)|})}$ time. We can immediately see that the dynamic program takes $n^{\OO(\sqrt{|\mathcal{B}(l)|})}$ time as well. The base case relies on the splinting algorithm, which takes polynomial time by Lemma~\ref{lem:splinting-algo}. As observed previously, for the root of the sphere-cut decomposition, any partial binder is equal to the given binder and any partial stretcher for the partial binder is a stretcher for the given binder. Hence, the correctness of the algorithm follows from Proposition~\ref{prp:partial-binders:base} and Proposition~\ref{prp:partial-binders:inductive}.
\end{proof}

\section{Proof of Theorem~\ref{thm:main}}
\rev{We are now ready to prove Theorem~\ref{thm:main}}

\thmmain*
\begin{proof}
First, in time $2^{\OO(k)} n^{\OO(1)}$ through the algorithm of Bienstock and Monma~\cite{BienstockM88}, we compute a set of faces of size $k$ that covers all the terminals. This is the set $\mathcal{F}$. We then \rev{change} the instance \rev{as discussed} Section~\ref{sec:connected}, and particularly Lemma~\ref{lem:transformnew} \rev{to make the instance transformed}. We then apply the algorithm of Lemma~\ref{lem:reduce-to-connected} to reduce to the case when the dual of any optimum solution is connected.

Now, by Lemma~\ref{lem:topology-enumerate}, we can enumerate all topologies in $2^{\OO(k^2 \log k)}$ time. For each topology $(S,s,h)$, build an embedding-aware bridge block tree $L = L(S)$ in linear time by Lemma~\ref{lem:bridge-block:build}. Now we perform the dynamic program of Lemma~\ref{lem:binders:dp} and~\ref{lem:binders:partial-dp} on $L(S)$. Since the nodes corresponding to the same cut vertex of $H$ effectively form a subtree of $L$ through Lemma~\ref{lem:bridge-block:cnode}, the binders contain the same cut vertex $w$ and $w^+$ throughout. Then it follows that by induction on the depth of $l$ in $L$ that a minimum-weight stretcher is computed for each optimal binder. Finally, for the root node of $L$, we consider the minimum-weight stretcher that is found among all binders. One of those will be the optimal binder, and thus we find a stretcher of weight at most that of the optimal stretcher for the optimal binder. Recalling that a topology has $O(k)$ vertices, it follows from Lemma~\ref{lem:binders:dp} and~\ref{lem:binders:partial-dp} that the running time is $n^{\OO(\sqrt{k})}$.

Finally, we note here that the running-time our algorithm does not depend on the sum of edge-weights $W$ in a unit cost model of computation. It only linearly depends on $\log W$.
\end{proof}

\bibliographystyle{ACM-Reference-Format}
\bibliography{references}
\end{document}